\documentclass[9pt]{extarticle}
\usepackage{geometry}

 \renewcommand{\paragraph}[1]{\subsubsection*{{#1}}}

\usepackage{amsthm}\newtheorem{definition}{Definition}\newtheorem{theorem}{Theorem}\newtheorem{lemma}{Lemma}\newtheorem{corollary}{Corollary}\newtheorem{proposition}{Proposition}

\newtheorem{preexperiment}{Experiment}
\newenvironment{experiment}{\begin{preexperiment}\rm}{\hfill\qed\end{preexperiment}}

\usepackage[cmex10]{amsmath} %
\usepackage{amsfonts,amssymb} %
\usepackage{ellipsis} %
\usepackage{url} %
\usepackage{epsfig,color} %
\usepackage{stmaryrd} %

\usepackage{cite}

\usepackage{booktabs} %
\newenvironment{tab}[1]
{\begin{center}
\let\oldarraystretch=\arraystretch
\renewcommand{\arraystretch}{1.2} %
\begin{tabular}{@{}#1@{}}
\toprule
}
{\bottomrule
\end{tabular}
\renewcommand{\arraystretch}{\oldarraystretch}
\end{center}}

\newcommand{\avg}{\mathsf{avg}}
\newcommand{\simm}{\mathsf{sim}}
\newcommand{\hit}{\mathsf{hit}}
\newcommand{\percent}{\mathsf{prc}}

\newcommand{\pt}{\mathsf{pt}}
\newcommand{\cnt}{\mathsf{count}}

\newcommand{\dist}[1]{\Delta({#1})}
\newcommand{\degen}{\delta}
\newcommand{\cross}{\times}
\newcommand{\pa}{\mathsf{par}}
\newcommand{\doo}{\mathsf{do}}
\newcommand{\msf}{\mathsf}

\newcommand*{\cons}{{\cdot}}

\newcommand*{\filter}[2]{\lfloor{#1}{\downarrow}{#2}\rfloor}

\newcommand{\posres}{\text{\tiny\texttt{+}}}
\newcommand{\negres}{\text{\tiny\texttt{-}}}

\newcommand*{\email}[1]{\url{#1}}

\newcommand*{\given}{\mathop{{|}}}

\newcommand{\lnmap}{\ln^{\!*}}
\DeclareMathOperator{\coss}{coss}

\DeclareMathOperator{\st}{\,s.t.\,}

\author{Michael Carl Tschantz\\
\email{mct@berkeley.edu}\\
UC Berkeley
\and Amit Datta\\
\email{amitdatta@cmu.edu}\\
Carnegie Mellon University
\and Anupam Datta\\
\email{danupam@cmu.edu}\\
Carnegie Mellon University
\and Jeannette M. Wing\\
\email{wing@microsoft.com}\\
Microsoft Research}

\title{A Methodology for Information Flow Experiments\thanks{The first three sections of this technical report have significant overlap with a previous technical report~\cite{tschantz13tr}.
This research was supported by the U.S.\ Army Research Office grants DAAD19-02-1-0389 and W911NF-09-1-0273 to CyLab, by the National Science Foundation (NSF) grants CCF0424422 and CNS1064688, and by the U.S.\ Department of Health and Human Services grant HHS 90TR0003/01. The views and conclusions contained in this document are those of the authors and should not be interpreted as representing the official policies, either expressed or implied, of any sponsoring institution, the U.S.\ government or any other entity.}}

\begin{document}
\maketitle

\begin{abstract}
Information flow analysis has largely ignored the setting where the analyst has neither control over nor a complete model of the analyzed system.  We formalize such limited information flow analyses and study an instance of it: detecting the usage of data by websites.  We prove that these problems are ones of causal inference.  Leveraging this connection, we push beyond traditional information flow analysis to provide a systematic methodology based on experimental science and statistical analysis.  Our methodology allows us to systematize prior works in the area viewing them as instances of a general approach.  Our systematic study leads to practical advice for improving work on detecting data usage, a previously unformalized area.  We illustrate these concepts with a series of experiments collecting data on the use of information by websites, which we statistically analyze.
\end{abstract}

\section{Introduction}
\label{sec:intro}

\paragraph{Web Data Usage Detection}
Concerns about privacy have led to much interest in determining how third-party associates of first-party websites use information they collect about the visitors to the first-party website.  Mayer and Mitchell provide a recent presentation of research that tries to determine what information these third-parties collect~\cite{mayer12sp}.  Others have attempted to determine what these third-parties \emph{do} with the information they collect~\cite{guha10imc,wills12wpes,balebako12w2sp,sweeney13cacm}.  We call this problem \emph{web data usage detection} (WDUD).

The researchers involved in WDUD each propose and use various analyses to determine what information is tracked and how it is used.  They primarily design their analyses by intuition and do not formally present or study their analyses.  Thus, questions remain:
\begin{enumerate}
\item Are the analyses used correct?
\item Are they related to more formal prior work?
\end{enumerate}

To answer these questions, we must start with a formal framework that can express the problem and the analyses.  In essence, each of these works is conducting an information flow analysis: the researchers want to know when information flows to a third-party and where it goes from there.  Thus, the natural starting point for such a formalism is prior research on \emph{information flow analysis} (IFA).
However, despite the great deal of research on IFA (see \cite{sabelfeld03journal} for a survey), we know of no attempt to relate or inform WDUD research with the models or techniques of IFA, even in an informal manner.

We believe this disconnect exists for an important reason: the traditional motivation for IFA, designing secure programs, pushes it away from analyzing third-party systems as done in WDUD.
Typically, the analyst is seen as verifying that a system under his control protects information sensitive to the system.  Thus, the problems studied and analyses proposed tend to presume that the analyst has access to the program running the system in question.  

In WDUD, the analyzed system can be adversarial with the analyst aligned with a \emph{data subject} whose information is collected by the system.  In this setting, the analyst has no access to the program running the third-party service, little control over its inputs, and a limited view of its behavior.  Thus, the analyst does not have the information presupposed by traditional IFAs.  
To understand the WDUD problem as an instance of IFA requires a fresh perspective on IFA.

\paragraph{Other Atypical IFA Problems}
The implicit assumptions underlying much of IFA research also obscure its connection to other areas of research. 

For example, the cryptography community has much work on identifying illicit flows of files.
Such work has included \emph{watermarking}~\cite{wagner83sp,swanson98ieee}, in which a key that links to the identity of the person to whom the publisher sold the copy is embedded in the work.  \emph{Traitor tracing} is the special case of determining who illicitly provided cryptographic keys to enable decrypting data~\cite{chor94crypto}.  

Closely related is the detection of plagiarism.  One approach the publisher can use for this problem is to employ a \emph{copyright trap}: deliberately unusual (typically, false) information inserted into reference works to detect copying.  For example, a map might include a \emph{trap street} that is purposely misplaced and/or misnamed~\cite{monmonier96maps}.  If another publisher mechanically copies the map, the inclusion of the trap street in the copy will indicate the copying.  

Organizations handling sensitive data are concerned about data misuse.
For example, governments are concerned with employees leaking classified documents to reporters or foreign spies.  
For ethical reasons and to comply with regulations, such as the HIPAA Privacy Rule~\cite{hipaa}, healthcare providers limit the use of personal health information.  
Thus, organizations have adopted a variety of methods to discourage the misuse of such data by their employees~\cite{symantec,rsa}.  For example, investigators have employed \emph{Barium meals}, 
a watermarking-like analysis~\cite{wright87spycatcher}.  To use a Barium meal, the investigators feed different versions of classified information to each suspect leaker.  While the investigators cannot see what each suspect does with this information directly, they may be able to infer the identity of the leaker based upon newspaper accounts of the leaked information.  As another example, a company can distribute email lists to business partners with varying fake addresses, or \emph{honeytokens}~\cite{spitzner10symantec} or with varying subsets of the data~\cite{papadimitriou11data}.

In essence, these works are all IFAs.   In particular, the analyst, who is aligned with the copyright holder or organization, would like to determine whether a system (typically a personal computer or person) is enabling an illicit flow of information.  However, those working on these problems have not typically discussed them as such since they do not fit into the traditional IFA setting.  In particular, the analyst has little if any access or control over the analyzed system.  Like with WDUD, the analyst must investigate an uncontrolled black box.  Indeed, we find that some of the intuitive approaches used in WDUD are related to cryptographic measures used in piracy detection.

\paragraph{Goal} 
Our goal is to systematize the information flow problems and analyses common to these areas of research.  
To do so, we identify the limited abilities of the analyst in these problems. 
as a form of analysis between the extremes of white box program analysis and black box monitoring.  We show that the ability of the analyst to control some inputs during an investigation enables \emph{information flow experiments} that manipulate the system in question to discover its use of information without a white box model of the system.  Our  framework provides a fresh perspective both on our diverse set of motivating applications and on IFA by allowing us to elucidate and challenge approaches in these areas and in IFA. 

The overarching contribution of this work is relating IFA in these nontraditional settings to experiments designed to determine causation.  To do so, we prove a connection between information flow and causality, which allows us to reduce these problems to well understood empirical ones.  In particular, it allows us to use statistical analyses in the place of traditional methods of IFA, such as program analysis.

\paragraph{Overview}
We start with a closer examination of our motivating applications of WDUD in Section~\ref{sec:areas}.  We focus on WDUD as the least understood of the motivating problems.

We then discuss IFA in general and the limitations of traditional IFA in Section~\ref{sec:ifa}.  
We abstract over particular problems to systematize a class of IFA that has gone unformalized.  We shift IFA from its traditional context of program analysis using white box models of software to the new context of investigating black box systems that hide much of their behavior and operate in uncontrolled environments.  This work systematizes the common but hitherto independent efforts of our motivating applications by unifying them under one framework. 

In particular, we formalize these problems in terms of a version of noninterference, the primary definition of traditional IFA~\cite{gm82security}, giving the first systematic expression of the WDUD problem.
We prove that sound information flow detection is impossible in this setting (Theorem~\ref{thm:unsound-inter}).

Motivated by the impossibility result, we look for an alternative statistical approach.  Fortunately, IFA is related to causality, a much studied concept for which statistical analyses already exist.  In Section~\ref{sec:cause}, we prove that a system has interference from a high-level user $H$ to low-level user $L$ in the sense of IFA if and only if inputs of $H$ can have a causal effect on the outputs of $L$ while the other inputs to the system remain fixed (Theorem~\ref{thm:inter-cause}).  This connection allows us to appeal to inductive methods employed in experimental science to study IFA.  
Such methods provide precisely what we need in the face of our unsoundness results to make high-assurance statistical claims about flows.

We leverage this observation to approach WDUD with information flow experiments.
 Section~\ref{sec:doe} discusses how to conduct such experiments.
While many of the issues discussed are well known to scientists, we must adapt them to our setting.
We show a correspondence between the features of WDUD and the requirements of a scientific study (Table~\ref{tbl:compare}).
 We pay particular attention to general principles that should guide the design of information flow experiments rather than attempting to provide a cookbook approach, which often leads to misapplication~\cite{ludbrook08ije}.

Section~\ref{sec:stats} reviews significance testing as a systematic method of quantifying the degree of certainty that an information flow experiment has observed interference.  In particular, we focus on permutation testing~\cite{good05book}, a method of significance testing that we find particularly well suited to the setting of WDUD in which we have little knowledge of the web tracker's internal behavior.

Section~\ref{sec:stings} provides a systematic look at prior works in WDUD.  We analyze each of them under the unifying method of permutation testing.  This unification allows us to compare and contrast their disparate, and often ad hoc, methods.  We find the strengths and weaknesses of their experimental designs.  We also empirically benchmark our interpretations of their approaches with our own WDUD study, which we believe to be the first to come with an analysis of correctness (Section~\ref{sec:pos}).

We end by discussing future work.  We first provide practical suggestions, which are summarized in Section~\ref{sec:sugg-conc}, for conducting future WDUD studies in a systematic fashion.
We then discuss directions for new research that apply the connection between information flow and causality to other security problems.

\paragraph{Contributions}
Our methodology is supported by a chain of contributions that follows the paper's outline:

\newcommand*{\itmnum}[1]{Section~{#1}}
\newcommand*{\egap}{1ex}
\begin{center}
\begin{tabular}{ll}
\itmnum{\ref{sec:ifa}} & a systematization of nontraditional IFA\\
\itmnum{\ref{sec:cause}} & a proof of a connection between IFA and causality\\
\itmnum{\ref{sec:doe}} & an experimental design leveraging this connection\\
\itmnum{\ref{sec:stats}} & a statistical approach to analyzing experimental data\\
\itmnum{\ref{sec:stings}} & a systematization prior studies under a unifying method\\
\end{tabular}
\end{center}

These contributions are each necessary for creating a chain of sound reasoning from intuition about vague problems to rigorous quantified results in a formal model.  This chain of reasoning provides a systematic, unifying, view of these problems, which leads to a concrete methodology based on well studied scientific methods.
While the notion of experimental science is hardly new, our careful justification provides guidance on the choices involved in actually conducting an information flow experiment.

Throughout this work, we present our own experiments to illustrate the abstract concepts we present.
These results may also be of independent interest to the reader.
  To keep the presentation clear, we focus on only WDUD.

In addition to containing details of experiments and results, the appendices found at the end of this document also contain formal models and the proof of each of our theorems.  We make the code used to run our experiments and the data collected available at:\\
\centerline{\url{http://www.cs.cmu.edu/~mtschant/ife/}}

The systematization of experimental approaches to security is becoming increasingly important as technology trends (e.g., Cloud and Web services) result in analysts having limited access to 
and control over systems whose properties they are expected to study. This paper provides a useful starting point towards such a systematization by 
providing a common model and a shared vocabulary of concepts that ties together seemingly disparate areas of security and privacy by placing them in the context of causality, experimentation, and statistical analysis.

\paragraph{Prior Work}
Three of the authors have previously identified the need to formalize the setting of information flow experiments~\cite{tschantz13tr}.  Their prior technical report overlaps significantly with the first three sections of this report.  Their approach did not use standard statistical analyses or experimental designs.  They instead identified assumptions that would categorically justify the implicit reasoning of prior information flow analyses.

Ruthruff, Elbaum, and Rothermel note the usefulness of experiments for program analysis~\cite{ruthruff06issta}. 
 Whereas our work focuses on problems where traditional white box analyses are impossible, their work examines experiments in the more traditional setting where the analyst has control over the system in question.
Furthermore, rather than provide an informal overview of how experiments can be used for program analysis, we develop a formalism relating informal flow and causality, provide proofs, and present a statistical analysis.  

While we could not find any prior articulation of this formal correspondence between informal flow and causality (our Theorem~\ref{thm:inter-cause}), we are not the first to note such a connection. 
McLean~\cite{mclean90sp} and Mowbray~\cite{mowbray92csf} each proposed a definition of information flow that uses the lack of a causal connection to rule out security violations even if there is a flow of information from the point of view of information theory. %
Sewell and Vitek provide a ``causal type system'' for reasoning about information flows in a process calculus~\cite{sewell00csf}.  We differ from these works by showing an equivalence between a standard notion of information flow, noninterference~\cite{gm82security}, and a standard notation of causality, Pearl's~\cite{pearl09book}, rather than using an ad hoc notion of causality to adjust an information theoretic notion of information flow.
 Furthermore, Mowbray's formalism requires white box access to the system while McLean's only considers temporal ordering as a source of causal knowledge.  
More importantly, they use causality to handle problematic edge cases in their formalisms whereas we reduce interference to causality so that we may apply standard methods from experimental science to IFA, which we discuss in the next section.

In Section~\ref{sec:areas}, we discuss prior works on WDUD and we show in Section~\ref{sec:stings} that our methodology can formalize them.  
In Section~\ref{sec:ifa} we discuss in detail prior work on IFA and why it is insufficient for our goal of black box program analysis.  
We draw on works from experimental design and statistical analysis, whose discussion we defer until the point of use.

\section{Web Data Usage Detection}
\label{sec:areas}

Users visiting websites provide vast amounts of information and yet have little understanding of how the website might use the information.  In particular, websites provide little information about how one provided input might affect what the user sees on that page or others.  
A visitor might be unaware of and surprised by the flows of information from one place to the next on the web and how these information flows impact their treatment.  (For a survey, see~\cite{mayer12sp}.)

A first step to understanding these flows is determining what information a website collects (e.g.,~\cite{krishnamurthy11w2sp}).     However, more difficult is detecting the \emph{usage} of such data.  That is, determining how the collected information impacts the treatment of the visitor on that and affiliated sites.
Researchers working on this problem of web data usage detection (WDUD), must infer from interactions over the Internet the unseen flows of information within and between web servers.

Wills and Tatar studied how Google selects ads based on information provided by the website visitor via first-party websites~\cite{wills12wpes}.  The authors draw conclusions about Google's information use in two ways.  First, they observed Google showing them (posing as normal website visitors) ads that included sensitive information they provided to Google by interacting with a website that uses a Google service, such as Ad Sense.  Second, when posing as two different users with different interests, they observed Google showing them ads differing in ways related to the differing interests.

Guha et al.\ study a similar problem using a more statistical approach~\cite{guha10imc}.  Like Wills and Tatar, they would pose as various visitors with different characteristics.  To test whether some change between two user profiles resulted in a change in Google's ads, they would pose as the first profile twice and as the second profile once.  By using the same profile twice, they could calculate the baseline amount of noise or ``ad churn'' in the ads independent of the change.  If the change between the first and second profile is larger than this baseline, they then conclude that the change between profiles caused the increased difference in the ads.  Balebako et al.\ adopt the methodology of Guha et al.\ to study the effectiveness of web privacy tools~\cite{balebako12w2sp}.

While Wills and Tatar look at the differences between ads to determine whether they have anything to do with sensitive information, Guha et al.\ do not attempt to interpret the ads to see what could have caused the change.  (They did look at the ads while validating their analysis.)  While the analysis of Wills and Tatar leads to a better understanding of how the website is using the information, the analysis of Guha et al.\ can find changes that people are apt to miss since the relationship between the changes in input and output are not immediately clear or because they take a larger sampling to notice than is possible with manual inspection.  For example, they find that a profile purportedly of a homosexual male gets a large increase in nursing school ads, which may have been missed by the Wills and Tatar's method since there is no clear connection between the change in the profile to the change in the ads.  As Guha et al.\ point out, this lack of connection makes this discovery more important since the website visitor would also be unlikely to realize that responding to the nursing ad could leak sensitive information to the nursing program.

Recently, Sweeney conducted an information flow experiment in which she examined the flow of information from a search field to ads shown along side the search results~\cite{sweeney13cacm}.  She found that searching for characteristically black names yielded more ads for InstantCheckmate featuring the word ``arrested'' than searching for characteristically white names.  She found this result on both the websites of Google and of Reuters.  Unlike the preceding studies, she used a statistical test, the $\chi^2$ test, to analyze her results and found them to be significant.  

After introducing the machinery necessary to so do, we will systematically analyze each of these information flow experiments in Section~\ref{sec:stings}.  We will examine each of them in relation to the permutation test, which will allow us to discuss their strengths and make suggestions for improvements.

\section{Information Flow Analysis}
\label{sec:ifa}

In this section, we discuss prior work on information flow analysis starting with noninterference, a formalization of information flows.  We next discuss the analyses used in prior work to determine whether a flow of information exists.  We present them systematically by the capabilities they require of the analyst.  We end by discussing the capabilities of the analyst in our motivating applications, how prior analyses are inappropriate given these capabilities, and the inherent limitations of these capabilities.

\subsection{Noninterference}

Goguen and Meseguer introduced \emph{noninterference} to formalize when a sensitive input to a system with multiple users is protected from untrusted users of that system~\cite{gm82security}.  Intuitively, noninterference requires that the system behaves identically from the perspective of untrusted users regardless of any sensitive inputs to the system.  

As did they, we will define noninterference in terms of a synchronous finite-state Moore machine.
The inputs that the system accepts are tuples where each component represents the input received on a different input channel.  
Similarly, our outputs are tuples representing the output sent on each output channel.
For simplicity, we will assume that the machine has only two input channels and two output channels, but all results generalize to any finite number of channels.  

We partition the four channels into $H$ and $L$ with each containing one input and one output channel.  
Typically, $H$ corresponds to all channels to and from high-level users, and $L$ to all channels to or from low-level users.  
The high-level information might be private or sensitive information that should not be mixed with public information, denoted by $L$.  In the area of taint analysis, the roles are reversed in that the tainted information is untrusted and should not be mixed with trusted information on the trusted channel.  However, either way, the goal is the same: keep information on the input channel of $H$ from reaching the output channel of $L$.

We will often have a single user using channels of both sets since we are concerned with not only to whom information flows but also under what contexts.  To this end, we interpret \emph{channel} rather broadly to include virtual channels created by multiplexing, such as a field of an HTML form or the ad container of a web page.  We also allow for each channel's input/output to be a null message indicating no new input/output.

A system $q$ consumes a sequence $\vec{\imath}$ of input pairs where each pair contains an input for the high and the low input channels.   We write $q(\vec{\imath})$ for the output sequence $\vec{o}$ that $q$ would produce upon receiving $\vec{\imath}$ as input where output sequences are defined as a sequence of pairs of high and low outputs.

For an input sequence $\vec{\imath}$, let $\filter{\vec{\imath}}{L}$ denote the sequence of low-level inputs that results from removing the high-level inputs from each pair of $\vec{\imath}$.
That is, it ``purges'' all high-level inputs.  
We define $\filter{\vec{o}}{L}$ similarly for output sequences.

\begin{definition}[Noninterference]\label{def:noninter}
A system $q$ has \emph{noninterference} from $L$ to $H$ iff for all input sequences $\vec{\imath}_1$ and $\vec{\imath}_2$, 
\[ \filter{\vec{\imath}_1}{L} = \filter{\vec{\imath}_2}{L} \text{ implies } \filter{q(\vec{\imath}_1)}{L} = \filter{q(\vec{\imath}_2)}{L} \]
\end{definition}

Intuitively, if inputs only differ in high-level inputs, then the system will provide the same low-level outputs.

To handle systems with probabilistic transitions, we will employ a probabilistic version of noninterference similar to the previously defined \emph{probabilistic nondeduciblity on strategies}~\cite{gray91sp}.  To define it, we let $Q(\vec{\imath})$ denote a probability distribution over output sequences given the input $\vec{\imath}$, a concept that can be made formal given the probabilistic transitions of the machine~\cite{gray91sp}.  We define $\filter{Q(\vec{\imath})}{L}$ to be the distribution $\mu$ over sequences $\vec{\ell}$ of low-level outputs such that $\mu(\vec{\ell}) = \sum_{\vec{o} \st \filter{\vec{o}}{L} = \vec{\ell}} Q(\vec{\imath})(\vec{o})$.

\begin{definition}[Probabilistic Noninterference]
A system $Q$ has \emph{probabilistic noninterference} from $L$ to $H$ iff for all input sequences $\vec{\imath}_1$ and $\vec{\imath}_2$, 
\[ \filter{\vec{\imath}_1}{L} = \filter{\vec{\imath}_2}{L} \text{ implies } \filter{Q(\vec{\imath}_1)}{L} = \filter{Q(\vec{\imath}_2)}{L} \]
\end{definition}

\subsection{Analysis}

Information flow analysis (IFA) is a set of techniques to determine whether a system has noninterference (or similar properties) for interesting sets $H$ and $L$.  
Proving (non)interference by brute force is difficult for systems with many possible inputs especially when the system, its inputs, or its outputs are out of the control or view of the analyst.  Thus, analysts must employ strategic analyses specialized to his capabilities. 

IFA grew out of the demand to build military computers respecting mandatory access controls (MAC).  Thus, much of the work in the area presumes that the analyst has a degree of control over the production of the analyzed system.
Examples include analyses employing type systems~\cite{volpano96jcs,sabelfeld03journal}, model checking of code~\cite{barthe04csf}, or dynamic approaches that instrument the code running the system to track values carrying sensitive information (e.g.,~\cite{vachharajani04micro,newsome05ndss,venkatakrishnan06icics,mccamant07plas}).

The above methods are inappropriate for WDUD since they require \emph{white box} access to the program.  That is, the analyst must be able to study and/or modify the code.  
In our applications, the analyst must treat the program as a \emph{black box}.  That is, the analyst can only study the I/O behavior of the program and not its internal structure.  Black box analyses vary based on how much access they require to the system in question.  
Figure~\ref{fig:taxonomy} shows a taxonomy of analyses.
\begin{figure}
\begin{center}
\includegraphics{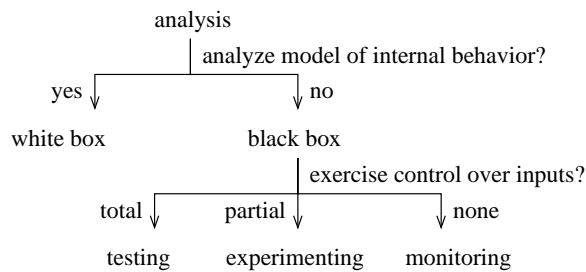}
\end{center}
\caption{Taxonomy of analyses}
\label{fig:taxonomy}
\end{figure}

Numerous black box analyses for detecting information flows exist that operate by running the program multiple times with varying inputs to detect changes in output that imply interference~\cite{yumerefendi07tightlip,guernic07asian,jung08ccs,capizzi08preventing,devriese10noninterference}.
However, these black box analyses continue to require access to the internal structure of the program even if they do not analyze that structure.  For example, the analysis of Yumerefendi et al.\ requires the binary of a program to copy it into a virtual machine for producing I/O traces~\cite{yumerefendi07tightlip}.  In theory, such black box analyses could be modified to not require any access to code by completely controlling the environment in which the program executes.  To do so, the analyst would run a single copy of the program and reset its environment to simulate having multiple copies of the system.  We call this form of black box analysis, with total control over the system, \emph{testing} as it is the setting typical to software testing 
and testing notions of equivalence (e.g.,~\cite{nicola83calp,nicola84tcs}).

Testing will not work for our applications.  For example, in the application of WDUD, the analyst cannot run the program multiple times since the analyst has only limited interactions with the program over a network.  Thus, it cannot force the program into the same initial environment to reset it.  Furthermore, unlike a program, Google's ad \emph{system} is stateful and, thus, modifying its environment alone would be insufficient to reset it.  In this setting, the analyst must analyze the \emph{system} as it runs, not a program whose environment the analyst can change at will.

At the opposite extreme of black box analysis is \emph{monitoring}, which passively observes the execution of a system. 
While some monitors are too powerful by being able to observe the internal state of the running system (e.g.~\cite{schneider00tissec}), others match our needs in that the analyst only has access to a subset of the program's outputs (e.g.,~\cite{garg11ccs}).  
However, all monitors are too weak since they cannot provide inputs to the system as our application analysts can.  We need a form of black box analysis between the extremes of testing and monitoring.

Thus, we find that no prior work on IFA that corresponds to the capabilities of the analyst in WDUD or other motivating applications.

\subsection{Information Flow Experiments}
\label{sec:ife-impossible}

Unlike the primary motivation of traditional IFA, developing programs with MAC, our motivating examples involve situations in which the analyst and the system in question are not aligned.  Thus, the information available to the analyst is much more limited than in the traditional security setting.
In particular, the analyst
\newcommand{\itemm}{\item} %
\begin{enumerate}
\itemm has no model of or access to the program running the system,
\itemm cannot observe or directly control the internal states of the system,
\itemm has limited control over and knowledge of the environment of the system,
\itemm can observe a subset of the system's outputs, and
\itemm has control over a subset of the inputs to the system.
\end{enumerate}
We will call performing IFA in this setting \emph{experimenting}.  Experiments may be viewed as an interactive extension of a limited form of execution monitoring that allows for analyst inputs but limits the analyst to only observing a subset of system I/O.

Prior work shows that no monitor can detect information flows~\cite{mclean94sp, schneider00tissec, volpano99sas}.  
We argue that experiments, with the additional ability to control some inputs to the system, do not improve upon this situation.  In particular, we prove that no non-degenerate analysis can be sound for interference or for noninterference, even on deterministic systems.

Before presenting the formal theorems, let us intuit why checking for interference would be difficult in this setting.  
To start, let us examine the difficulties in producing a sound (no false positives) method for determining that Google has interference.  That is, we would like a method that upon returning a positive result implies that Google did in fact use some high-level information to select some low-level output.  For example, the high-level information could be a search query to Google and the low-level outputs could be the ads that Google shows at some later point.

Note that first two limitations above forces the analyst to determine interference by examining only the inputs and outputs to the system.  Since this prohibits white box analysis, to conclude interference, she would need to observe input sequences $\vec{\imath}_1$ and $\vec{\imath}_2$ such that $\filter{\vec{\imath}_1}{L} = \filter{\vec{\imath}_2}{L}$ but $\filter{q(\vec{\imath}_1)}{L} \neq \filter{q(\vec{\imath}_2)}{L}$ where $q$ is the system and $L$ is the set of low-level inputs (recall Definition~\ref{def:noninter}).

However, the third limitation prevents the analyst from observing all the inputs to determine that $\filter{\vec{\imath}_1}{L} = \filter{\vec{\imath}_2}{L}$ unless $L$ includes only inputs that the analyst can observe.  Since every input must be either low-level or high-level and only the user's gender is high-level, the low-level inputs include many inputs that the analyst cannot observe such as inputs from advertisers to Google.
(Furthermore, ideally, the analyst would have control over the inputs to ensure that they are equal instead of merely hoping that equality occurs.)

To eliminate the unobservable low-level inputs, the analyst must shrink the set of low-level inputs.  One means of achieving this goal is to consider more inputs high-level.  However, if the inputs converted to be high-level are already known to determine the ads shown (such as inputs from advertisers), then the analysis would be of little interest.  Another means would be to eliminate the inputs from Google, but the analyst does not have such control over Google.  However, the analyst does have control over which system she studies.  Rather than study Google in isolation, she could study the composite system of Google and the advertisers operating in parallel.  By doing so, she converts the unobserved low-level inputs to Google from the advertisers into internal messages of the composite system, which are irrelevant to whether interference occurs.

In some sense we have converted the problem from one of experimenting proper to one more akin to testing the composite system.  However, even with this conversion, the analyst still does not have total control over the system in question (i.e., the composite one) since the analyst still cannot alter the internal structures of the system.  In particular, by the second limitation, the analyst cannot reset the system as analysts commonly do while testing the system's behavior on various input sequences.  Thus, the analyst in our setting cannot actually run two input sequences since doing so changes the internal initial state of the second run; we are not truly in a testing situation.  Furthermore, this limitation results in unsoundness even for the composite system as we show below.

To prove this unsoundness of black box analyses for interference, we consider an arbitrary system $q$ for which an analysis returns a positive result indicating interference.  In our setting, the analysis must base its decision solely upon its interactions with the system.  Thus, it will return the same positive result for a system $q_{\mathrm{N}}$ that always produces the same outputs as $q$ did irrespective of its inputs.  Since $q_{\mathrm{N}}$ always produces these outputs, it has noninterference making the positive result false.

\begin{theorem}\label{thm:unsound-inter}
Any black box analysis that ever returns a positive result from interference for $H$ to $L$ is unsound for interference from $H$ to $L$.
\end{theorem}

The argument for noninterference is symmetric, but requires that interference is possible given the system's input and output space.  That is, the system must have at least two high inputs and two low outputs.

\begin{theorem}\label{thm:unsound-noninter}
Any black box analysis that ever returns a positive result for noninterference from $H$ to $L$ is unsound for noninterference from $H$ to $L$ if $H$ has two inputs and $L$ has two outputs.
\end{theorem}

Note that these theorems hold even if the analyst can observe every input in $H$ and $L$ making the above shift of focus to the composite system of Google operating in its environment unsuccessful.  However, as we will later see, we can probabilistically handle the lack of total internal control of the composite system using statistical techniques.  Since we can never be sure whether we have started a particular sequence of inputs from the same initial state as another sequence, we use many instances of each sequence instead of one for each.  Intuitively, if the outputs for one group of inputs are consistently different from outputs for the other group of inputs, then it is likely that the difference is introduced by the difference between the groups instead from the initial states differing.  We formalize this idea to present a probabilistically sound method of detecting interference.   We leave detecting noninterference to future work.

\section{Causality} 
\label{sec:cause}

In this section, we discuss a formal notion of causality motivated by the studies of the natural sciences.  We then prove that noninterference corresponds to a lack of an effect.  This result allows us to repose WDUD as a problem of statistical inference from experimental data using causal reasoning.

\subsection{Background} 

Let us start with a simple example.  A scientist might like to determine whether a Drug X causes an effect on mouse mortality.
More formally, she is interested in whether the value of the \emph{experimental factor} $X$, recording whether the mouse gets Drug X, causes an effect to a \emph{response variable} $Y$, a measure of mouse mortality, holding all other factors (possible causes) constant.

Pearl~\cite{pearl09book} provides a formalization of \emph{effect} using \emph{structural equation models} (SEMs), a formalism widely used in the sciences (e.g.,~\cite{hoyle12book}).
A probabilistic SEM $M = \langle \mathcal{V}_{\msf{en}}, \mathcal{V}_{\msf{ex}}, \mathcal{E}, \mathcal{P}\rangle$ includes a set of \emph{variables} partitioned into \emph{endogenous} (or dependent) variables $\mathcal{V}_{\msf{en}}$ and \emph{exogenous} (or independent) variables $\mathcal{V}_{\msf{ex}}$.  $M$ also includes in $\mathcal{E}$, for each endogenous variable $V$, a \emph{structural equation} $V := F_V(\vec{V})$ where $\vec{V}$ is a list of other variables not equal to $V$ and $F_V$ is a possibly randomized function.  A structural equation is directional like variable assignments in programming languages. 
Each exogenous variable is defined by a probability distribution given by $\mathcal{P}$.  Thus, every variable is a random variable defined in terms of a probability distribution or a function of them.

Let $M$ be an SEM, $X$ be an endogenous variable of $M$, and $x$ be a value that $X$ can take on. 
Pearl defines the \emph{sub-model} $M[X{:=}x]$ to be the SEM that results from replacing the equation $X := F_X(\vec{V})$ in $\mathcal{E}$ with the equation $X := x$.  
The sub-model $M[X{:=}x]$ shows the \emph{effect} of setting $X$ to $x$.
Let $Y$ be an endogenous variable called the \emph{response variable}.  We define \emph{effect} in a manner similar to Pearl~\cite{pearl09book}.
\begin{definition}[Effect]
The experimental factor $X$ has an \emph{effect} on $Y$ given $Z:=z$ iff there exists $x_1$ and $x_2$ such that the probability distribution of $Y$ in $M[X{:=}x_1][Z{:=}z]$ is not equal to its distribution in $M[X{:=}x_2][Z{:=}z]$.
\end{definition}
Intuitively, there is an effect if $F_Y(x_1, \vec{V}) \neq F_Y(x_2, \vec{V})$ where $\vec{V}$ are the random variables other than $X$.

\subsection{The Relationship of Interference and Causality}

Intuitively, interference is an effect from a high-level input to a low-level output.  Noninterference corresponds to lack of an effect, which Pearl calls \emph{causal irrelevance}~\cite{pearl09book}.

We can make the connection between interference and causality formal by providing a conversion from a probabilistic system to an SEM.  Given a system model $Q$, we define a SEM $M_Q$.  
For each time $t$, $M_Q$ contains the endogenous variables $V_{\mathsf{hi},t}$ and $V_{\mathsf{li},t}$ for the high and low input, and $V_{\mathsf{lo},t}$ for the low output at the time $t$.
The behavior of $Q$ provides functions $F_{\mathsf{lo},t}$ defining the low output at time $t$ in terms of the previous and current inputs, which can be saved to a variable representing state.  (Details may be found in Appendix~\ref{app:connection}.)

To state the theorem, we use $\vec{V}_{\msf{lo}}^{t}$ to denote a vector of low-output response variables ranging in time from $1$ to $t$ and $\vec{V}_{\msf{i}}^{t}$ to represent a similar vector of input factors combining $V_{\mathsf{hi},t}$ and $V_{\mathsf{li},t}$.

\begin{theorem}\label{thm:inter-cause}
$Q$ has probabilistic interference iff there exists low inputs $\ell$ of length $t$ such that $\vec{V}_{\mathsf{hi}}^{t}$ has an effect on $\vec{V}_{\msf{lo}}^{t}$ given $V_{\mathsf{li}}^{t} := \ell$.  
\end{theorem}

Notice that Theorem~\ref{thm:inter-cause} requires that the low-level inputs to the system in question be fixed to a set value $\ell$.  Thus, the experimenter must ensure that the entire sequence of low-level inputs is equal to $\ell$, recalling the issue of having a lack of total control over inputs discussed in Section~\ref{sec:ife-impossible}.  As discussed, our solution is to consider the system operating in its environment allowing us to include these inputs as internal to the composite system.

In the case of studying Google, the impact of considering this composite system, rather than Google proper, is that finding an effect for Google while experimenting with Google might not imply interference within Google proper, but rather interference in the composite system.  An example of such interference would be Google passing a high-level input to an advertiser that alters its low-level inputs to Google resulting in a change in Google's output to the experimenter (Figure~\ref{fig:compose}).
\begin{figure}
\begin{center}
\includegraphics{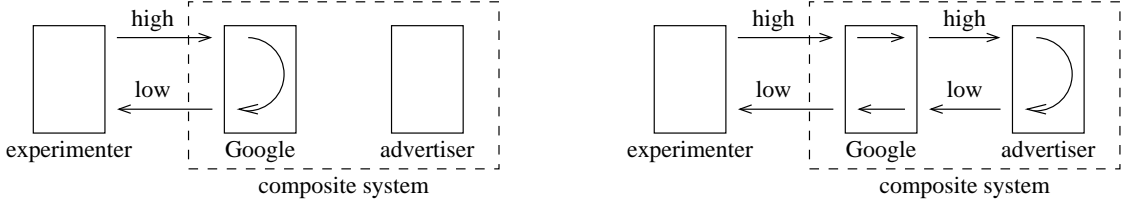}
\end{center}
\caption{The left shows a flow of information within Google proper that implies interference.  The right shows a flow of information that implies interference for the composite system consisting of Google and the advertiser but not for Google considered in isolation.}
\label{fig:compose}
\end{figure}
That Google operating in its environment can have interference while Google considered in isolation does not is related to noninterference not being preserved under composition (e.g.,~\cite{mclean94sp}).  

\section{Experimentation}
\label{sec:doe}

To understand the role of experimentation in determining causal relations, we start by returning to the mouse study and continuing it in a manner suggested by an epidemiology methods paper~\cite{greenland86epidemiology}.  We then discuss the design of experiments in general. Section~\ref{sec:ife} applies these general principles to information flow experiments to justify a particular methodology.  In particular, we present a correspondence between well known features of experimental science to less familiar features of information flow experiments, which we summarize in Table~\ref{tbl:compare}.  Lastly, we comment on a few secondary concerns.

\subsection{Example}

The scientist would like to learn whether for some mouse $k$ there is an effect of Drug X on the mouse's ability to survive for a week.  That is, whether there exists some conditions $\vec{z}_k$ such that $F_k(0, \vec{z}_k, \vec{U}_k) \neq F_k(1, \vec{z}_k, \vec{U}_k)$ where we use $X_k = 1$ to denote the mouse $k$ getting treated with Drug X and $X_k = 0$ for not getting treated,  
and where the range of $F_k$ is $1$ for dying and $0$ for living a week.  
For simplicity, let us assume that $F_k$ is deterministic given $X_k$, and that it does not depend upon $\vec{Z}_k$ and $\vec{U}_k$ (which we drop).

Even with these simplifications, the scientist's task is difficult.  
If $F_k$ were known to the scientist, she could compare its calculated value at $0$ and $1$, which is similar to white box program analysis.  However, $F_k$ is unknown and the scientist can only observe its output once: either $F_k(0)$ or $F_k(1)$ since each mouse can only be treated or not.

If the scientist could get two mice $k$ and $j$ such that $F_k = F_j$,
then she could check whether $F_{k}(0) \neq F_{j}(1)$.  If so, she can infer an effect of $X_k$ to $Y_k$ and $X_j$ to $Y_j$.
Requiring that $F_k = F_j$ does not require the mice to be identical, just that they react in the same manner to Drug X as far as living for a week is concerned.  In fact, there are only four functions $F_k$ and $F_j$ could be: the constant $0$ function (always live), the constant $1$ function (always die), the identity function (die iff treated), and the ``negation'' function $1 - X$ (live iff treated).  

To leverage this observation, the scientist gets a large number of mice and splits them randomly into two groups of equal size.  She then gives only the mice in the first group, the experimental group, Drug X and treats the mice in the second group, the control group, otherwise identically.  To make the example extreme, suppose she then observes that every mouse treated died and every mouse not treated lived.  These results could be explained by the experimental group consisting solely of mice that are characterized by the constant $1$ function (always die) and all the mice in the control group being characterized by the constant $0$ function (always live).  However, to randomly assign mice in such a fashion is extremely unlikely even if the population of mice consist of only those functions in a 50\%/50\% split.  Rather, such results suggest that at least one mouse (and probably almost all) are characterized by the identity function since such a population makes the result much more likely.  Thus, the scientist concludes that there exists at least one mouse $k$ such that $X_k$ has an effect on $Y_k$.

\subsection{Experimental Design}

This reasoning can be extended to the case where $X_k$ and $F_k$ take on more than two values and $F_k$ depends upon $\vec{Z}_k$ and $\vec{U}_k$ in a randomized fashion.  In general, the scientist takes a sample of \emph{experimental units} (e.g., mice), the number of which is the \emph{sample size}.  She also prepares a vector $\vec{x}$ with a length equal to the sample that hold values, called \emph{treatments}, that each $X_k$ can take on.  She randomly assigns each experimental unit $k$ to an index $i_k$ of $\vec{x}$ so that no unit is assigned the same index.  For each $k$, she then sets $X_k$ to be value at the $i_k$th slot of $\vec{x}$.  Units assigned the same treatment are called a \emph{group}.

The defining feature of an experiment is that the experimental units are randomly assigned their treatment groups.  Proper randomization over larger sample sizes makes negligible the probability that the groups vary in a systematic manner in terms of the \emph{noise factors} $F_k$, $\vec{Z}_k$, and $\vec{U}_k$ before the application of treatments.   This key property, \emph{exchangeability}, allowed the scientist to reject as unlikely the explanation that all the mice in the experimental group were of the always-die type and all the mice of the control group were of the always-live group~\cite{greenland86epidemiology}.  

However, randomization and a large sample are not sufficient to ensure valid conclusions.  The scientist must also ensure that no systematic differences are introduced to the groups after the application of the treatment.  For example, in addition to giving the mice in the experimental treatment group Drug X, the scientist also handled them more (to give them the drug), then any effects detected by the experiment could have resulted from the handling rather than the drug.

Under such conditions, the units will remain exchangeable under the \emph{null hypothesis} that the treatment has no effect. 
Thus, any difference in response that consistently shows up in one group but not another can only be explained by chance under the null hypothesis.  If given the sample size, this chance is small, 
then the scientist can reject the null hypothesis as very unlikely, providing probabilistic evidence of causal relationship, which we make precise in Section~\ref{sec:stats}.

Much of experimental design focuses on increasing the odds of finding an effect if one exists or on making such findings generalize to larger populations of units (see, e.g.,~\cite{cox00book}).
However, due to reasons of space, we limited our discussion to only issues of soundness, which we summarize as: 
\begin{enumerate}
\item start with exchangeable units,
\item randomly assign them treatments and introduce no other systematic differences, and
\item use a large sample to make ``unlucky'' assignments rare.
\end{enumerate}

\subsection{Information Flow Experiments}
\label{sec:ife}

To understand these issues in the context of information flow experiments, we consider how they apply to WDUD experiments.
At a high level, the fourth-party tracker would like to determine how a third-party web service uses information from or about visitors for selecting ads on first-party websites~\cite{guha10imc,wills12wpes,balebako12w2sp}.   
To model this problem as an experiment, we treat the information of interest as the factor $X_k$ that we will vary by applying treatments.  We treat the ads received as the response variable $Y_k$.  The additional factors $\vec{Z}_k$ and $\vec{U}_k$ that we will attempt to hold constant or randomize over includes the behavior of other users, advertisers, and other websites.

Mapping these goals to experimental science centers around deciding what counts as an experimental unit during the course of an information flow experiment. 
One obvious answer for WDUD is Google, the subject of our studies and the entity that processes the information of concern.  However, under this view, we have only a single system in question.
(While Google uses more than one server, they are interconnected.  For this reason, and simplicity, we treat Google as a single monolithic entity.)
Since we need at least two experimental units to compare across, 
we must separate our interactions with Google into multiple experimental units.

At the opposite extreme, we could count each input/output interaction with Google as a separate experimental unit, which gives each time step $t$ its own unit.  In WDUD, this could be viewed as treating each ad sent from Google in response to some request as a separate unit.

However, recall that  one of the major goals of WDUD is to determine the nature of Google's behavioral tracking of people.
  This suggests that interactions with Google at the granularity of people could be an appropriate experimental unit.  However, since we desire automated studies, we substitute separate browser instances for actual people.
In particular, we can use multiple browser instances with separate caches and cookies to simulate multiple users interacting with the web tracker.  
We can apply treatments to browsers by having them controlled by different scripts that automate different behaviors.
Table~\ref{tbl:compare} shows an overview of the relationship between experimentation for the experimental sciences and for IFA in general and WDUD in particular under this view.

\begin{table}
\begin{tab}{@{}lll@{}}
Experimental Science & Information Flow           & WDUD\\
\midrule
natural process      & system in question         & Google in its environment\\
population of units  & subset of interactions     & browser instances\\
factors              & input channels             & visitor behavior\\
treatments           & controlled inputs          & behavior profiles\\
noise factors        & uncontrolled channels      & other users, advertisers\\
response variables   & observed output channels   & sequences of ads\\
effect               & interference               & use of data\\
\end{tab}
\caption{Experimental Science, IFA, and WDUD Compared}
\label{tbl:compare}
\end{table}

\subsection{Limitations, Extensions, and Secondary Concerns}
We have not mentioned a few issues heavily emphasized in the design of experiments and statistics.
We consider them here to emphasize that they are not required for determining interference.

\paragraph{Random Sampling}
Acquiring units by randomly sampling from a more general population will, with high likelihood, provide a representative sample, which allows findings of effects to generalize to the population as a whole.  Random sampling is not needed if one just wants to prove the existence of an effect and not that the effect is widespread~\cite{zieffler11book}.
While results need not be general to show that Google tracks some behavior, showing that Google often does is more interesting. 
Thus, one may choose to run units at randomly selected times or locations for more general results.  %

Producing a representative sample could be abnormally difficult in our setting due to the possibility that Google alters its behavior in response to the atypical patterns of access exhibited by our experiments.  For example, Google could purposely make the reverse engineering its of system difficult by showing special behavior towards users that it suspects to be automated or probing.  Such atypical reactions from Google would not invalidate our conclusion that a flow information exists, but it mean that Google does not typically exhibit a flow.

\paragraph{Cross-unit Effects}
Many experimental designs emphasize the \emph{stable unit-treatment value assumption}, which requires that giving or withholding a treatment from one unit will not have an effect upon the other units~\cite{rubin86asa}.  Using experimental units that could plausibly satisfy this assumption is emphasized since it allows for a much wider ranger of statistical techniques.  However, it is not required for the permutation test of whether an effect exists~\cite{rosenbaum07asa}, which we discuss next.

The fact that determining the existence of an effect does not require a lack of interactions between units is key to our ability to do WDUD studies.  Any choice of unit other than all of Google, which leads to a sample size of one,  will possibly exhibit cross-unit effects by virtue of being multiplexed onto a single system.  Indeed, we found cross unit effects both at the level of ads and at the level of browsers.

\begin{experiment}\label{exp:cross-browsers}
To check for cross-unit effects, we studied whether multiple browser instances running in parallel affect one another.  Specifically, we compared the ads collected from a browser instance running alone to the ads collected by an instance running with seven additional browser instances each collecting ads from the same page.

A primary browser instance would first establish an interest in cars by visiting car-related websites.
We selected car-related sites by collecting, before the experiment, the top $10$ websites returned by Google when queried with the search terms  ``BMW buy'', ``Audi purchase'', ``new cars'', ``local car dealers'', ``autos and vehicles'', ``cadillac prices'', and ``best limousines''.
After manifesting this interest in cars, the instance would collect text ads served by Google on the International Homepage of Times of India.\footnote{\url{http://timesofindia.indiatimes.com/international-home}}  We attempted to reload the collection page $10$ times, but occasionally it would time out.  Each successful reload would have $5$ text ads, yielding as many as $50$ ads.

Our experiment repeated this round of interest manifestation and ad collection $10$ times using a new primary browser instance during each round.  We randomly selected $5$ of the rounds to also include seven additional browsers.  When the additional browsers were present, three of them performed the same actions as the primary one.  The other four would wait doing nothing instead of visiting the car-related websites and then went on to collecting ads after waiting.  All instances would start collecting ads at the same time.

The experiment showed that the primary browsers ran in isolation  would receive a more diverse set of ads than those running in parallel with other browsers.  We repeated the experiment four times (twice using $20$ rounds) and found this pattern each time:
\begin{center}
\begin{tab}{@{}ccc@{}}
Rounds        & Unique ads in isolation & Unique ads in parallel\\
\midrule
    10 &       37 &       25    \\
   10  &      46 &     33       \\
    20   &     58 &    47        \\
    20    &   57   &       52     \\            
\end{tab}
\end{center}
The presence of this pattern makes assuming an absence of cross-unit effects for browser instances tenuous at best.  While a statistical test could report whether the observed effect is significant (it is in one of the subsequent experiments), doing so would inappropriately shift the burden of proof: if a scientist would like to use a statistical analysis that requires an absence of cross-unit effects, then the onus is on him to justify the absence.

This and all other experiments were carried out using Python bindings for Selenium WebDriver, which is a browser automation framework. A test browser instance launched by Selenium uses a temporary folder that can be accessed only by the process creating it. So, two browser instances launched by different processes do not share cookies, cache, or other browsing data. All our tests were carried out with the Firefox browser running in a 64-bit Ubuntu 12.04 VM on a server located in [redacted]. When observing Google's behavior, we first ``opted-in'' to receive interest-based Google Ads across the web on every test instance. This placed a Doubleclick cookie on the browser instance. No ads were clicked in an automated fashion throughout any experiment.
\end{experiment}

\paragraph{Independent, Identically Distributed Samples} 
I.i.d.\ samples allow for powerful statistical techniques, which in some cases allow for smaller sample sizes or more detailed characterizations of a research finding.  However, this assumption is difficult to justify in our setting for the same reason that we cannot guarantee a lack of cross-unit interactions.  Fortunately, exchangeability, which can be seen as a weaker form of i.i.d., is sufficient for our purposes~\cite{greenland86epidemiology}.

\paragraph{Controlling Conditions}
Most experimental designs emphasize subjecting the units to conditions that are identical except for the experimental treatment.
The maxim goes \emph{control what you can; randomize what you can't}, but for our purposes it should read \emph{randomize what you want; control what you can't randomize} since relieved of the burden of creating i.i.d.\ samples, one need only control those aspects of the experiment that cannot be randomized.  However, ensuring that every experimental unit is subjected to approximately the same environment will typically produce less noisy results allowing one to reduce the sample size and make more definitive statements.

If one were interested in determining whether Google proper (not Google composed with its environment) had interference, then controlling conditions would take on a new significance.  In particular, the experimenter would have to control the low-level inputs from advertisers that could depend upon the high-level inputs to avoid confounding.

\bigskip
While some points in this section may seem pedantic, or even rudimentary, we will see in Section~\ref{sec:stings} that they are subtle enough to have led to real studies with poor statistical properties.  We now turn to making these properties precise.

\section{Statistical Analysis}
\label{sec:stats}

After designing and running an experiment, scientists must analyze the data collected.  In particular, they must quantify the probability that the collected responses could have occurred by chance through an unlucky random assignment of units to treatments.  
In this section, we reduce such quantification for information flow experiments to well known methods from statistics (Corollary~\ref{cor:independ-noninter}). 
We then discuss a particular method, \emph{permutation testing}, that is well suited for our setting of analyzing a complex black box system.
In the next section, we show that the test is general enough to formalize each of the prior WDUD studies.

A common approach to quantifying experimental results is by \emph{significance testing}~\cite{fisher35doe}.  The possibility of an unlucky assignment of units is formalized as a \emph{null hypothesis} that states that the groups differ by chance. %
A \emph{statistical test} of the data provides a \emph{p-value}, the probability of seeing results at least as extreme as the observed data under the assumption that the null hypothesis is true.  A small p-value implies that the data is unlikely under the null hypothesis.  Typically, scientists are comfortable rejecting the null hypothesis if the p-value is below a threshold of $0.05$ or $0.01$ depending on field.  Rejecting the null hypothesis makes the alternative hypotheses more plausible.

In our case, the null hypothesis is that the system in question has noninterference and the alternative of interest is the system has interference.
A combination of Theorem~\ref{thm:inter-cause} and the experimental design of Section~\ref{sec:ife} allows us to use the large class of statistical tests for independence of random variables to test for interference.

\begin{corollary}\label{cor:independ-noninter}
A test for independence of two random variables in science is a test of noninterference for information flow experiments. 
\end{corollary}
Since, as we discussed in Section~\ref{sec:ifa}, IFA lacks methods of conducting these studies, Corollary~\ref{cor:independ-noninter} fills an important gap.

However, some tests of independence require difficult-to-justify assumptions about the system in question.
For example, the most common statistical tests are \emph{parametric} tests that assume that the system in question's behavior is drawn from some known family of distributions with a small number of unknown parameters.  
Our experimental results show that such a family of distributions would have to be complex.

\begin{experiment}\label{exp:binning}
To understand how ads served by Google on a third-party website varies over time, we simultaneously started two browser instances, and collected the ads served by Google on the Breaking News page of \url{ChicagoTribune.com}.\footnote{\url{http://www.chicagotribune.com/news/local/breaking/}}  Each instance reloaded the web page 200 times, with a one minute interval between successive reloads. 

Figure~\ref{fig:exp-bin} shows a temporal plot of the ads served for each of these instances.
\begin{figure}
\begin{center}
\includegraphics[width = 11.5cm]{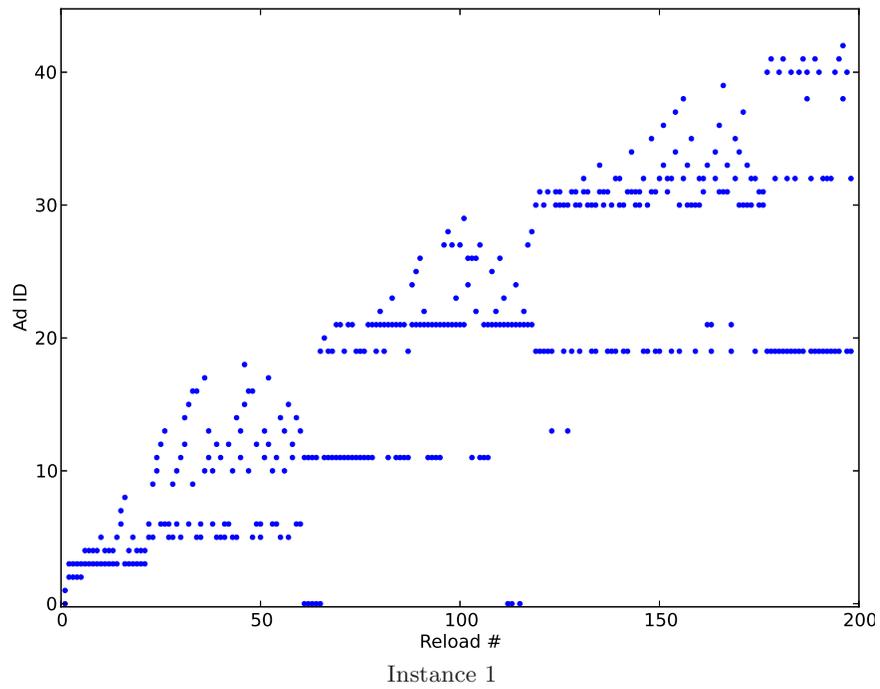}\\
Instance 1\\[3ex]
\includegraphics[width = 11.5cm]{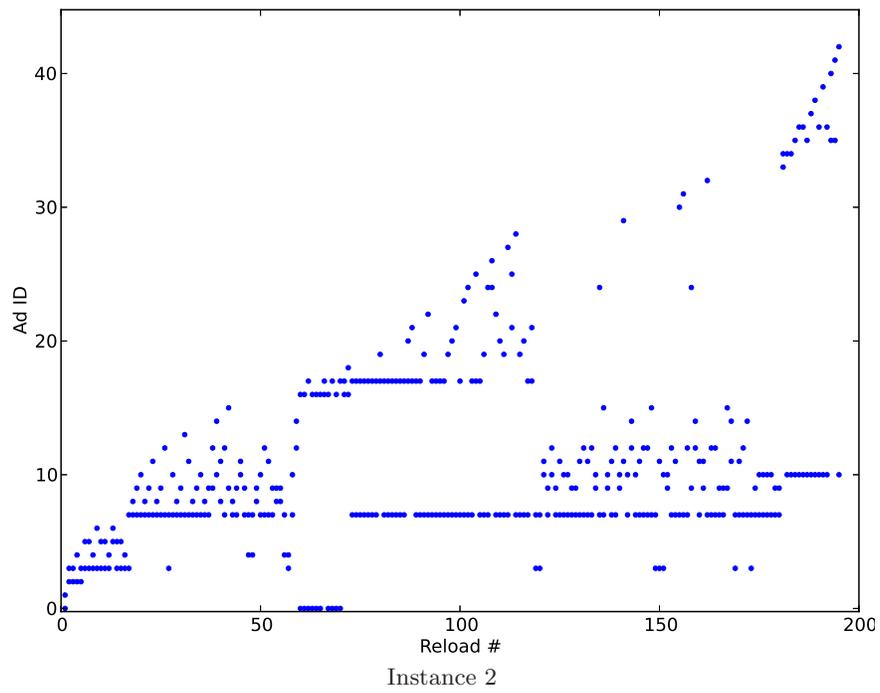}\\
Instance 2\\
\end{center}
\caption{The x-axes ranges over unique ads ordered by the time at which the instance first observed it in the experiment. The y-axis ranges over time measured in terms of page reloads.}
\label{fig:exp-bin}
\end{figure}
The plots suggest that each instance received certain kinds of ads for a period of time, before being switched to receiving a different kind.  One explanation for this behavior is that Google associates users with various ad pools switching users from pool to pool over time.  While hierarchical families of parametric models could capture this behavior, we are not comfortable making such an assumption and the resulting models would be more complex than those typically used in parametric tests.
\end{experiment}

Our results do not mean that one could not reverse engineer enough of Google to find an appropriate model.  However, they suggest that such reserve engineering would be difficult.  Furthermore, it runs against the spirit of performing black box information flow analysis.

Thus, we focus on \emph{non-parametric} tests, which do not require assuming a family of distributions and instead treat the generating distribution as a black box.  In particular, we will focus on \emph{permutation tests} %
(see e.g.,~\cite{good05book}).  Crucially, permutation tests (also known as \emph{randomization tests}) allow cross-unit interactions~\cite{rosenbaum07asa}, which can occur in WDUD studies (Experiment~\ref{exp:cross-browsers}).  

At the core of a permutation test is a \emph{test statistic} $s$, which is a function from the data, represented as a vector of responses, to a number.  
The vector of responses $\vec{y}$ has one response for each experimental unit.  The vector must be ordered by the random indices $i_k$ used to assign each unit $k$ a treatment from the treatment vector $\vec{x}$ prepared during the experiment.  Thus, the $k$th entry of $\vec{y}$ received the treatment at the $k$th entry of $\vec{x}$.

For example, an intuitive test statistic for an experiment with two treatment groups could use the first $n$ components of the data vector as the results of the experimental group and the remaining $m$ as the results for the control group where the groups have $n$ and $m$ units, respectively.  
A common test statistic over such data is the mean of the first $n$ responses less the mean of the last $m$ responses.
Intuitively, the higher the value of the test statistic, the more different the responses of the two groups are and larger the evidence of interference.

Since the scientist is allowed to pick any function $s$ from response vectors to numbers for the test statistic, the permutation test needs to gauge whether an observed data vector $\vec{y}$ produces a large value with respect to $s$.  To do so, it compares the value of $s(\vec{y})$ to the value of $s(\pi(\vec{y}))$ for every permutation $\pi$ of $\vec{y}$.  
Intuitively, this mixes the treatment groups together and compares the observed value of $s$ to its value for these arbitrary groupings.  Every time $s(\vec{y}) \leq s(\pi(\vec{y}))$ occurs, the test counts it as evidence that $s(\vec{y})$ is not particularly large.  

The significance of these comparisons is that under the null hypothesis of independence (noninterference), the groups should have remained exchangeable after treatment and there is no reason to expect $s(\vec{y})$ to differ in value from $s(\pi(\vec{y}))$.  Thus, we would expect to see at least half of the comparisons succeed.  Thus, we call a permutation $\pi$ such that $s(\vec{y}) \leq s(\pi(\vec{y}))$ fails to hold a \emph{rejecting permutation} since too many rejecting permutations leads to rejecting the null hypothesis.   

Formally,
the value produced by a (one-tailed signed) permutation test given observed responses $\vec{y}$ and a test statistic $s$ is
\begin{align}
\pt(s, \vec{y}) &= \frac{1}{|\vec{y}|!} \sum_{\pi \in \Pi(|\vec{y}|)} I[s(\vec{y}) \leq s(\pi(\vec{y}))] \label{eqn:pt}
\end{align}
where $I[\cdot]$ returns $1$ if its argument is true and $0$ otherwise, $|\vec{y}|$ is the length of $\vec{y}$ (i.e., the sample size), and $\Pi(|\vec{y}|)$ is the set of all permutations of $|\vec{y}|$ elements, of which there are $|\vec{y}|!$.

Recall that under significance testing, a p-value is the probability of seeing results at least as extreme as the observed data under the assumption that the null hypothesis is true.
$\pt(s, \vec{y})$ is a (one-tailed) p-value using $s$ and $\leq$ to define \emph{at least as extreme as} in the definition of p-value.  To see this, note that each permutation of data is equally likely under the null hypothesis $H_0$ that the treatments have no effect since the order of the responses is by treatment and otherwise random.  Thus,
\begin{align}
\pt(s, \vec{y}) &=  \sum_{\pi \in \Pi(|\vec{y}|): s(\vec{y}) \leq s(\pi(\vec{y}))} \Pr[\vec{Y} = \vec{y} \given H_0] \label{eqn:pvalue}
\end{align}
matching the definition of a p-value.  One could use other definitions of \emph{as extreme as} by replacing the $\leq$ in~\eqref{eqn:pt} and~\eqref{eqn:pvalue} by $\geq$ or by comparing the absolute values of $s(\vec{y})$ and $s(\pi(\vec{y}))$ to check for extremism in both  directions (a two-tailed test).

Good discusses using sampling to make the computation of $\pt(s,\vec{y})$ tractable for large $\vec{y}$~\cite{good05book}.  Greenland provides detailed justification of using permutation tests to infer causation~\cite{greenland11book}.

We do not claim that permutation tests are the only suitable statistical tests.  However, we find it sufficient to characterize the prior WDUD works, which we do next.

\section{Formalization of Prior Work}
\label{sec:stings}

We examine the four WDUD studies that attempt to determine how Google uses the information it collects~\cite{wills12wpes,guha10imc,sweeney13cacm,balebako12w2sp}.  
We are able to systematically explain, extend, and compare the works by framing them as permutation tests for analyzing the results of information flow experiments.
Our framework makes clear the reasoning employed by these works and identifies improvements to their experimental designs.  To that end, we make suggestions for conducting future studies throughout, which we summarize in Section~\ref{sec:sugg-conc}.
However, we select and scrutinize these studies because they contain interesting and important results that we would like to place into the context of IFA; not because we believe them to contain major flaws.

We organize our presentation by the type of test statistic used by each work.  In the case of Sweeney's study~\cite{sweeney13cacm}, the test statistic is provided by her own statistical analysis.  For the others, we select one that naturally captures their informal reasoning.  We discuss the study of Wills and Tatar twice since they employ two very different styles of reasoning.
We end with an empirical comparison of the test statistics discussed.
In addition to shedding light on foundations of these studies, this tour of prior work shows that the permutation test is a general framework for reasoning about the statistical significance of information flow experiments.

\subsection{The $\chi^2$ Test}

We will start by considering a key finding in Sweeney's study~\cite{sweeney13cacm}: searching for a characteristically black first name will produce a higher rate of Instant Checkmate ads including the word ``arrest'' than searching for characteristically white first names.  While much of Sweeney's study consisted of finding appropriate names to test and exploring the ramifications of these results, we will focus on the core finding of a flow of information from the first name of the search query to the ads shown.

She made her finding by Googling for various names and checking the ads returned with the results over the course of a month.  For each Instant Checkmate ad returned, she recorded whether it contained the word ``arrest''.  Consistent with our recommendation, she used a new browser instance each time she Googled a name.  Thus, we can view each browser instance as an experimental unit.  Each unit received the treatment of either a characteristically black or white name.  She did not provide details of how she allocated treatments to units.  Thus, a methodological concern is that her allocation might not have been properly randomized since Google's behavior could be time dependent.

Given the long period of time over which she conducted her experiments, even larger temporal effects may be present.
(Indeed, the theoretical benefit from increasing sample size is often partly removed by the increase in variation among units from a decreased ability to hold conditions constant across them~\cite{cox00book}.)
However, since we have no reason to suspect that changes in Google's behavior would affect these results, for analyzing her study, we will assume she randomized the treatments.

To model her work in terms of a SEM, we use the factor $X_k$ to denote the race of the first name of the $k$th instance.  The response variable $Y_k$ can be modeled as taking on three values: $1$ for an Instant Checkmate ad with the word ``arrest'', $-1$ for one without, and $0$ for no Instant Checkmate ads.  (She never observed more than one Instant Checkmate ad for a search.)  

Unlike the other studies we will consider, Sweeney already provided a statistical analysis of her results.  She used the $\chi^2$ test, a popular nonparametric statistic.  A theoretical justification of the $\chi^2$ test is that it asymptotically approaches a permutation test~\cite{ludbrook08ije}.  Thus, we can understand her test in terms of permutation testing.  
With the size of her data, such approximations become not only accurate, but useful for computational reasons.  Nevertheless, we believe the permutations continue to provide the semantics behind such approximations, especially considering that the justification of the $\chi^2$ test includes an assumption that the experimental units are independent~\cite{lehmann05book}, which is unlikely as discussed in Section~\ref{sec:ife}.

\subsection{Counting}

Consider the WDUD study of Wills and Tatar in which they pose as various visitors to first-party websites~\cite{wills12wpes}.  They perform multiple experiments looking at different features of Google's behavior.  Here we will discuss one of their approaches in detail; %
we discuss another in~\ref{sec:nonce}.

Consistent with our approach (Section~\ref{sec:ifa}), they use separate browser instances to simulate separate users, which represent their experimental units.
The treatments they apply to each instance corresponds to either inducing some interest or not by searching for a word on a website.  They had each instance participate in multiple sessions that consisted of inducing the interest followed by visiting a different third-party web page that serves Google ads.
(Actually, to reduce resource use they induced more than one interest per unit making their study multi-factorial in design.  For simplicity, we will ignore this complication, but it can be handled by our framework.  See, e.g., \cite{good05book}.)

Formally, the factor of interest $X_k$ is the search entry field.  The response variables $Y_k$ are the ads seen at the third-party website.
Their test statistic is the percentage of sessions that included a non-contextual ad containing a keyword associated with the treatment.
To formalize their test statistic, let $W_t$ be the set of keywords they associated with interest $t$.
Representing the data collected during a session as a list $\ell$ of ad-context pairs, let $\hit(\ell,t)$ be true iff there exists a pair $\langle a, c\rangle$ in $\ell$ such that the ad $a$ contains a keyword in $W_t$ and $c$ is not a context relative to $t$.
(They determined context by hand.)

The data collected is a vector $\vec{y}$ of responses for each unit where each response is a list of sessions.  Let first $n$ of them be those with the induced interest.  Let $\percent$ compute the percentage of sessions with a non-contextual ad among the responses within a range: $\percent(\vec{y}, a, b) = 100*\sum_{k=a}^{b} \sum_{j = 1}^{|\vec{y}[k]|} \hit(\vec{y}[k][j], t)/N$ where $N$ is the number of sessions in that range: $N = \sum_{k=1}^n |\vec{y}[k]|$.  In their Figure~5, they plot $\percent(\vec{y}, 1, n)$ and $\percent(\vec{y}, n+1, n+m)$ where $n$ and $m$ are the numbers of instances with and without the interest induced.

Whereas they reasoned informally by comparing these two numbers, we can provide rigorous statistics based upon them by using a test statistic based on them.  One such test statistic would be $s_\percent(\vec{y}) = \percent(\vec{y}, 1, n) - \percent(\vec{y}, n+1, n+m)$.  If inducing the interest increases the number of ads shown about it, then we would expect $s_\percent(\vec{y})$ to be larger than $s_\percent(\pi(\vec{y}))$ for permutations that mix the responses.

A feature of their design is that their instances are long running with multiple sessions spanning a week.  While these long-running instances do not increase the sample size, collecting more data on each unit allows for a more complete view of that unit allowing for the detection of subtle differences and more detailed test statistics over multiple measurements~\cite{good05book}.
Furthermore, 
it allows them to see behavior that Google might not manifest over a short time period. 
Indeed, consistent with their own finding, we found that Google would not update its listing of a person's gender until over a day of interactions.

\begin{experiment}\label{exp:gender}
We created two browser instances and randomly assigned one to visit the top $100$ websites for females as determined by Alexa, which takes approximately 5.5 hours.  The other visited the top $100$ sites for males.
Before visiting each site, we checked the gender inferred by Google on its Ad Settings page, which provides users with a summary of Google's profile of them.
The instances idled on each site for three minutes.  After visiting all $100$ pages, they idled for two hours.  They repeated this process until Google inferred a gender.
Google inferred the gender of both instances during the fifth round of training at 30 hours 19 minutes for the female and 30 hours 12 minutes for the male.
\end{experiment}

\subsection{Cosine Similarity}
\label{sec:baseline}

Guha et al.\ present a methodology for performing WDUD~\cite{guha10imc}, which is also followed by Balebako et al.~\cite{balebako12w2sp}.  Their methodology uses three browser instances.  Two of them receive the same treatment and can be thought of as controls.  The third receives some experimental treatment.  The treatments consist of having them visit web pages, perform searches, and click on links.  For each instance, after having them display behavior dependent upon their treatment, they collect the ads Google serves them, which they compare using a similarity metric.
Based on experimental performance, they decided to use one that only looks at the URL displayed in each ad.  For each instance, they perform multiple page reloads and record the number of page reloads for which each displayed URL appears.  From these counts, they construct a vector for each unit where the $i$th component of the vector contains the logarithm of the number of reloads during which the $i$th ad appears.  To compare runs, they compare the vectors resulting from the instances using the cosine similarity of the vectors.

More formally, their similarity metric is $\simm(\vec{v}, \vec{w}) = \coss(\lnmap(\vec{v}), \lnmap(\vec{w}))$ where
$\vec{v}$ and $\vec{w}$ are vectors that record the number of page reloads during which each displayed URL ad appears, $\lnmap$ applies a logarithm to each component of a vector, and $\coss$ computes the cosine similarity of two vectors.  
They conclude that a flow of information is likely if $\simm(\vec{v}_{\mathsf{c1}}, \vec{v}_{\mathsf{c2}})$ is much larger than $\simm(\vec{v}_{\mathsf{c1}}, \vec{v}_{\mathsf{e}})$ where $\vec{v}_{\mathsf{c1}}$ and $\vec{v}_{\mathsf{c2}}$ are the responses from the two control instances and $\vec{v}_{\mathsf{e}}$ is the response from the experimental instance.

Their intuition of comparing two control instances to get a baseline amount of noise in the system is a good one.  However, as we discuss in Section~\ref{sec:ife}, browser instances make for good units, not individual ads.  Thus, their experiment only consists of $3$ experimental units, too few to achieve reliable results.  Indeed, the p-value of a permutation test cannot be less than $1/3! \approx 0.17$ with just $3$ units.

To generalize their method to larger sample sizes, we replace their metric $\simm$ with one that can compare more than two vectors.  One choice is to first aggregate together multiple URL-count vectors by computing the average number of times each URL appeared across the aggregated units.  Formally, let $\mathsf{avg}(\vec{u})$ compute the component-wise average of the vectors in $\vec{u}$, a vector of vectors of URL counts.  
We can then define a test statistic $s_{\simm}(\vec{y}) = -\simm(\avg(\vec{y}_{1:n}), \avg(\vec{y}_{n+1:n+m}))$ where $\vec{y}_{a:b}$ is the sub-vector consisting of the entries $a$ though $b$ of $\vec{y}$, the first $n$ responses are from the experimental group, and the next $m$ are those from the control group.  
We use negation since our permutation test takes a metric of difference, not similarity.
Intuitively, the permutation test using the test statistic $s_{\simm}$ will compare the \emph{between-group} dis-similarity to the dis-similarity of vectors that mix up the units by a permutation.  In aggregate, the dis-similarity of these mixed up vectors provide a view on the global dis-similarity inherit in the system.

\subsection{Simulated Comparisons: Nonce Presence}
\label{sec:nonce}

During their study, Wills and Tatar observe Google serving the ad ``LGBT for Obama'' on \url{thefreedictionary.com}, a site that is not about LGBT (lesbian, gay, bisexual, or transgendered) issues~\cite{wills12wpes}.  While they do not conclude that Google necessarily selects ads based upon a sensitive interest in LGBT issues, they note this behavior as suspicious.  Their suspicion is based on using LGBT like a nonce by virtue of it being rare.  
  That is, LGBT serves to connect Google's selection of a low-level ad to sensitive high-level information provided by browsing LGBT-related sites that are otherwise unrelated to the ad.

Since only $3.4\%$ of U.S.\ adults self-identify as LGBT~\cite{gates12gallup}, Google selecting LGBT ads without using some information seems unlikely.
However, assuming that, without tracking, Google would present ads in proportion to the target population size, we would expect that $3.4\%$ of ads that target a sexual orientation would be LGBT targeting ads.  Thus, if the LGBT related ad was only one of a large number of ads targeting sexual orientation, then a conclusion of a flow of information could be a false positive.  

To examine the quality of LGBT as a nonce, we searched $397,361$ ads that we collected during our studies.  Only $30$ of them contained any of the words  ``gay'', ``lgbt'', ``lesbian'', or ``queer''.  
With just $0.0075\%$ of the ads in our sample containing these words, seeing one is a noteworthy event.

Another test of Wills and Tatar involved using LinkedIn and Pandora profiles with the location set to New York City.  The authors wanted to determine whether Google used the profile locations for selecting advertisements.  However, despite seeing numerous ads for NYC, the authors do not conclude that Google uses the profile location since (1) NYC ``is a popular location in general'' and (2) they did ``not have a baseline for comparison''~\cite[page 9]{wills12wpes-tr}.  We found $2028$ instances of ``NYC'' and ``New York'', %
$0.5\%$ of the ads we sampled, despite our server not  being located near NYC.  Thus, seeing NYC related words is much less noteworthy than LGBT related words.

Such reasoning might appear to have nothing to do with permutation tests.  However, we can even view it as a special case of the permutation test in which most of the test runs were not actually done explicitly.  Such a view does not strictly adhere to the assumptions needed to draw causal conclusions since it lacks randomization.  Nevertheless, it provides a conceptual basis for converting informal checks like the one above into actual randomized experiments.

To see how, let the data vector $\vec{y}$ have the observed response with the nonce in it at its first position and the observations that led the scientist to believe that the nonce is in fact rare fill every other slot.  Ideally, these observations would be from a randomized experiment, but the reasoning leads to an informal assessment of a convenience sample, such as us looking at all the ads we collected.  %
Let the test statistic $s_n$ return $1$ if the first component of a data vector contains the nonce $n$ and $0$ otherwise.
It may seem odd to choose a test statistic that ignores all but the first response, but since the test statistic will be used in a permutation test, every response of $y_k$ will contribute to the overall p-value produced by being shifted into the first position by permutations. 
The p-value produced by the permutation test will be $\pt(s_n, \vec{y}) = \mathsf{count}(\vec{y}, n) / |\vec{y}|$ where $\cnt(\vec{y}, n)$ counts up the number of responses of $\vec{y}$ that contains the nonce $n$. 

The above model also extends to nonces justified on theoretical grounds, such as those from a random number generator.  For example, if we take $\vec{y}_m$ to be a vector of length $m$ with the nonce only in its first component, then $\lim_{m \to \infty} \pt(s_n, \vec{y}_m) = \lim_{m \to \infty} 1/m = 0$ showing that the p-value allows rejection of the null hypothesis (acceptance of interference) with certainty given a perfect nonce.  If we let $\vec{y}_{m,w}$ be a vector of length $m$ with the $n$ in the first component and $w$ of the following components, then $\lim_{m \to \infty} \pt(s_n, \vec{y}_{m, \lceil p*m \rceil})$ and $\lim_{m \to \infty} \pt(s_n, \vec{y}_{m, \lfloor p*m \rfloor})$ both equal $p$, capturing the idea that seeing a nonce with probability $p$ of occurring by chance (such as those produced by a random number generator) implies that one can infer causation with a p-value of $p$.

The nonce analysis has appeared elsewhere.  Both watermarks and trap streets, mentioned in the introduction for copyright infringement detection, are nonces~\cite{wagner83sp,swanson98ieee,monmonier96maps}.
Sekar used a similar analysis to find web application vulnerabilities in a black box fashion~\cite{sekar09ndss}.

Nonces are typically thought of in terms of information flow, not physical causation, raising the question of what using a nonce corresponds to in the natural sciences.  In that setting, nonces correspond to an experimental treatment and a response so extreme that the scientist dispenses with the control group.  For example, the scientists testing the ability of a bomb to destroy an island (such as during Operation Crossroads), do not typically set aside a control island. 
\subsection{Comparison of Test Statistics}
\label{sec:pos}

Given all the test statistics discussed, one might wonder how they compare.  
We will empirically compare the tests in our motivating setting of WDUD.  However, we caution that our experiment should not be considered definitive since other WDUD problems may result in different results.  We recommend that each experiment is preceded by a pilot study to determine the best test(s) for the experiment's needs.  For example, we have found pilot studies useful for selecting distinguishing keywords to search for in ads.

\begin{experiment}\label{exp:pos}

Each run of the experiment involved ten simultaneous browser instances, each of which represent an experimental unit.  We used a sample size of ten due to the processing power and RAM restrictions of our server.
For each run, the script driving the experiment randomly assigns five of the instances, the experimental group,  to receive the treatment of manifesting an interest in cars.  
As in Experiment~\ref{exp:cross-browsers}, an instance manifests its interest by visiting the top $10$ websites returned by Google when queried with certain automobile-related terms: ``BMW buy'', ``Audi purchase'', ``new cars'', ``local car dealers'', ``autos and vehicles'', ``cadillac prices'', and ``best limousines''.
The remaining five instances made up our control group, which remained idle as the experimental group visited the car-related websites.  Such idling is needed to remove time as a factor ensuring that the only systematic difference between the two groups was the treatment of visiting car-related websites.

As soon as the experimental group completed visiting the websites, all ten instances began collecting text ads served by Google on the International Homepage of Times of India.  
As in Experiment~\ref{exp:cross-browsers}, each instance attempted to collect $50$ text ads by reloading a page of five ads ten times, but page timeouts would occasionally result in an instance getting fewer.  We repeated this process for 20 runs with fresh instances to collect 20 sets of data, each containing ads from each of ten instances.

Across all runs of the experiment, we collected $9832$ ads with $281$ being unique.  Instances collected between $40$ and $50$ ads with two outliers each collecting zero.  Both outliers were in the $19$th run and in the experimental group.
We analyzed the data with multiple test statistics.
Table~\ref{tbl:positive} summarizes the results with the last row showing the number of statistically significant results under the traditional cutoff of $5\%$.
\begin{table}
\caption{p-values for the permutation tests}\label{tbl:positive}
\begin{tab}{@{}ccccc@{}}
 Data set & $s_{\simm}$ & $s_{\msf{kw}}$  & $s_{\percent}$ & $\chi^2$ \\
\midrule
$1$ & $0.007937$         & $0.003968$    & $0.222222$    & $3.1815\ \times 10^{-33}$ \\
$2$ & $0.007937$         & $0.003968$    & $1.000000$    & $1.75166\times 10^{-24}$ \\
$3$ & $0.015873$         & $0.019841$    & $0.500000$    & $7.33209\times 10^{-13}$ \\
$4$ & $0.007937$         & $0.003968$    & $0.083333$    & $6.31635\times 10^{-33}$ \\
$5$ & $0.007937$         & $0.099206$    & $1.000000$    & $4.15872\times 10^{-07}$ \\
$6$ & $0.007937$         & $0.003968$    & $0.500000$    & $3.5201\ \times 10^{-31}$ \\
$7$ & $0.007937$         & $0.003968$    & $0.222222$    & $4.87536\times 10^{-25}$ \\
$8$ & $0.007937$         & $0.003968$    & $1.000000$    & $2.93566\times 10^{-30}$ \\
$9$ & $0.007937$         & $0.003968$    & $1.000000$    & $2.30865\times 10^{-25}$ \\
$10$ & $0.007937$        & $0.003968$    & $0.222222$    & $2.73048\times 10^{-26}$ \\
$11$ & $0.460317$        & $0.015873$    & $0.500000$    & $1.84605\times 10^{-07}$ \\
$12$ & $0.023810$        & $0.023810$    & $1.000000$    & $8.78432\times 10^{-13}$ \\
$13$ & $0.007937$        & $0.003968$    & $1.000000$    & $1.74223\times 10^{-20}$ \\
$14$ & $0.015873$        & $0.003968$    & $1.000000$    & $3.3131\ \times 10^{-26}$ \\
$15$ & $0.039683$        & $0.011905$    & $1.000000$    & $2.16042\times 10^{-17}$ \\
$16$ & $0.007937$        & $0.003968$    & $0.500000$    & $4.1491\ \times 10^{-29}$ \\
$17$ & $0.031746$        & $0.003968$    & $1.000000$    & $9.44887\times 10^{-17}$ \\
$18$ & $0.015873$        & $0.007937$    & $0.500000$    & $3.42116\times 10^{-15}$ \\
$19$ & $0.007937$        & $0.087302$    & $1.000000$    & $4.44136\times 10^{-21}$ \\
$20$ & $0.111111$        & $0.003968$    & $0.500000$    & $3.17792\times 10^{-27}$ \\
\midrule
$\text{Number} < 5\%$ & 18  &              18        &      0       &                        20\\
\end{tab}
\end{table}
First, we used the permutation test with $s_\simm$, the extension of Guha et al.'s cosine similarity metric~\cite{guha10imc}, as the test statistic.  
Observe that there are  $10! > 3$ million different permutations for the ten instances.  However, since $\simm$ treats the response vector provided to it as two sets, intuitively, the experimental and control groups, many permutations will produce the same value for $s_\simm$.  To speed up the calculation, we replaced comparing all permutations with comparing all partitions of the responses into equal sized sets of $5$, yielding only $\binom{10}{5} = 252$ comparisons.  Since cosine similarity is a symmetric statistic, there can be at most $126$ unique values.  Since at least one of these will be equal to the actual observed $\coss(\vec{y})$, the minimum possible p-value is $1/126 = 0.007937$.  Looking at the p-values from Table~\ref{tbl:positive}, we see that twelve out of the twenty have the minimum possible p-value. 

Second, we carried out the permutation test using $s_{\percent}$, the keyword-based statistic of Wills and Tatar~\cite{wills12wpes}, as the test statistic.  
From our initial search terms, we created a set of keywords containing ``bmw'', ``audi'', ``car'', ``vehicle'', ``automobile'', ``cadillac'', and ``limo'', words whose presence we believe to be indicative of an instance being in the experimental group.
The statistic $s_\percent$ counts the number of sessions that had an ad with any of these keywords present. 
However, our instances did not participate in multiple sessions, as Wills and Tatar's did.  Thus, directly applying $s_\percent$ to our data, treating each response as a single session, produces lackluster results.  
Since we had only five instances per group, the values of this statistic can only take on $10$ different values, making it a blunt instrument for distinguishing groups.  
Indeed, the p-values were not conclusive with most of them being either $1.0$ or $0.5$.  The choice of test statistic is an important one.  

To give the keyword approach a fair chance, we also tested an adapted one $s_{\msf{kw}}$, which looks at the number of ads that each instance received containing a keyword rather than the number of sessions.
We defined our statistic to be the the number of ads that contained any of the keywords amongst the first half (intuitively, the experimental group) of the responses less the number in the second half (intuitively, the control group).  
As with $s_{\simm}$, we have at most 252 unique comparisons to make.  Thus, the minimum possible p-value from our experiment is $1/252 = 0.003968$.  Most the p-values computed from our data sets are at their minimum.  Observe from Table~\ref{tbl:positive} that the p-values obtained from $s_{\msf{kw}}$ are less than the corresponding p-values from $s_{\simm}$. 
We believe this improvement is from the domain knowledge provided by the keywords.  

Lastly, we conducted the $\chi^2$ test on a $2\times 2$ contingency table computed from the data from each round.  The type of treatment was represented in rows, while the presence or absence of keywords was represented in the columns. Thus, the top-left entry in the table was the number of ads in $AD_t$ containing a keyword. The p-values obtained from running the $\chi^2$ test on our data is shown in Table~\ref{tbl:positive}.  While these results are impressively low, they can be misleading given that the $\chi^2$ test assumes the independence of experimental units~\cite{lehmann05book}.

For comparison purposes, we re-run the above experiment without having the experimental group manifest any interests.  That is, we compared two control groups against one another expecting to not find statistically significant differences.  We found that each of the statistics produced one statistically significant result except for the $\chi^2$, which produced $12$.
\end{experiment}

The wide range of tests might tempt one into running more than one test on data.
However, running multiple tests increases the chance of getting a low p-value for one of them by an unlucky randomization of units rather from an effect.  Thus, one cannot look just at the test that produced the lowest p-value.  Rather one must report them all or apply a correction for multiple tests such as 
those for the %
false discovery rate~\cite{benjamini95rss}.
\section{Conclusions and Suggestions}
\label{sec:sugg-conc}

We have identified a range of problems that can be approached systematically as information flow experiments.  This work provides a fresh perspective on these problems and on IFA, which has long been dominated by white box program analysis.
We have explained information flow experiments in terms of causality by relating noninterference to the notion of effect.
This observation allows us to apply in a rigorous manner the methods and statistics of experimental science to problems of information flow.
In particular, we have recommended an experimental methodology and a statistical analysis, based on permutation testing, that is well suited to our setting.
This viewpoint has allowed us to systematically find the limitations and abilities of information flow experiments in general and of specific studies individually.  

In particular, we have examined the emerging area of WDUD and formalized studies in the area as experiments in our framework.  
The value of this exercise is two fold.  First, by placing these empirical studies into our formal framework, we can closely study their reasoning using standard notions, such as experimental units, and metrics of soundness, such as the p-value.  In particular, we discuss whether the implicit assumptions made by these works are reasonable and how to improve their analyses.
Second, we test the applicability of noninterference to real studies outside of its comfort zone of program analysis.  

This process has allowed us to convert the abstract principles of experimental design and analysis into concrete suggestions:

\newcommand{\itemx}{\item}
\begin{enumerate}
\itemx\emph{Use an appropriate statistical test.}   %
Attempting to shoehorn data into familiar statistics can result in incurring requirements that cannot be met.  
A lack of cross unit effects, random samples, and independent, identically distributed experimental units each enable additional statistical techniques, but are difficult, if not impossible, to achieve in our setting.  Fortunately, they are \emph{not} required for permutation testing.

\itemx \emph{Start with exchangeable units.}
The exchangeability of units is a requirement.  We ran multiple browser instances in parallel to obtain exchangeable units.  While cross-unit effects were likely to have occurred, we met the requirements of our chosen statistical analysis: permutation testing.

\itemx \emph{Randomly assign treatments to units.}  
Randomization provides the justification needed for permutation testing and for avoiding the more difficult conditions above.

\itemx \emph{Let the requirements of a statistical analysis guide data collection.}
For example, Wills and Tatar's intuitive analysis using New York City did not collect data enough for analysis (Section~\ref{sec:nonce}).  The right time to select an analysis is before the experiment since it can reveal the data that needs to be collected.

\itemx \emph{Use domain knowledge gained during pilot studies to select a test statistic.}  Finding the correct keywords to examine in ads allowed us to not only get results that were statistically significant, but also intuitive.

\itemx \emph{Be selective.}
Finding websites that produced consistent results was difficult.  For example, before trying the Times of India, we used Fox News.  Despite using Google for advertising, we could not find any effects.  Since we wanted to find an effect from Google, not Fox News, we were free to be selective and should have moved on to another site earlier.  The situation is different if you want to prove that an effect is widespread, which requires random sampling~\cite{zieffler11book}.
\end{enumerate}

While statistical analysis can be intimidating due to their complex requirements, selecting the correct test is liberating by also identifying what conditions you need not worry about.

\section{Future Directions}
\label{sec:fut-work}

\paragraph{Demonstrating Noninterference}
An analyst might wish to show that a system has noninterference. 
However, the permutation test requires that the null hypothesis be that the system has noninterference.  Thus, it can only provide a quantitative measure of the evidence \emph{against} noninterference.  
Conceptually, proving noninterference would require looking at every test statistic under every input sequence.
Since examining an infinite set of sequences is impossible, using the scientific method to show that a system has noninterference would require building a theory of the system's operation and then proving noninterference in that theory. 

\paragraph{Other Notions of Information Flow and Causation}
We examined only one information flow property, a probabilistic noninterference, and one notion of causality, effect.  Exploring the many alternatives could tighten the connection between the two fields and further organize each.
We believe the interplay between these two fields can be rich with each benefiting from the other's perspective and techniques.

\paragraph{Monitoring and Observational Studies}
Passive monitoring in IFA corresponds to observational studies.  A wide range of work deals with the cases under which one can infer causation from a correlation learned from an observational study (see, e.g.,~\cite{pearl09book}).  Future work can import these results to IFA showing how monitoring could be useful in some cases despite its inherit unsoundness~\cite{mclean94sp, schneider00tissec, volpano99sas}.

For example, author de-anonymization (e.g.,~\cite{stamatatos09jasist,narayanan12sp}), detecting cheating (e.g.~\cite{palazzo10j}), and detecting plagiarism of a third-party's work (e.g.,~\cite{maurer06jucs}) all correspond to monitoring since the analyst does not control the sensitive messages (e.g., an anonymous posting).  However, in practice, authors are de-anonymized using comparisons.  We conjecture that such analyses could be shown sound under similar assumptions as those used for observational studies.

\paragraph{Related Experiments}
Problems outside of IFA are also instances of investigations.
For example, Google ran a nonce-like experiment to determine whether Bing's search results were mimicking Google's~\cite{sullivan11blog}.  Thus, rather than tracking information flows, Google's experiment involved tracking flows of behavior.  In particular, their nonce involved Google returning unusual search results.  Google then observed Bing mimicking this behavior after Bing observed users clicking on the unusual results in Internet Explorer.  
Bowen et al.\ conduct access-control experiments by monitoring decoy files that attract adversaries into accessing them~\cite{bowen09baiting}.  
Another problem is \emph{provenance}, tracking the handling of data~\cite{buneman01icdt}, which is an extended form of IFA in which the analyst needs to know not just the source of the data, but also the step-by-step flow and handling of the data in a network.
Comparing and combining experiments from these fields with our own approach would provide a more comprehensive approach to data governance.

In general, information flow experiments allow an analyst to exercise oversight and detect transgressions by an entity not controlled by the analyst and unwilling to provide the analyst complete access to the system.  We see this setting becoming ever more common: data lives in the cloud, jobs are outsourced, products licensed, and services replace infrastructure.  In each of these cases, a party has ceded control of a resource for efficiency.  Nevertheless, each party must ensure that the other abides by their agreement while having only limited access to the other.  Thus, we envision experimentation, as opposed to white box verification, playing an increasing role in computer security and society in general.

\paragraph{Acknowledgments}
We thank Divya Sharma and Arunesh Sinha for many helpful comments on this work.

\appendix
\newgeometry{top=1in, bottom=1in, left=0.75in, right=0.75in}
\section*{Appendix}

\section{System Formalism}

For a finite set, let $\dist{X}$ be the set of distributions over $X$.
Let $\degen(x)$ be the degenerate distribution assigning probability $1$ to $x$.
Let $[]$ be the empty list.
Let $\vec{\imath}\cons i$ be the list created by appending $i$ to $\vec{\imath}$, and let $i \cons \vec{\imath}$ be the list created by prepending $i$ to $\vec{\imath}$.

Let a probabilistic Moore Machine be $Q = \langle \mathcal{S}, s_0, \mathcal{I}, \mathcal{O}, \tau, \sigma\rangle$ where $S$ is a finite set of states, $s_0$ is the initial state, $\mathcal{I}$ is a finite input set, $\mathcal{O}$ is a finite output set, $\tau: \mathcal{S} \cross \mathcal{I} \to \dist{\mathcal{S}}$ is the state transition function, and $\sigma: \mathcal{S} \to \mathcal{O}$ is the output function.  

Let $Q(s, \vec{\imath})(\vec{o}, \vec{s})$ be the probability of the seeing the trace $\vec{s}[1], \vec{o}[1], \vec{\imath}[1], \vec{s}[2], \vec{o}[2], \vec{\imath}[2], \ldots, \vec{s}[k], \vec{o}[k], \vec{\imath}[k],  \vec{s}[k+1], \vec{o}[k+1]$:
\begin{align}
Q(s, [])([\sigma(s)], [s]) &= 1 &&\\
Q(s, i\cons\vec{\imath})(\sigma(s)\cons\vec{o}, s\cons\vec{s}) &= \sum_{s'} \tau(s, i)(s') * Q(s',\vec{\imath})(\vec{o},\vec{s}) &&\\ 
Q(s, \vec{\imath})(\vec{o}, \vec{s}) &= 0 &&\text{otherwise}
\end{align}

We take the distribution $Q(\vec{\imath})$ over outputs to such that $Q(\vec{\imath})(\vec{o}) = \sum_{\vec{s}} Q(s_0, \vec{\imath})(\vec{o}, \vec{s})$.

The following lemma provides a closed form for $Q(s,\vec{\imath})$, which will become useful later.
\begin{lemma}\label{lem:q-out}
For all $Q$, $s$, and $\vec{\imath}$, $\vec{o}$, and $\vec{s}$ of equal lengths $k \geq 0$, $k+1$, and $k+1$, respectively,
\begin{align}
Q(s, \vec{\imath})(\vec{o}, \vec{s}) 
&= \degen(s)(\vec{s}[1]) *  \degen(\sigma(\vec{s}[1]))(\vec{o}[1]) * \prod_{\kappa = 1}^{k} \tau(\vec{s}[\kappa], \vec{\imath}[\kappa])(\vec{s}[\kappa+1]) * \degen(\sigma(\vec{s}[\kappa+1]))(\vec{o}[\kappa+1]) 
\end{align}
\end{lemma}
\begin{proof}
Proof by induction.
Base Case: $k = 0$.
$Q(s, [])([\sigma(s)], [s]) = \degen(s)(\vec{s}[1]) *  \degen(\sigma(\vec{s}[1]))(\vec{o}[1])$, which is $1$ if $\vec{s}[1] = s$ and $\sigma(\vec{s}[1]) = \vec{o}[1]$ and $0$ otherwise as needed.

Inductive Case: $k > 0$.
\begin{align}
&Q(s, i\cons\vec{\imath})(o\cons\vec{o}, s'\cons\vec{s})\\
&= \degen(s)(s') * \degen(\sigma(s))(o) * \sum_{s''} \tau(s, i)(s'') * Q(s'',\vec{\imath})(\vec{o},\vec{s})\\
&= \degen(s)(s') * \degen(\sigma(s))(o) * \sum_{s''} \tau(s, i)(s'') * \degen(s'')(\vec{s}[1]) *  \degen(\sigma(\vec{s}[1]))(\vec{o}[1]) * \prod_{\kappa = 1}^{k-1} \tau(\vec{s}[\kappa], \vec{\imath}[\kappa])(\vec{s}[\kappa+1]) * \degen(\sigma(\vec{s}[\kappa+1]))(\vec{o}[\kappa+1]) \label{ln:ih}\\
&= \degen(s)(s') * \degen(\sigma(s))(o) * \tau(s, i)(\vec{s}[1]) * \degen(\vec{s}[1])(\vec{s}[1]) *  \degen(\sigma(\vec{s}[1]))(\vec{o}[1]) * \prod_{\kappa = 1}^{k-1} \tau(\vec{s}[\kappa], \vec{\imath}[\kappa])(\vec{s}[\kappa+1]) * \degen(\sigma(\vec{s}[\kappa+1]))(\vec{o}[\kappa+1])\label{ln:one}\\
&= \degen(s)(s') * \degen(\sigma(s))(o) * \tau(s, i)(\vec{s}[1]) * \degen(\sigma(\vec{s}[1]))(\vec{o}[1]) * \prod_{\kappa = 1}^{k-1} \tau(\vec{s}[\kappa], \vec{\imath}[\kappa])(\vec{s}[\kappa+1]) * \degen(\sigma(\vec{s}[\kappa+1]))(\vec{o}[\kappa+1])\\[1ex]
&= \begin{array}{ll}
  & \degen(s)(s') * \degen(\sigma(s))(o) * \tau(s, i)(s'\cons\vec{s}[1+1]) * \degen(\sigma(s'\cons\vec{s}[1+1]))(o\cons\vec{o}[1+1])\\
* & \prod\limits_{\kappa = 1}^{k-1} \tau(s'\cons\vec{s}[\kappa+1], i\cons\vec{\imath}[\kappa+1])(s'\cons\vec{s}[\kappa+1+1]) * \degen(\sigma(s'\cons\vec{s}[\kappa+1+1]))(o\cons\vec{o}[\kappa+1+1])\\
\end{array}\label{ln:something}\\[2ex]
&= \begin{array}{ll}
  & \degen(s)(s'\cons\vec{s}[1]) * \degen(\sigma(s'\cons\vec{s}[1]))(o\cons\vec{o}[1]) 
* \tau(s'\cons\vec{s}[1], i\cons\vec{\imath}[1])(s'\cons\vec{s}[1+1]) * \degen(\sigma(s'\cons\vec{s}[1+1]))(o\cons\vec{o}[1+1])\\
* & \prod\limits_{\kappa = 1}^{k-1} \tau(s'\cons\vec{s}[\kappa+1], i\cons\vec{\imath}[\kappa+1])(s'\cons\vec{s}[\kappa+1+1]) * \degen(\sigma(s'\cons\vec{s}[\kappa+1+1]))(o\cons\vec{o}[\kappa+1+1])\\
\end{array}\label{ln:ss}\\[2ex]
&= \begin{array}{ll}
  & \degen(s)(s'\cons\vec{s}[1]) * \degen(\sigma(s'\cons\vec{s}[1]))(o\cons\vec{o}[1]) \\
*\!\! & \tau(s'\cons\vec{s}[1], i\cons\vec{\imath}[1])(s'\cons\vec{s}[1+1]) * \degen(\sigma(s'\cons\vec{s}[1+1]))(o\cons\vec{o}[1+1])
* \prod\limits_{\kappa = 1+1}^{(k-1)+1} \tau(s'\cons\vec{s}[\kappa], i\cons\vec{\imath}[\kappa])(s'\cons\vec{s}[\kappa{+}1]) * \degen(\sigma(s'\cons\vec{s}[\kappa{+}1]))(o\cons\vec{o}[\kappa{+}1])\\
\end{array}\label{ln:reindex}\\
&= \degen(s)(s'\cons\vec{s}[1]) *  \degen(\sigma(s'\cons\vec{s}[1]))(o\cons\vec{o}[1]) * \prod_{\kappa = 1}^{k} \tau(s'\cons\vec{s}[\kappa], i\cons\vec{\imath}[\kappa])(s'\cons\vec{s}[\kappa+1]) * \degen(\sigma(s'\cons\vec{s}[\kappa+1]))(o\cons\vec{o}[\kappa+1]) \label{ln:suckup}
\end{align}
where \eqref{ln:ih} comes from the inductive hypothesis,
\eqref{ln:one} follows since $\degen(s'')(\vec{s}[1])$ will be $0$ for all other values of $s''$,
\eqref{ln:ss} follows since unless $s' = s$, the value will be zero due the $\degen(s)(s')$ term,
\eqref{ln:reindex} changes the indexing of the product so that \eqref{ln:suckup} can roll the two terms before the product into the product by starting the indexing from $1$ instead of $1+1$.
\end{proof}

\section{Universal Unsoundness and Incompleteness Proofs}

\subsection{Theorem~\ref{thm:unsound-inter}}
The theorem states:
\begin{quote}
Any black box analysis that ever returns a positive result from interference for $H$ to $L$ is unsound for interference from $H$ to $L$.
\end{quote}
\begin{proof}%
Assume that analysis $A$ can return a positive result for interference from interacting with a system.  Then, there must exist a system $q_\posres = \langle \mathcal{S}_\posres, s_{0\posres}, \mathcal{I}_\posres, \mathcal{O}_\posres, \tau_\posres, \sigma_\posres\rangle$ and $\vec{\imath}_\posres$ such that the output $q_\posres(\vec{\imath}_\posres)$ leads to $A$ returning a positive result.  
$q_\posres(\vec{\imath}_\posres)$ leads to a trace $[s_1, o_1, i_1, s_2, o_2, i_2, \ldots, s_k, o_k, i_{k}, s_{k+1}, o_{k+1}]$ where $s_1 = s_{0\posres}$, $o_j = \sigma_\posres(s_j)$, $i_j = \vec{\imath}_\posres[j]$, $s_j = \tau(s_{j-1}, i_{j-1})$, and $|\vec{\imath}_\posres| = k$.

Let $q_{\mathrm{N}}$ be a system that has noninterference but behaves like $q_\posres$ on $\vec{\imath}_\posres$.  That is, let $q_{\mathrm{N}}$ be $\langle \mathcal{S}_{\mathrm{N}}, s_{0{\mathrm{N}}}, \mathcal{I}_{\posres}, \mathcal{O}_{\posres}, \tau_{\mathrm{N}}, \sigma_{\mathrm{N}}\rangle$ where 
\begin{itemize}
\item $\mathcal{S}_{\mathrm{N}} = \{s^{\mathrm{N}}_1, \ldots, s^{\mathrm{N}}_k, s^{\mathrm{N}}_{k+1}\}$, 
\item $s_{0{\mathrm{N}}} = s^{\mathrm{N}}_1$,
\item $\tau_{\mathrm{N}}(s^{\mathrm{N}}_j, i) = s^{\mathrm{N}}_{j+1}$ for all $j \leq k$ and $\tau_{\mathrm{N}}(s^{\mathrm{N}}_{k+1}, i) = s^{\mathrm{N}}_{k+1}$ for all $i$, and 
\item $\sigma_{\mathrm{N}}(s^{\mathrm{N}}_j) = o_j$ for all $j \leq k + 1$.
\end{itemize}
Since the behavior of $q_{\mathrm{N}}$ does not depend upon any inputs, it has noninterference.  However, by construction, $q_{\mathrm{N}}(\vec{\imath}_\posres) = q_{\posres}(\vec{\imath}_\posres)$.  Thus, $A$ cannot tell them apart even with the ability to observe every input and output to the system.  Thus, it must produce an unsound positive result for interference on $q_{\mathrm{N}}$.
\end{proof}

\subsection{Theorem~\ref{thm:unsound-noninter}}
The theorem states:
\begin{quote}
Any black box analysis that ever returns a positive result for noninterference from $H$ to $L$ is unsound for noninterference from $H$ to $L$ if $H$ has two inputs and $L$ has two outputs.
\end{quote}
\begin{proof}%
Assume that $A$ can return a positive result from interacting with a system.  Then, there must exist a system $q_\negres$ and $\vec{\imath}_\negres$ such that the output $q_\negres(\vec{\imath}_\negres)$ lead $A$ to return positive for noninterference.  Let the trace of $q_\negres$ on $\vec{\imath}_\negres$ be $[s_1, o_1, i_1, s_2, o_2, \ldots]$. 

Let $q_{\mathrm{I}}$ be a system that has interference but behaves like $q_\negres$ on $\vec{\imath}_\negres$.  
That is, $q_{\mathrm{I}}$ be $\langle \mathcal{S}_{\mathrm{I}}, s_{0{\mathrm{I}}}, \mathcal{I}_{\negres}, \mathcal{O}_{\negres}, \tau_{\mathrm{I}}, \sigma_{\mathrm{I}}\rangle$ where 
\begin{itemize}
\item $\mathcal{S}_{\mathrm{I}} = \mathcal{S}_{\negres} \cup \{s_{00}, s_{01}\}$;

\item $s_{0{\mathrm{I}}} = s_{00}$;

\item $\tau_{\mathrm{I}}(s, i) = \tau_{\negres}(s, i)$ for all $i$ and $s$ in $\mathcal{S}_\negres$,
$\tau_{\mathrm{I}}(s_{00}, \vec{\imath}_\negres[1]) = \tau(s_0, \vec{\imath}_\negres[1])$,
$\tau_{\mathrm{I}}(s_{00}, i) = s_{01}$ for all $i \neq \vec{\imath}_\negres[1]$, and
$\tau_{\mathrm{I}}(s_{01}, i) = s_{01}$ for all $i$;

\item $\sigma_{\mathrm{I}}(s) = \sigma_{\negres}(s)$ for all $s$ other than $s_{01}$ and 
  $\sigma_{\mathrm{I}}(s_{01}) = o_{01}$ where $o_{01} \neq o_2 = \sigma_{\negres}(s_2)$.
\end{itemize}
Note that since $|\mathcal{O}| \geq 2$, such an $o_{01}$ exists, and since $|\mathcal{I}| \geq 2$, an $i \neq \vec{\imath}_\negres[1]$ exists making $s_{01}$ reachable.

The behavior of $q_{\mathrm{I}}$ at the state $s_{01}$ versus $s_1$ shows that it has interference when we consider an input $i$ such that $i \neq \vec{\imath}_\negres[1]$ and $i$ differs from $\vec{\imath}_\negres[1]$ by just high-level information.  However, by construction, $q_{\mathrm{I}}(\vec{\imath}_\negres) = q_{\negres}(\vec{\imath}_\negres)$.  Thus, $A$ cannot tell them apart  even with the ability to observe every input and output to the system.  Thus, it must produce an unsound result for $q_{\mathrm{I}}$ having noninterference.
\end{proof}

\section{Background: Causality}

In this section, we review Pearl's formalism of causality~\cite{pearl00book}.  In particular, we use notation and results found in Chapters~1 and~7 of~\cite{pearl00book}.

\paragraph{Background on Probability}
Recall that for any two propositions $A_1$ and $A_2$, $\mathcal{P}(A_1 \land A_2) = \mathcal{P}(A_2 \given A_1)*\mathcal{P}(A_1)$ if $\mathcal{P}(A_1) > 0$.  
If $\mathcal{P}(A_1) = 0$, then $\mathcal{P}(A_2 \given A_1)$ is not defined.
We adopt the convention that the product of a undefined term by zero will be zero (which is similar to~\cite{knuth92maa}).  Under this convention, $\mathcal{P}(A_1 \land A_2) = \mathcal{P}(A_2 \given A_1)*\mathcal{P}(A_1)$ holds in general.
Under this convention, the chain rule of probability iterates the above equation:
\begin{align}
\mathcal{P}(\land_{j=1}^{J} A_j) &= \prod_{j=1}^J \mathcal{P}(A_j \given \land_{k=1}^{j-1} A_k)
\end{align}

\paragraph{SEMs}
Recall that a probabilistic SEM $M$ is a tuple $\langle \mathcal{V}_{\msf{en}}, \mathcal{V}_{\msf{ex}}, \mathcal{E}, \mathcal{P}\rangle$ where $\mathcal{V}_{\msf{en}}$ is the endogenous variables,  $\mathcal{V}_{\msf{ex}}$ is the exogenous variables, $\mathcal{E}$ provides a \emph{structural equation} for each endogenous variable $V$,
 and $\mathcal{P}$ is a probability distribution.  

To define $\mathcal{E}$ in more detail,
let the space of functions $\mathcal{F}_{V}$ be (possibly randomized) functions from the ranges of a subset of the variables other than $V$ to the range of $V$.  
$\mathcal{E}$ maps a variable $V$ in $\mathcal{V}_{\msf{en}}$ to a function in $\mathcal{F}_V$.
If $V$ is mapped to a function $F_V$ that does not include the range of the variable $V'$, then $V$ does not have a direct dependence upon $V'$.
We write $V := F_V(\vec{V})$ where $\vec{V}$ is a list of other variables not equal to $V$ if $\mathcal{E}$ maps $V$ to a function $F_V$ that directly depends upon the variables $\vec{V}$.  Let $\pa(V)$ denote the variables $\vec{V}$, called the \emph{parents} of $V$.  Let $\pa(V)$ be the empty set for exogenous variables $V$.

To define $\mathcal{P}$ in more detail, let $\mathcal{P}$ map each exogenous variable $V$ to a probability distribution $\mathcal{P}_V$ over the range of $V$.  Note that exogenous variables are assumed to be independent and, thus, these marginal distributions suffice for explaining their behavior.

We call a SEM \emph{recursive} if the graph of variables created by including a directed edge from every parent to every child variable (node) is acyclic.  We will limit our discuss to recursive SEMs.  We will implicitly order their variables by the topology created by this graph.

\paragraph{Assigning Probabilities: Factorization}
We can use the topological ordering on the variables to extend to $\mathcal{P}$ to assign probabilities to assignments of values to variables.
To do so, we define some notation.
For a vector $\vec{V}$, we use $\vec{V}[j]$ to denote its $j$th component.
We take $\vec{V}=\vec{v}$ be shorthand for $\bigwedge_{j=1}^{t} \vec{V}[j]=\vec{v}[j]$ where $\vec{V}$ is a vector of length $t$ holding variables.
Similarly, let $\vec{V}^{j:k}=\vec{v}$ be shorthand for $\bigwedge_{t=j}^{k} \vec{V}[t]=\vec{v}[t]$.
We use $\pa(V)=\vec{w}$ as sort hand for $\bigwedge_{W_j \in \pa(V)} W_j = \vec{w}[j]$ where there is some implicit ordering on variables associating the $j$th element of $\pa(V)$ to the $j$th component of $\vec{w}$.

We start by assigning a probability to a variable given its parents in the SEM $M$.  For exogenous variables $V$,
let $\mathcal{P}^M(V = v \given \pa(V) = \vec{v})$ be $\mathcal{P}_V(v)$.  (Recall that $\pa(V)$ is the empty set for exogenous variables.  Thus, the vector $\vec{v}$ of values is empty as well.)
For endogenous variables $V$ defined by a deterministic function $f_V$, let $\mathcal{P}^M(V = v \given \pa(V) = \vec{v})$ be $1$ if $v = f_V(\vec{v})$ and be $0$ otherwise.  For randomized functions $F_V$, let $\mathcal{P}^M(V = v \given \pa(V) = \vec{v})$ be the probability that $v = F_V(\vec{v})$.

For a vector of all the variables $\vec{V}$ and a vector of values $\vec{v}$ they can take on, we determine $\mathcal{P}^M(\vec{V}{=}\vec{v})$ using a \emph{factorization} created by the chain rule:
\begin{align}
\mathcal{P}^M(\vec{V}{=}\vec{v}) 
&= \prod_{j=1}^{|\vec{V}|} \mathcal{P}^M(\vec{V}[j]{=}\vec{v}[j] \given \vec{V}^{1:j-1} = \vec{v}^{1:j-1})\\
&= \prod_{j=1}^{|\vec{V}|} \mathcal{P}^M(\vec{V}[j]{=}\vec{v}[j] \given \pa(\vec{V}[j]) = \vec{v}_{\pa(\vec{V}[j])}) \label{ln:fac}
\end{align}
where $j$ ranges over $\vec{V}$ in a manner that respects the variables' topology,
$\vec{v}_{\pa(\vec{V}[j])}$ is $\vec{v}$ restricted to just these components corresponding to elements of $\pa(\vec{V}[j])$, and we take $\pa(V)$ to be the empty set for exogenous variables $V$.
\eqref{ln:fac} follows since $\vec{V}^{1:j-1} = \vec{v}^{1:j-1}$ includes all the parents of $\vec{V}[j]$ by using the topological ordering and $\vec{V}[j]$ is independent of its non-parents given its parents.

For $\vec{W} = \vec{w}$ involving a subset of the variables, we use the following:
\begin{align}
\mathcal{P}^M(\vec{W}{=}\vec{w}) 
&= \sum_{\vec{u}} \mathcal{P}^M(\vec{W} = \vec{w}, \vec{U} = \vec{u})
= \sum_{\vec{u}} \prod_{j=1}^{|\vec{V}|} \mathcal{P}^M(\vec{V}[j]{=}\vec{v}[j] \given \pa(\vec{V}[j]) = v(\vec{w},\vec{u})_{\pa(\vec{V}[j])}) \label{ln:sum-over-nus}
\end{align}
where $\vec{U}$ are the remaining variables, $\vec{V}$ is a vector consisting of the components of $\vec{W}$ and $\vec{U}$ put into order, and
$v(\vec{w},\vec{u})$ is the vector $\vec{v}$ that results from combining the components of $\vec{w}$ and $\vec{u}$ in order.

\paragraph{Sub-Models and Truncated Factorization}

Recall that for an SEM $M$, endogenous variable $X$, and value $x$ that $X$ can take on,
the \emph{sub-model} $M[X{:=}x]$ is the SEM that results from replacing the equation $X := F_X(\vec{V})$ in $\mathcal{E}$ with the equation $X := x$.  
That is, for $M = \langle \mathcal{V}_{\msf{en}}, \mathcal{V}_{\msf{ex}}, \mathcal{E}, \mathcal{P}\rangle$, 
$M[X{:=}x] = \langle \mathcal{V}_{\msf{en}}, \mathcal{V}_{\msf{ex}}, \mathcal{E}[X := x], \mathcal{P}\rangle$ where 
$\mathcal{E}[X := x](X) = \lambda . x$ (the function that takes no arguments and always returns $x$) and 
$\mathcal{E}[X := x](V) = \mathcal{E}(V)$ for $V \neq X$.

$\mathcal{P}^M$ and $\mathcal{P}^{M[X{:=x}]}$ are related by \emph{truncated factorization}.  To define it, let $X$ be the $k$th variable in the topological order.  For $\vec{V}{=}\vec{v}$ that assigns $X$ the value $x$ (i.e., $\vec{v}[k] = x$),
\begin{align}
\mathcal{P}^{M[X{:=}x]}(\vec{V}{=}\vec{v}) 
&\:\:=\:\: \prod_{j=1}^{|\vec{V}|} \mathcal{P}^{M[X{:=}x]}(\vec{V}[j]{=}\vec{v}[j] \given \pa(\vec{V}[j]) = \vec{v}_{\pa(\vec{V}[j])}) 
&= \prod_{j=1 : j \neq k}^{|\vec{V}|} \mathcal{P}^{M}(\vec{V}[j]{=}v(\vec{w},\vec{u})[j] \given \pa(\vec{V}[j]) = \vec{v}_{\pa(\vec{V}[j])}) \label{ln:skip}
\end{align}
where the produce in \eqref{ln:skip} skips $X$, the $k$th variable.
For $\vec{V}{=}\vec{v}$ that assigns $X$ a value other than $x$, $\mathcal{P}^{M[X{:=}x]}(\vec{V}{=}\vec{v})$ is $0$.
The above extends to subsets of all variables as in \eqref{ln:sum-over-nus}.

Henceforth, for readability, we adopt Pearl's \emph{do} notation.  We will drop the $M$ from $\mathcal{P}^M$ when $M$ is clear from context.  We will denote $\mathcal{P}^{M[X{:=}x]}(\vec{V}{=}\vec{v})$ as $\mathcal{P}(\vec{V}{=}\vec{v} \given \doo(X{:=}x))$.
We take $\doo(\vec{X} := \vec{x})$ be shorthand for $\bigwedge_{j=1}^{|\vec{X}|} \doo(\vec{X}[j]:=\vec{x}[j])$.  We understand $\mathcal{P}(\vec{V}{=}\vec{v} \given \doo(\vec{X}{:=}\vec{x}))$ to be iterative application of taking a sub-model with
\begin{align}
\mathcal{P}(\vec{V}{=}\vec{v} \given \doo(\vec{X}{:=}\vec{x}))
&= \prod_{j=1 : j \notin K}^{|\vec{V}|} \mathcal{P}(\vec{V}[j]{=}\vec{v}[j] \given \pa(\vec{V}[j]) = \vec{v}_{\pa(\vec{V}[j])}) 
\end{align}
where $K$ is the set containing the indexes of the variables in $\vec{X}$.

Pearl presents two useful properties~\cite[pg 24]{pearl00book}.  The first allows converting normal conditional statements to \emph{do} statements when conditioning upon all of a variable's parents.  The second allows for dropping irreverent \emph{do} statements when conditioning upon all of a variable's parents.

\begin{lemma}[Pearl's Property~1]
\[ \mathcal{P}(Y{=}y \given \pa(Y){=}\vec{x})  
\:\:=\:\: \mathcal{P}(Y{=}y \given \doo(\pa(Y){:=}\vec{x})) \]
\end{lemma}

\begin{lemma}[Pearl's Property~2]
\[ \mathcal{P}(Y{=}y \given \doo(\pa(Y){:=}\vec{x}), \doo(\vec{Z}{:=}\vec{z}))
\:\:=\:\: \mathcal{P}(Y{=}y \given \doo(\pa(Y){:=}\vec{x})) \]
\end{lemma}

\section{Interference and Causation}
\label{app:connection}

\subsection{Model}

Given a probabilistic Moore Machine $Q$, we define a SEM $M_Q$ of $Q$.  Intuitively, it contains endogenous variables for each input and output and exogenous variables for each user.  The behavior of $Q$ provides functions $F_{\mathsf{lo},t}$ defining the low output at time $t$ in terms of the previous and current inputs.
In more detail, for each time $t$, we create the endogenous variables
$\mathsf{HI}_t$,
$\mathsf{HO}_t$,
$\mathsf{LI}_t$, and
$\mathsf{LO}_t$ for the high input and output, and low input and output, respectively, at the time $t$.
We add exogenous variables $\mathsf{HU}_t$ and $\mathsf{LU}_t$ that represents the behavior of high and low users of the system at time $t$.

For a indexed family of variables $\vec{V}$, we use $\vec{V}^t$ to denote the vector holding those with an index of $t$ or less (in order).
That is, $\vec{V}^{t}=\vec{v}$ be shorthand for $\bigwedge_{j=1}^{t} \vec{V}[j]=\vec{v}[j]$.

The following table shows how we define these variables and functions:
\begin{tab}{@{}lllll@{}}
$V$               &           & $\pa(V)$                                               & $F_V$ \\
\midrule 
$\mathsf{HU}_{t+1}$   & high user   & $\emptyset$                                                 & (exogenous) & for all $t \geq 0$\\
$\mathsf{LU}_{t+1}$   & low user    & $\emptyset$                                                 & (exogenous) & for all $t \geq 0$\\
$S_{0}$           & initial state       &  $\emptyset$  & $F_{\mathsf{s},0}() = \degen(s_0)$\\
$S_{t+1}$           & state       &  $\{S_{t}, \mathsf{HI}_{t}, \mathsf{LI}_{t} \}$  & $F_{\mathsf{s},t+1}(s_{t}, \mathsf{hi}_{t}, \mathsf{li}_{t})(s') = \tau(s_{t}, \langle \mathsf{hi}_{t}, \mathsf{li}_{t})\rangle)$ & for all $t \geq 0$\\
$\mathsf{HI}_{t+1}$ & high input  & $\{ \mathsf{HU}_{t+1}, \mathsf{LU}_{t+1}, \mathsf{HO}_1, \ldots, \mathsf{HO}_{t}, \mathsf{LO}_1, \ldots, \mathsf{LO}_{t} \}$ & $F_{\mathsf{hi},t+1}(\mathsf{HU}_{t+1}, \msf{LU}_{t+1}, \vec{\mathsf{HO}}^{t}, \vec{\msf{LO}}^t)$ & for all $t \geq 0$\\
$\mathsf{LI}_{t+1}$ & low input   & $\{ \mathsf{HU}_{t+1}, \mathsf{LU}_{t+1}, \mathsf{HO}_1, \ldots, \mathsf{HO}_{t}, \mathsf{LO}_1, \ldots, \mathsf{LO}_{t} \}$ & $F_{\mathsf{li},t+1}(\mathsf{HU}_{t+1}, \mathsf{LU}_{t+1}, \vec{\msf{HO}}^{t}, \vec{\mathsf{LO}}^{t})$ & for all $t \geq 0$\\
$\mathsf{HO}_t$ & high output & $\{ S_t \}$ & $F_{\mathsf{ho},t}(s_{t}) = \degen(\filter{\sigma(s_t)}{H})$ & for all $t \geq 0$\\
$\mathsf{LO}_t$ & low input   & $\{ S_t \}$ & $F_{\mathsf{lo},t}(s_{t}) = \degen(\filter{\sigma(s_t)}{L})$ & for all $t \geq 0$\\
\end{tab}

The form of $\mathcal{P}(V{=}v \given \pa(V))$ depends upon the type of variable that $V$ is.  Here are the options based on the above table:
\begin{tab}{@{}lllll@{}}
$V$                 & $\mathcal{P}(V{=}v \given \pa(V))$  &  &\\
\midrule 
$\mathsf{HU}_{t+1}$ & $\mathcal{P}(\mathsf{HU}_{t+1}{=}\msf{hu}_{t+1})$  & & & for all $t \geq 0$\\
$\mathsf{LU}_{t+1}$ & $\mathcal{P}(\mathsf{LU}_{t+1}{=}\msf{lu}_{t+1})$  & & & for all $t \geq 0$\\
$S_{0}$             & $\mathcal{P}(S_{0}{=}s)$                          & $=$ & $\degen(s_0)(s)$ &\\
$S_{t+1}$           & $\mathcal{P}(S_{t+1}=s_{t+1} \given S_{t}{=}s_{t}, \mathsf{HI}_{t}{=}\msf{hi}_{t}, \mathsf{LI}_{t}{=}\msf{li}_{t})$  &$=$&  $\tau(s_{t}, \langle \mathsf{hi}_{t}, \mathsf{li}_{t})\rangle)(s_{t+1})$ & for all $t \geq 0$\\
$\mathsf{HI}_{t+1}$ & $\mathcal{P}(\mathsf{HI}_{t+1}{=}\msf{hi}_{t+1} \given \mathsf{HU}_{t+1}{=}\msf{hu}_{t+1}, \mathsf{LU}_{t+1}{=}\msf{lu}_{t+1}, \vec{\mathsf{HO}}^{t}{=}\vec{\msf{ho}}, \vec{\mathsf{LO}}^{t}{=}\vec{\msf{lo}})$ & $=$ & $F_{\mathsf{hi},t+1}(\mathsf{hu}_{t+1}, \mathsf{lu}_{t+1}, \vec{\mathsf{ho}}, \vec{\mathsf{lo}})(\msf{hi}_{t+1})$ & for all $t \geq 0$\\
$\mathsf{LI}_{t+1}$ &  $\mathcal{P}(\mathsf{LI}_{t+1}{=}\msf{li}_{t+1} \given  \mathsf{HU}_{t+1}{=}\msf{hu}_{t+1}, \mathsf{LU}_{t+1}{=}\msf{lu}_{t+1}, \vec{\mathsf{HO}}^{t}{=}\vec{\msf{ho}}, \vec{\mathsf{LO}}^{t}{=}\vec{\msf{lo}})$ & $=$ & $F_{\mathsf{li},t+1}(\mathsf{hu}_{t+1}, \mathsf{lu}_{t+1}, \vec{\mathsf{ho}}, \vec{\mathsf{lo}})(\msf{li}_{t+1})$ & for all $t \geq 0$\\
$\mathsf{HO}_t$ &  $\mathcal{P}(\mathsf{HO}_{t}{=}\mathsf{ho}_t \given S_t{=}s_t)$ &$=$& $\degen(\filter{\sigma(s_t)}{H})(\msf{ho}_t)$ & for all $t \geq 0$\\
$\mathsf{LO}_t$ &  $\mathcal{P}(\mathsf{LO}_{t}{=}\msf{lo}_t \given S_t{=}s_t)$    &$=$&  $\degen(\filter{\sigma(s_t)}{L})(\msf{lo}_t)$ & for all $t \geq 0$\\
\end{tab}

Let $M_Q$ consist of the variables and equations defined above plus an unknown probability distribution $\mathcal{P}$.

\subsection{Relation of Models}

Let $\vec{V}^{j:k}=\vec{v}$ be shorthand for $\bigwedge_{t=j}^{k} \vec{V}[t]=\vec{v}[t]$.
Let $\doo(V:=v)$ be Pearl's \emph{do} operation denoting an intervention fixing a value, such as by applying a treatment to an experimental unit~\cite{pearl00book}.
Let $\doo(\vec{V}^{j:k}:=\vec{v})$ be short hand for $\bigwedge_{t = j}^{k} \doo(\vec{V}[t] := \vec{v}[t])$.
Let $\vec{O}^{j:k} = \vec{o}$ be shorthand for $\filter{\msf{HO}^{j:k}}{H} = \filter{\vec{o}}{H} \land \filter{\msf{LO}^{j:k}}{L} = \filter{\vec{o}}{L}$.
Let $\vec{I}^{j:k} = \vec{\imath}$ be shorthand for $\filter{\msf{HI}^{j:k}}{H} = \filter{\vec{\imath}}{H} \land \filter{\msf{LI}^{j:k}}{L} = \filter{\vec{\imath}}{L}$.
We define $\doo(\vec{O}^{j:k} := \vec{o})$ and $\doo(\vec{I}^{j:k} := \vec{\imath})$ similarly.

We define the equivalent of $Q(s, \vec{\imath})(\vec{s},\vec{o})$ for an SEM $M_Q$ as follows:
let 
\begin{align}
\msf{fix}^{t}(M_Q)(s, \vec{\imath})(\vec{s},\vec{o}) &\quad=\quad \mathcal{P}(\vec{S}^{t:t+k} {=} \vec{s} \land \vec{O}^{t:t+k} {=} \vec{o} \:\given\: \doo(S_t{:=}s),\, \doo(\vec{I}^{t:t+k-1} {:=} \vec{\imath}))
\end{align}
where $\vec{\imath}$, $\vec{o}$, and $\vec{s}$ are of lengths $k \geq 0$, $k+1$, and $k+1$, respectively.
The time $t \geq 0$ represents the time at which $M_Q$ starts operating.  
Note that when $k = 0$, $\vec{I}^{t:t+k-1}$ is $\vec{I}^{t:t-1}$, which is an empty sequence, as is $\vec{\imath}$.  Thus, $\doo(\vec{I}^{t:t+k-1} {:=} \vec{\imath})$ is vacuously true when $k = 0$.  On the other hand, $\vec{S}^{t:t+k}$ is $\vec{S}^{t:t} = [\vec{S}_t]$, a sequence with a single component, which is compared to the single component of $\vec{s} = [s_1]$.

\begin{lemma}\label{lem:tc.fix-out}
For all $Q$, $s$, and $t \geq 0$, and $\vec{\imath}$, $\vec{o}$, and $\vec{s}$ of lengths $k \geq 0$, $k+1$, and $k+1$, respectively,
\[ \msf{fix}^{t}(M_Q)(s, \vec{\imath})(\vec{s}, \vec{o}) = Q(s,\vec{\imath})(\vec{o},\vec{s}) \]
\end{lemma}
\nopagebreak
\begin{proof}
\begin{align}
&\msf{fix}^{t}(M_Q)(s, \vec{\imath})(\vec{s}, \vec{o}) \\
&= \mathcal{P}(\vec{S}^{t:t+k} {=} \vec{s} \,\land\, \vec{O}^{t:t+k} {=} \vec{o} \:\given\: \doo(S_t{:=}s),\, \doo(\vec{I}^{t:t+k-1} {:=} \vec{\imath}))\\ %
&= \mathcal{P}(\bigwedge_{\kappa = 0}^{k} \vec{S}[t+\kappa] {=} \vec{s}[1+\kappa] \,\land\, \vec{O}[t+\kappa] {=} \vec{o}[1+\kappa] \:\given\: \doo(S_t{:=}s),\, \doo(\vec{I}^{t:t+k-1} {:=} \vec{\imath})) \label{ln:tc.seq-land}\\
&= \prod_{\kappa = 0}^{k} \mathcal{P}(\vec{S}[t+\kappa]{=}\vec{s}[1+\kappa] \,\land\, \vec{O}[t+\kappa]{=}\vec{o}[1+\kappa] \:\given\: \vec{S}^{t:t+\kappa-1}{=}\vec{s}^{1:\kappa},\, \vec{O}^{t:t+\kappa-1}{=}\vec{o}^{1:\kappa},\, \doo(S_t{:=}s),\, \doo(\vec{I}^{t:t+k-1}{:=}\vec{\imath})) \label{ln:tc.chain1}\\
&=  \prod_{\kappa = 0}^{k} \begin{array}{l}
\mathcal{P}(\vec{S}[t+\kappa]{=}\vec{s}[1+\kappa] \:\given\: \vec{S}^{t:t+\kappa-1}{=}\vec{s}^{1:\kappa}, \vec{O}^{t:t+\kappa-1}{=}\vec{o}^{1:\kappa}, \doo(S_t{:=}s), \doo(\vec{I}^{t:t+k-1}{:=}\vec{\imath}))\\
* \mathcal{P}(\vec{O}[t+\kappa]{=}\vec{o}[1+\kappa] \:\given\: \vec{S}[t+\kappa]{=}\vec{s}[1+\kappa], \vec{S}^{t:t+\kappa-1}{=}\vec{s}^{1:\kappa}, \vec{O}^{t:t+\kappa-1}{=}\vec{o}^{1:\kappa}, \doo(S_t{:=}s), \doo(\vec{I}^{t:t+k-1}{:=}\vec{\imath}))
\end{array}
\label{ln:tc.chain2}%
\end{align}
where
\eqref{ln:tc.seq-land} expands $\vec{S}^{t:t+k} {=} \vec{s} \,\land\, \vec{O}^{t:t+k} {=} \vec{o}$ into $\bigwedge_{\kappa = 0}^{k} \vec{S}[t+\kappa] {=} \vec{s}[1+\kappa] \,\land\, \vec{O}[t+\kappa] {=} \vec{o}[1+\kappa]$.  Since $\kappa$ ranges from $0$ to $k$ while we index the sequences $\vec{s}$ and $\vec{o}$ from $1$ to $k+1$, we add $1$ to $\kappa$ while indexing into $\vec{s}$ and $\vec{o}$.
Both \eqref{ln:tc.chain1} and \eqref{ln:tc.chain2} follow from the chain rule of probability.  Note that 
when $\kappa$ is $0$, the term
$\vec{S}^{t:t+\kappa-1}{=}\vec{s}^{1:\kappa}$ becomes $\vec{S}^{t:t-1}{=}\vec{s}^{1:0}$, which compares the empty sequence to the empty sequence.  This comparison is vacuously true as it should be since no state precedes the first state $\vec{s}[1+\kappa] = \vec{s}[1+0] = \vec{s}[1]$ and, thus, the probability of this state should not be conditioned on a preceding state.  The same holds for the output $\vec{o}[1]$.

In \eqref{ln:tc.chain2}, $\mathcal{P}(\vec{S}[t+\kappa]{=}\vec{s}[1+\kappa] \:\given\: \vec{S}^{t:t+\kappa-1}{=}\vec{s}^{1:\kappa}, \vec{O}^{t:t+\kappa-1}{=}\vec{o}^{1:\kappa}, \doo(S_t{:=}s), \doo(\vec{I}^{t:t+k-1}{:=}\vec{\imath}))$ is looking at the probability of $\vec{S}[t+\kappa]{=}\vec{s}[1+\kappa]$ conditional upon every term on which $\vec{S}[t+\kappa]$ depends in the model $M_Q$ (i.e., all of the variables in $\pa(S_{t+\kappa})$).  The same holds for the outputs $\vec{O}[t+\kappa]$.
Thus, Pearl's Property~1~\cite[pg 24]{pearl00book} applies to \eqref{ln:tc.chain2} and justifies \eqref{ln:tc.prop1} in the following:%
\begin{align}
&\msf{fix}^{t}(M_Q)(s, \vec{\imath})(\vec{s}, \vec{o}) \\
&=  \prod_{\kappa = 0}^{k} \begin{array}{l}
\mathcal{P}(\vec{S}[t+\kappa]{:=}\vec{s}[1+\kappa] \:\given\: \doo(\vec{S}^{t:t+\kappa-1}{:=}\vec{s}^{1:\kappa}), \doo(\vec{O}^{t:t+\kappa-1}{:=}\vec{o}^{1:\kappa}), \doo(S_t{:=}s), \doo(\vec{I}^{t:t+k-1}{:=}\vec{\imath}))\\
* \mathcal{P}(\vec{O}[t+\kappa]{=}\vec{o}[1+\kappa] \:\given\: \doo(\vec{S}[t+\kappa]{:=}\vec{s}[1+\kappa]), \doo(\vec{S}^{t:t+\kappa-1}{:=}\vec{s}^{1:\kappa}), \doo(\vec{O}^{t:t+\kappa-1}{:=}\vec{o}^{1:\kappa}), \doo(S_t{:=}s), \doo(\vec{I}^{t:t+k-1}{:=}\vec{\imath}))
\end{array}
\label{ln:tc.prop1}\\[1ex]
&=  
\begin{array}{l}
\mathcal{P}(\vec{S}[t+0]{:=}\vec{s}[0+1] \:\given\: \doo(\vec{S}^{t:t+0-1}{:=}\vec{s}^{1:0}), \doo(\vec{O}^{t:t+0-1}{:=}\vec{o}^{1:0}), \doo(S_t{:=}s), \doo(\vec{I}^{t:t+k-1}{:=}\vec{\imath}))\\
* \mathcal{P}(\vec{O}[t+0]{=}\vec{o}[0+1] \:\given\: \doo(\vec{S}[t+0]{:=}\vec{s}[0+1]), \doo(\vec{S}^{t:t+0-1}{:=}\vec{s}^{1:0}), \doo(\vec{O}^{t:t+0-1}{:=}\vec{o}^{1:0}), \doo(S_t{:=}s), \doo(\vec{I}^{t:t+k-1}{:=}\vec{\imath}))\\
* \prod\limits_{\kappa = 1}^{k} \begin{array}{l}
\mathcal{P}(\vec{S}[t+\kappa]{:=}\vec{s}[1+\kappa] \:\given\: \doo(\vec{S}^{t:t+\kappa-1}{:=}\vec{s}^{1:\kappa}), \doo(\vec{O}^{t:t+\kappa-1}{:=}\vec{o}^{1:\kappa}), \doo(S_t{:=}s), \doo(\vec{I}^{t:t+k-1}{:=}\vec{\imath}))\\ 
* \mathcal{P}(\vec{O}[t+\kappa]{=}\vec{o}[1+\kappa] \:\given\: \doo(\vec{S}[t+\kappa]{:=}\vec{s}[1+\kappa]), \doo(\vec{S}^{t:t+\kappa-1}{:=}\vec{s}^{1:\kappa}), \doo(\vec{O}^{t:t+\kappa-1}{:=}\vec{o}^{1:\kappa}), \doo(S_t{:=}s), \doo(\vec{I}^{t:t+k-1}{:=}\vec{\imath}))
\end{array}
\end{array}
\label{ln:tc.special}\\[1ex]
&=  
\begin{array}{l}
\mathcal{P}(\vec{S}[t]{:=}\vec{s}[1] \:\given\: \doo(S_t{:=}s), \doo(\vec{I}^{t:t+k-1}{:=}\vec{\imath}))
 * \mathcal{P}(\vec{O}[t]{=}\vec{o}[1] \:\given\: \doo(\vec{S}[t]{:=}\vec{s}[1]), \doo(S_t{:=}s), \doo(\vec{I}^{t:t+k-1}{:=}\vec{\imath}))\\
* \prod\limits_{\kappa = 1}^{k} \begin{array}{l}
  \mathcal{P}(\vec{S}[t+\kappa]{:=}\vec{s}[1+\kappa] \:\given\: \doo(\vec{S}^{t:t+\kappa-1}{:=}\vec{s}^{1:\kappa}), \doo(\vec{O}^{t:t+\kappa-1}{:=}\vec{o}^{1:\kappa}), \doo(S_t{:=}s), \doo(\vec{I}^{t:t+k-1}{:=}\vec{\imath}))\\
* \mathcal{P}(\vec{O}[t+\kappa]{=}\vec{o}[1+\kappa] \:\given\: \doo(\vec{S}[t+\kappa]{:=}\vec{s}[1+\kappa]), \doo(\vec{S}^{t:t+\kappa-1}{:=}\vec{s}^{1:\kappa}), \doo(\vec{O}^{t:t+\kappa-1}{:=}\vec{o}^{1:\kappa}), \doo(S_t{:=}s), \doo(\vec{I}^{t:t+k-1}{:=}\vec{\imath}))
\end{array}
\end{array}
\label{ln:tc.simp}\\[1ex]
&=  
\begin{array}{l}
\mathcal{P}(\vec{S}[t]{:=}\vec{s}[1] \:\given\: \doo(S_t{:=}s), \doo(\vec{I}^{t:t+k-1}{:=}\vec{\imath}))
* \mathcal{P}(\vec{O}[t]{=}\vec{o}[1] \:\given\: \doo(\vec{S}[t]{:=}\vec{s}[1]))\\
* \prod\limits_{\kappa = 1}^{k} \begin{array}{l}
\mathcal{P}(\vec{S}[t+\kappa]{:=}\vec{s}[1+\kappa] \:\given\: \doo(\vec{S}[t+\kappa-1]{:=}\vec{s}[\kappa]), \doo(\vec{I}[t+\kappa-1]{:=}\vec{\imath}[\kappa]))\\
* \mathcal{P}(\vec{O}[t+\kappa]{=}\vec{o}[1+\kappa] \:\given\: \doo(\vec{S}[t+\kappa]{:=}\vec{s}[1+\kappa]))  
\end{array}
\end{array}
\label{ln:tc.prop2}\\[1ex]
&= 
\degen(s)(\vec{s}[1]) * \degen(\sigma(\vec{s}[1]))(\vec{o}[1]) *
\prod_{\kappa = 1}^{k} \tau(\vec{s}[\kappa], \vec{\imath}[\kappa])(\vec{s}[1+\kappa]) * \degen(\sigma(\vec{s}[1+\kappa]))(\vec{o}[1+\kappa])\label{ln:tc.model-cons}\\
&= Q(s,\vec{\imath})(\vec{o},\vec{s}) \label{ln:tc.q-out}
\end{align}
where
\eqref{ln:tc.special} simply pulls out the case where $\kappa = 0$;
\eqref{ln:tc.simp} just removes terms that are vacuously true;
\eqref{ln:tc.prop2} follows from Pearl's Property~2~\cite[pg 24]{pearl00book}, which removes \emph{do} terms that are not parents in $M_Q$ of the term whose probability we are computing;
\eqref{ln:tc.model-cons} comes from how we construct the model $M_Q$;
and \eqref{ln:tc.q-out} comes from Lemma~\ref{lem:q-out}.
\end{proof}

\begin{lemma}\label{lem:tc.eq-filtered}
For all $Q$, $\vec{\imath}$, and $\vec{\msf{lo}}$ of lengths $t$ and $t+1$, respectively,
$\mathcal{P}(\vec{\msf{LO}}^{1:t+1}{=}\vec{\msf{lo}} \:\given\: \doo(\vec{I}^{1:t}{:=}\vec{\imath}))
= \filter{Q(\vec{\imath})}{L}(\vec{\msf{lo}})$.
\end{lemma}
\nopagebreak
\begin{proof}
\begin{align}
\mathcal{P}(\vec{\msf{LO}}^{1:t+1}{=}\vec{\msf{lo}} \:\given\: \doo(\vec{I}^{1:t}{:=}\vec{\imath}))
&= \sum_{\vec{o} : \filter{\vec{o}}{L}{=}\vec{\msf{lo}}} \mathcal{P}(\vec{O}^{1:t+1}{=}\vec{o} \:\given\: \doo(\vec{I}^{1:t}{:=}\vec{\imath}))\label{ln:tc.fil.outputs-mutex}\\
&= \sum_{\vec{s} \in \mathcal{S}^{t+1}} \sum_{\vec{o} : \filter{\vec{o}}{L}{=}\vec{\msf{lo}}} \mathcal{P}(\vec{S}^{1:t+1} = \vec{s} \land \vec{O}^{1:t+1}{=}\vec{o} \:\given\: \doo(\vec{I}^{1:t}{:=}\vec{\imath}))\label{ln:tc.fil.states-mutex}\\
&= \sum_{\vec{s} \in \mathcal{S}^{t+1}} \sum_{\vec{o} : \filter{\vec{o}}{L}{=}\vec{\msf{lo}}} \mathcal{P}(\vec{S}^{1:t+1} = \vec{s} \land \vec{O}^{1:t+1}{=}\vec{o} \:\given\: S_0{=}s_0, \doo(\vec{I}^{1:t}{:=}\vec{\imath}))\label{ln:tc.fil.start}\\
&= \sum_{\vec{s} \in \mathcal{S}^{t+1}} \sum_{\vec{o} : \filter{\vec{o}}{L}{=}\vec{\msf{lo}}} \mathcal{P}(\vec{S}^{1:t+1} = \vec{s} \land \vec{O}^{1:t+1}{=}\vec{o} \:\given\: \doo(S_0{:=}s_0), \doo(\vec{I}^{1:t}{:=}\vec{\imath}))\label{ln:tc.fil.prop1}\\
&= \sum_{\vec{s} \in \mathcal{S}^{t+1}} \sum_{\vec{o} : \filter{\vec{o}}{L}{=}\vec{\msf{lo}}} Q(s_0,\vec{\imath})(\vec{o},\vec{s})\label{ln:tc.fil.lemma}\\
&= \sum_{\vec{o} : \filter{\vec{o}}{L}{=}\vec{\msf{lo}}} Q(\vec{\imath})(\vec{o})\\
&= \filter{Q(\vec{\imath})}{L}(\vec{\msf{lo}})
\end{align}
where 
\eqref{ln:tc.fil.outputs-mutex} and \eqref{ln:tc.fil.states-mutex} hold since output sequences and state sequences are mutually exclusive,
\eqref{ln:tc.fil.start} follows since $S_0$ is known to be $s_0$,
\eqref{ln:tc.fil.prop1} follows from Pearl's Property~1~\cite[pg 24]{pearl00book},
and \eqref{ln:tc.fil.lemma} follows from Lemma~\ref{lem:tc.fix-out}.
\end{proof}

\subsection{Proof of Theorem~3}

\begin{theorem}\label{thm:tc.inter-cause}
$Q$ has probabilistic interference iff there exists low inputs $\ell$ of length $t$ such that $\vec{V}_{\mathsf{hi}}^{t}$ has an effect on $\vec{V}_{\msf{lo}}^{t}$ given $V_{\mathsf{li}}^{t} := \ell$ in $M_Q$.  
\end{theorem}
\begin{proof}
In the notation of this appendix, $\vec{V}_{\mathsf{hi}}^{t}$ is $\vec{\msf{HI}}^{1:t}$ and $\vec{V}_{\msf{lo}}^{t}$ is $\vec{\msf{LO}}^{1:t}$.  For consistency, we write $\ell$ as $\vec{\msf{li}}$.

\newcommand{\cnum}[1]{\textbf{{#1}}}
Under this notation, we must show that 
\cnum{(1)} there exists input sequences $\vec{\imath}_1$ and $\vec{\imath}_2$ such that 
$\filter{\vec{\imath}_1}{L} = \filter{\vec{\imath}_2}{L}$ and $\filter{Q(\vec{\imath}_1)}{L} \neq \filter{Q(\vec{\imath}_2)}{L}$
if and only if
\cnum{(2)} there exists low inputs $\vec{\msf{li}}$ of length $t$ and high inputs $\vec{\msf{hi}}_1$ and $\vec{\msf{hi}}_2$ of length $t$
such that the probability distribution of $\vec{\msf{LO}}^{1:t}$ in $M_Q[\vec{\msf{HI}}^{1:t} := \vec{\msf{hi}}_1][\vec{\msf{LI}}^{1:t} := \vec{\msf{li}}]$
is not equal to its distribution in $M_Q[\vec{\msf{HI}}^{1:t} := \vec{\msf{hi}}_2][\vec{\msf{LI}}^{1:t} := \vec{\msf{li}}]$.

The distribution of $\vec{\msf{LO}}^{1:t}$ in $M_Q[\vec{\msf{HI}}^{1:t} := \vec{\msf{hi}}][\vec{\msf{LI}}^{1:t} := \vec{\msf{li}}]$ is given by 
$\mathcal{P}(\vec{\msf{LO}}^{1:t}{=}\vec{\msf{lo}} \:\given\: \doo(\vec{\msf{HI}}^{1:t}{:=}\vec{\msf{hi}}),\, \doo(\vec{\msf{LI}}^{1:t}{:=}\vec{\msf{li}}))$ for various values of $\vec{\msf{lo}}$.
Thus, Condition (2) is equivalent to 
\cnum{(3)} there exists $\vec{\imath}_1$ and $\vec{\imath}_2$ of length $t$ such that $\filter{\vec{\imath}_1}{L} = \filter{\vec{\imath}_2}{L}$ and there exists $\vec{\msf{lo}}$ such that
$\mathcal{P}(\vec{\msf{LO}}^{1:t}{=}\vec{\msf{lo}} \:\given\: \doo(\vec{I}^{1:t}{:=}\vec{\imath}_1))
\neq \mathcal{P}(\vec{\msf{LO}}^{1:t}{=}\vec{\msf{lo}} \:\given\: \doo(\vec{I}^{1:t}{:=}\vec{\imath}_2))$.

By Pearl's Property~1~\cite[pg 24]{pearl00book}, Condition (3) is equivalent to 
\cnum{(4)} there exists $\vec{\imath}_1$ and $\vec{\imath}_2$ of length $t$ such that $\filter{\vec{\imath}_1}{L} = \filter{\vec{\imath}_2}{L}$ and there exists $\vec{\msf{lo}}$ such that
$\mathcal{P}(\vec{\msf{LO}}^{1:t}{=}\vec{\msf{lo}} \:\given\: \doo(\vec{I}^{1:t-1}{:=}\vec{\imath}_1))
\neq \mathcal{P}(\vec{\msf{LO}}^{1:t}{=}\vec{\msf{lo}} \:\given\: \doo(\vec{I}^{1:t-1}{:=}\vec{\imath}_2))$ since the output at time $t$ does not depend upon the input at time $t$ in $M_Q$.

By Lemma~\ref{lem:tc.eq-filtered},
$\mathcal{P}(\vec{\msf{LO}}^{1:t}{=}\vec{\msf{lo}} \:\given\: \doo(\vec{I}^{1:t-1}{:=}\vec{\imath}))
= \filter{Q(\vec{\imath})}{L}(\vec{\msf{lo}})$.
Thus, Condition (4) is equivalent to 
\cnum{(5)} there exists input sequences $\vec{\imath}_1$, $\vec{\imath}_2$, and $\vec{\msf{lo}}$ such that 
$\filter{\vec{\imath}_1}{L} = \filter{\vec{\imath}_2}{L}$ and $\filter{Q(\vec{\imath}_1)}{L}(\vec{\msf{lo}}) \neq \filter{Q(\vec{\imath}_2)}{L}(\vec{\msf{lo}})$.

Condition (5) is equivalent to Condition (1).  Thus, Conditions (1) and (2) are equivalent as needed.
\end{proof}

\section{Proof of Corollary~\ref{cor:independ-noninter} Relating Independence Testing to Noninterference}

We start by showing that testing independence is the same as testing equality.
\begin{lemma}\label{lem:ind-eq}
For a set $C_1, \ldots, C_n$ of a mutually exclusive and exhaustive conditions,  for each $i$, there exists a $j$ such that $\Pr[A \given C_i] \neq \Pr[A \given C_j]$ iff there exists a $k$ such that $\Pr[A \given C_k] \neq \Pr[A]$.
\end{lemma}
\begin{proof}
If $\Pr[A \given C_i] \neq \Pr[A \given C_j]$, then either $\Pr[A \given C_i] \neq \Pr[A]$ or $\Pr[A \given C_j] \neq \Pr[A]$.

For the other direction, suppose $\Pr[A \given C_k] \neq \Pr[A]$. %
By the chain rule, $\Pr[A] = \sum_{h} \Pr[A \given C_h]\Pr[C_h]$.  
For showing a contradiction, suppose $\Pr[A \given C_k] = \Pr[A \given C_h]$ for all $h$.  Then, 
\begin{align}
\Pr[A] 
&= \sum_{h} \Pr[A \given C_h]\Pr[C_h]\\
&= \sum_{h} \Pr[A \given C_k]\Pr[C_h]\\
&= \Pr[A \given C_k] \sum_{h} \Pr[C_h]\\
&= \Pr[A \given C_k]
\end{align}
which is a contradiction.  Thus, there must exist some $h$ such that $\Pr[A \given C_k] \neq \Pr[A \given C_h]$.  
\end{proof}

The corollary states:
\begin{quote}
A test for independence of two random variables in science is a test of noninterference for information flow experiments. 
\end{quote}
\begin{proof}%
By Lemma~\ref{lem:ind-eq}, a test of independence is the same as a test of equality.  A test of equality is the same as a test for noninterference by Theorem~\ref{thm:tc.inter-cause}.
\end{proof}

\section{Proofs for Nonce Analysis}

\begin{proposition}
For all $\vec{y}$, 
$\pt(s_n, \vec{y}) = \mathsf{count}(\vec{y}, n) / |\vec{y}|$.
\end{proposition}
\begin{proof}
\begin{align}
\pt(s_n, \vec{y}) 
&= \frac{1}{|\vec{y}|!} \sum_{\pi \in \Pi(|\vec{y}|)} I[s_n(\vec{y}) \leq s_n(\pi(\vec{y}))]\\
&= \frac{1}{|\vec{y}|!} \sum_{\pi \in \Pi(|\vec{y}|)} 1 \leq s_n(\pi(\vec{y}))]\\
&= \frac{1}{|\vec{y}|!} \sum_{\pi \in \Pi(|\vec{y}|)} s_n(\pi(\vec{y}))\\
&= \frac{1}{|\vec{y}|!} \mathsf{count}(\vec{y}, n) * (|\vec{y}|-1)!\label{ln:num-times}\\
&= \mathsf{count}(\vec{y}, n) / |\vec{y}|
\end{align}
where \eqref{ln:num-times} is the number of permutations that puts a particular instance of the nonce into the first position times the number of instances of the nonce.
\end{proof}

\begin{proposition}
$\lim_{m \to \infty} \pt(s_n, \vec{y}_{m, \lceil p*m \rceil}) = p$ and $\lim_{m \to \infty} \pt(s_n, \vec{y}_{m, \lfloor p*m \rfloor}) = p$.
\end{proposition}
\begin{proof}

Since both
\begin{align}
\lim_{m \to \infty} \pt(s_n, \vec{y}_{m, \lfloor p*m \rfloor}) 
&= \lim_{m \to \infty} \frac{\mathsf{count}(\vec{y}_{m, \lfloor p*m \rfloor})}{|\vec{y}_{m, \lfloor p*m \rfloor}|}\\
&= \lim_{m \to \infty} \frac{1 + \lfloor p*m \rfloor}{m}\\
&= \lim_{m \to \infty} \frac{1 + \lfloor p*m \rfloor}{m}\\
&\geq \lim_{m \to \infty} \frac{1 + p*m - 1}{m}\\
&= \lim_{m \to \infty} p\\
&= p
\end{align}
and
\begin{align}
\lim_{m \to \infty} \pt(s_n, \vec{y}_{m, \lfloor p*m \rfloor}) 
&= \lim_{m \to \infty} \frac{\mathsf{count}(\vec{y}_{m, \lfloor p*m \rfloor})}{|\vec{y}_{m, \lfloor p*m \rfloor}|}\\
&= \lim_{m \to \infty} \frac{1 + \lfloor p*m \rfloor}{m}\\
&= \lim_{m \to \infty} \frac{1 + \lfloor p*m \rfloor}{m}\\
&\leq \lim_{m \to \infty} \frac{1 + p*m + 1}{m}\\
&= \lim_{m \to \infty} p\\
&= p
\end{align} 
it must be the case that $\lim_{m \to \infty} \pt(s_n, \vec{y}_{m, \lfloor p*m \rfloor}) = p$.

Since both
\begin{align}
\lim_{m \to \infty} \pt(s_n, \vec{y}_{m, \lceil p*m \rceil}) 
&= \lim_{m \to \infty} \frac{\mathsf{count}(\vec{y}_{m, \lceil p*m \rceil})}{|\vec{y}_{m, \lceil p*m \rceil}|}\\
&= \lim_{m \to \infty} \frac{1 + \lceil p*m \rceil}{m}\\
&= \lim_{m \to \infty} \frac{1 + \lceil p*m \rceil}{m}\\
&\leq \lim_{m \to \infty} \frac{1 + p*m + 1}{m}\\
&= \lim_{m \to \infty} p + \frac{2}{m}\\
&= p
\end{align}
and
\begin{align}
\lim_{m \to \infty} \pt(s_n, \vec{y}_{m, \lceil p*m \rceil}) 
&= \lim_{m \to \infty} \frac{\mathsf{count}(\vec{y}_{m, \lceil p*m \rceil})}{|\vec{y}_{m, \lceil p*m \rceil}|}\\
&= \lim_{m \to \infty} \frac{1 + \lceil p*m \rceil}{m}\\
&= \lim_{m \to \infty} \frac{1 + \lceil p*m \rceil}{m}\\
&\geq \lim_{m \to \infty} \frac{1 + p*m - 1}{m}\\
&= \lim_{m \to \infty} p\\
&= p
\end{align}
it must be the case that $\lim_{m \to \infty} \pt(s_n, \vec{y}_{m, \lceil p*m \rceil}) = p$
\end{proof}

\section{Details of Experiments}

\renewenvironment{tab}[1]
{\begin{center}
\let\oldarraystretch=\arraystretch
\renewcommand{\arraystretch}{1.05} %
\begin{tabular}{@{}#1@{}}
\toprule
}
{\bottomrule
\end{tabular}
\renewcommand{\arraystretch}{\oldarraystretch}
\end{center}}

All the experiments were carried out using Python bindings for Selenium WebDriver version 2.31 for the Firefox browser 25.0. Experiments were carried out with a script in Python 2.7 running on one of two identical 64-bit Ubuntu 12.04 VM with 24GB of RAM and 8 Intel Xeon E5540 CPUs. 
All network requests were made from behind a proxy server. 

When observing Google's behavior, we first ``opted-in'' to receive interest-based Google Ads across the web on every test instance by visiting the Google Ad Settings page at \url{https://www.google.com/settings/ads} and clicking the \texttt{Opt-in} link. This placed a Doubleclick cookie on the browser instance.

\subsection{Experiment~\ref{exp:cross-browsers}}

A primary browser instance would first establish an interest in cars by visiting car-related websites.

The car-related sites selected by collecting the top $10$ websites excluding images, news articles or ads returned by Google when queried with the search terms  ``BMW buy'', ``Audi purchase'', ``new cars'', ``local car dealers'', ``autos and vehicles'', ``cadillac prices'', and ``best limousines'' are shown in Table~\ref{tab:car}. Note that the results from ``local car dealers'' has only $9$ results because the page \url{local.yahoo.com/[redacted_location]/Automotive/Dealers/Used+Car+Dealers} 
took a long time to load and was manually removed from the training pages. 

\begin{table}
\caption{For Experiments~\ref{exp:cross-browsers} and~\ref{exp:pos}, the list of websites returned by Google upon searching with corresponding term. These websites were used for creating the profile of an auto enthusiast.}
\label{tab:car}
{\raggedright
\begin{description}
\item[``BMW buy''] \url{www.bmwusa.com/}, \url{www.autotrader.com/find/BMW-328i-cars-for-sale.jsp}, \url{www.autotrader.com/find/used-BMW-cars-for-sale.jsp}, \url{www.bmw.com/}, \url{www.bmwmotorcycles.com/}, \url{autos.aol.com/new-cars/}, \url{www.exchangeandmart.co.uk/used-cars-for-sale/bmw}, \url{en.wikipedia.org/wiki/BMW}, \url{www.cars.com/bmw/}, \url{autos.aol.com/bmw/}
\item[``Audi purchase''] \url{www.audiusa.com/inventory/european-delivery}, \url{www.audiusa.com/help/leasing}, \url{www.audiusa.com/myaudi/finance}, \url{www.audiusa.com/myaudi/offers-programs}, \url{www.audiusa.com/inventory/certified-pre-owned}, \url{www.audisupplier.com/}, \url{townhall-talk.edmunds.com/direct/view/.f1cc6d7}, \url{www.autotrader.com/find/Audi-A3-cars-for-sale.jsp}, \url{en.wikipedia.org/wiki/Audi}, \url{jalopnik.com/5903083/why-audi-just-bought-ducati}
\item[``new cars''] \url{www.edmunds.com/new-cars/}, \url{www.edmunds.com/car-reviewsautos.yahoo.com/new-cars.html}, \url{autos.yahoo.com/new-cars.html}, \url{www.kbb.com/new-cars/}, \url{www.autotrader.com/research/new-cars/}, \url{www.autotrader.com/buy-a-new-car.jsp}, \url{www.cars.com/}, \url{autos.aol.com/new-cars/}, \url{www.newcars.com/}, \url{www.motortrend.com/new_cars/}
\item[``local car dealers''] \url{www.edmunds.com/dealerships/}, \url{www.cars.com/dealers/search.action}, \url{www.cochran.com/}, \url{www.autotrader.com/find/[redacted_location].jsp}, \url{www.baierl.com/}, \url{www.kbb.com/car-dealers-and-inventory/}, \url{www.enterprisecarsales.com/location/.../Enterprise_Car_Sales_[redacted_location]}, \url{autos.aol.com/new-cars}, \url{www.toyota.com/dealers/}
\item[``autos and vehicles''] \url{www.youtube.com/channel/HCLfhQGBROujg}, \url{www.youtube.com/channel/HCHXCPGmshRz4}, \url{www.youtube.com/live/autos}, \url{en.wikipedia.org/wiki/Automobile}, \url{www.veoh.com/list/videos/autos_and_vehicles}, \url{vidstatsx.com/most-popular-autos-vehicles-videos-today}, \url{www.savevid.com/category/auto-vehicles}, \url{www.smbiz.com/sbrl003.html}, \url{www.pinterest.com/hasaniqbal/autos-and-vehicles/}, \url{www.justluxe.com/lifestyle/car/articles-2.php}
\item[``cadillac prices''] \url{www.truecar.com/prices-new/cadillac/}, \url{www.motortrend.com/new_cars}, \url{autos.msn.com/browse/Cadillac.aspx}, \url{www.nadaguides.com/Cars/Cadillac}, \url{autos.yahoo.com/new-cars.html}, \url{www.gizmag.com/cadillac-elr-plug-in-hybrid-price/29389/}, \url{www.autonews.com/article/20131011/RETAIL03/131019967/}, \url{www.cadillac.com/}, \url{usnews.rankingsandreviews.com/cars-trucks/browse/cadillac}, \url{www.automobilemag.com/car_prices/01/cadillac/}
\item[``best limousines''] \url{www.medialightbox.com/blog/.../the-10-best-limousines-in-the-world/}, \url{www.bestlimousines.com/}, \url{www.celebritylimos[redacted_location].com/}, \url{www.tdflimo.com/}, \url{www.limo.com/limo-[redacted_location]-limousines.php}, \url{www.[redacted_location]luxurylimoservice.com/}, \url{www.angieslist.com/companylist/}, \url{www.forbes.com/2005/03/10/cx_dl_0310feat_bill05.html}, \url{www.thebestlimousine.com/}, \url{www.youtube.com/watch?v=0iqi6jHviJ0}
\end{description}
}
\end{table}

During the $10$ rounds of ad collections, each round would attempt to reload the International Homepage of Times of India (\url{http://timesofindia.indiatimes.com/international-home}) $10$ times.  Occasionally it would time out instead of reloading. We set the page-load-timeout to be $60$ seconds. 
We repeated the experiment four times (twice using $10$ rounds and twice using $20$ rounds) and  found that the page would not always load completely resulting in fewer ads being collected. Details on the number of ads collected by the primary browser instance in each round are shown in Table~\ref{tab:epoch-rounds}.

\begin{table}
\caption{For Experiment~\ref{exp:cross-browsers}, the number of unique ads collected.  $I$ denotes the set of all ads collected from the primary browser instance running in isolation, while $P$ denotes the same collected from the primary browser instance running in parallel. This table shows the number of ads collected in each round as well as the total number of ads and the number of unique ads in $I$ and $P$. The stars represent numbers from the instances running in isolation. }
\label{tab:epoch-rounds}
\begin{tab}{@{}ccccc@{}}%
Data set & \#rounds & ads (unique) collected by primary browser per round & 
total (unique) in $I$ & total (unique) in $P$ \\
\midrule

1 & 10 	& *50(13), *50(13), 50(8), 50(10), *50(10), 	& 250(37) & 250(25) \\
   & 		& 50(12), *50(13), 50(11), 50(7), *50(17)	& & \\

\midrule
2 & 10 	& 50(11), *50(14), 50(15), 50(11), 50(13), 	&  250(46) & 245(33) \\
   & 		& *50(19), *50(13), *50(14), 45(11), *50(14) 	& & \\
\midrule
3 & 20 	& *50(12), *50(12), 42(11), 50(14), *50(12), 	& 490(58) & 492(47) \\
   & 		& 50(11), 50(13), *50(18), *50(15), 50(15), 	& & \\
   & 		& 50(14), 50(9), 50(17), *50(17), 50(10), 	& & \\
   & 		& *45(10), *50(12), 50(13), *50(16), *45(13)  & & \\
\midrule

4 & 20 	& 50(10), 50(10), 50(15), *50(14), 50(10), 	&  485(57) & 495(52) \\
   &  		& *50(17), 50(13), *40(11), 50(10), 50(16), 	& & \\
   &		& *50(14), *50(11), *50(14), 50(13), 45(11),	& & \\ 
   & 		& *50(14), *50(14), 50(12), *50(12), *45(16)	& & \\
\end{tab}
\end{table}

\clearpage
\subsection{Experiment~\ref{exp:binning}}
This experiment suggested that Google associates users with various ad pools switching users from pool to pool over time. 
Plotting the ads from both the instances together, as in Figure~\ref{fig:exp-bin-com}, we observe that for a period of time (between approximately the $60^{th}$ and $120^{th}$ reload), both the instances appear to receive ads from the same pool. We also ran the same experiment with different intervals between successive reloads. We tested intervals of $0$s, $5$s, $15$s, $30$s, $60$s, and $120s$, the ad-plots of which are shown in Figure~\ref{fig:all-bin}
\vfill
\begin{figure}[h]
\begin{center}
\includegraphics[width = 6.5in]{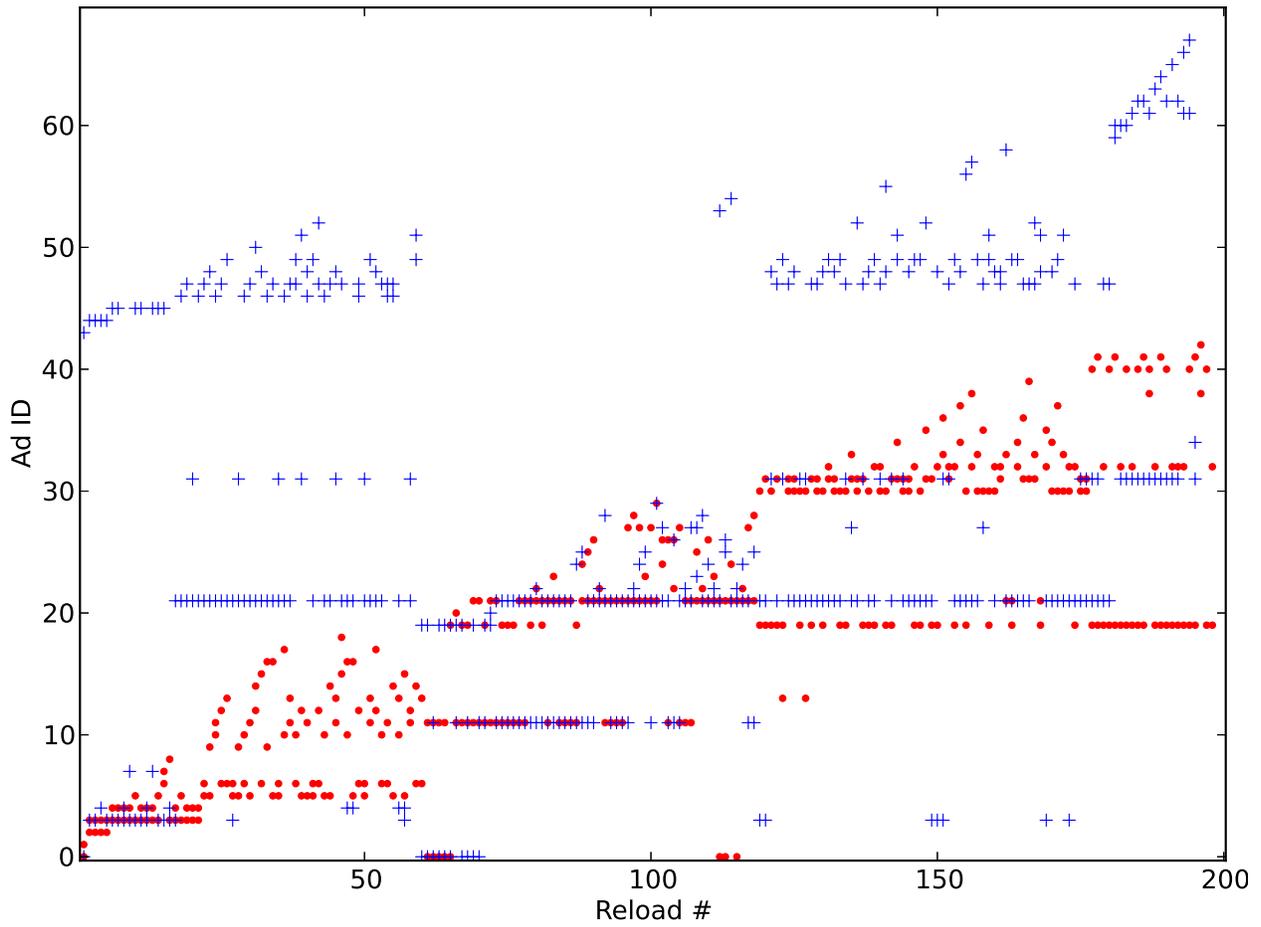}
\end{center}
\caption{For Experiment~\ref{exp:binning}, the combined plot of ads from Instances 1 and 2}\label{fig:exp-bin-com}
\end{figure}
\vfill
\mbox{}

\begin{figure}
\begin{center}
\begin{tabular}{cc}
\includegraphics[width = 3in]{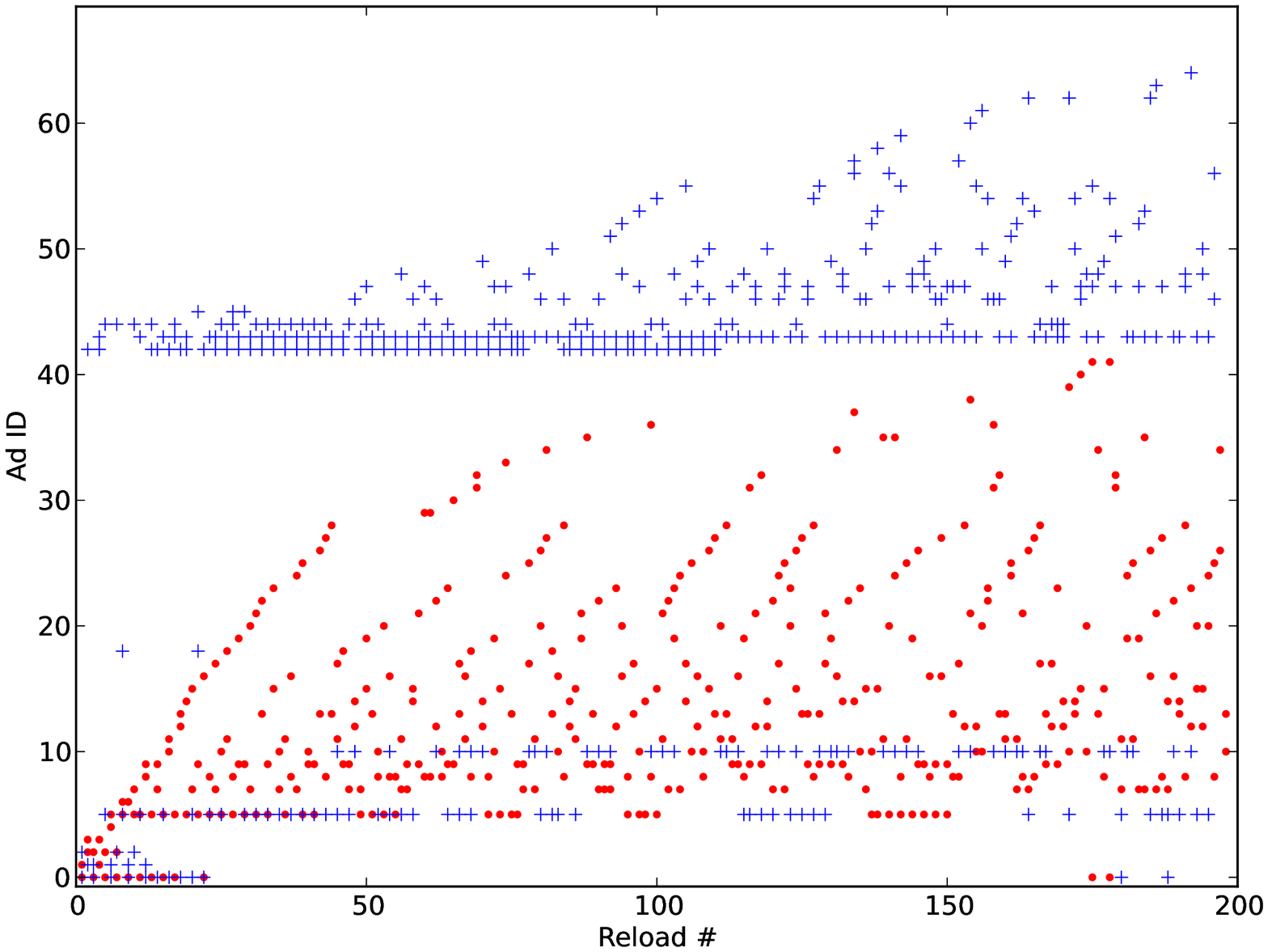} & \includegraphics[width = 3in]{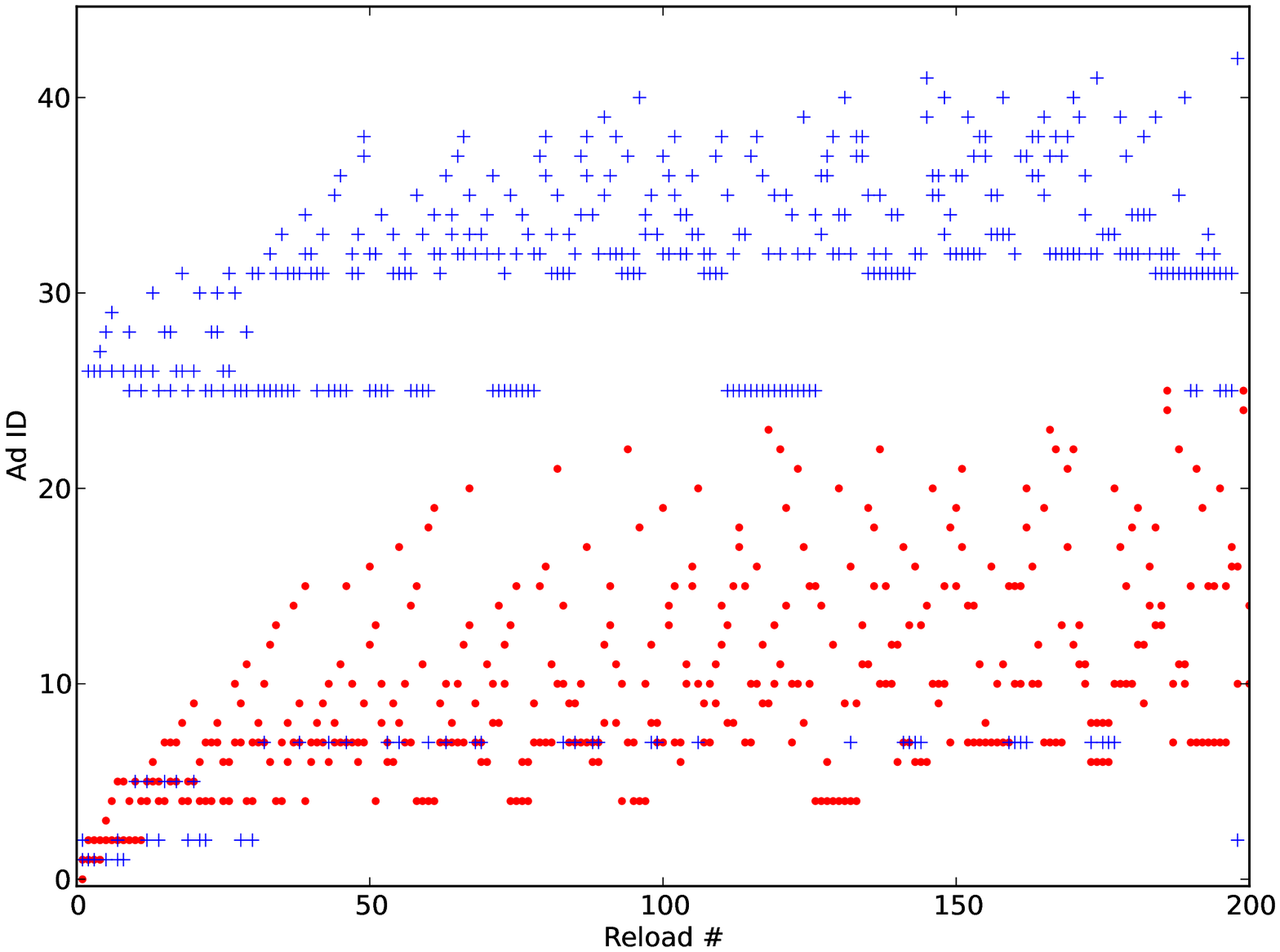} \\
(a) interval = $0s$ & (b) interval = $5s$ \\[3ex]
\includegraphics[width = 3in]{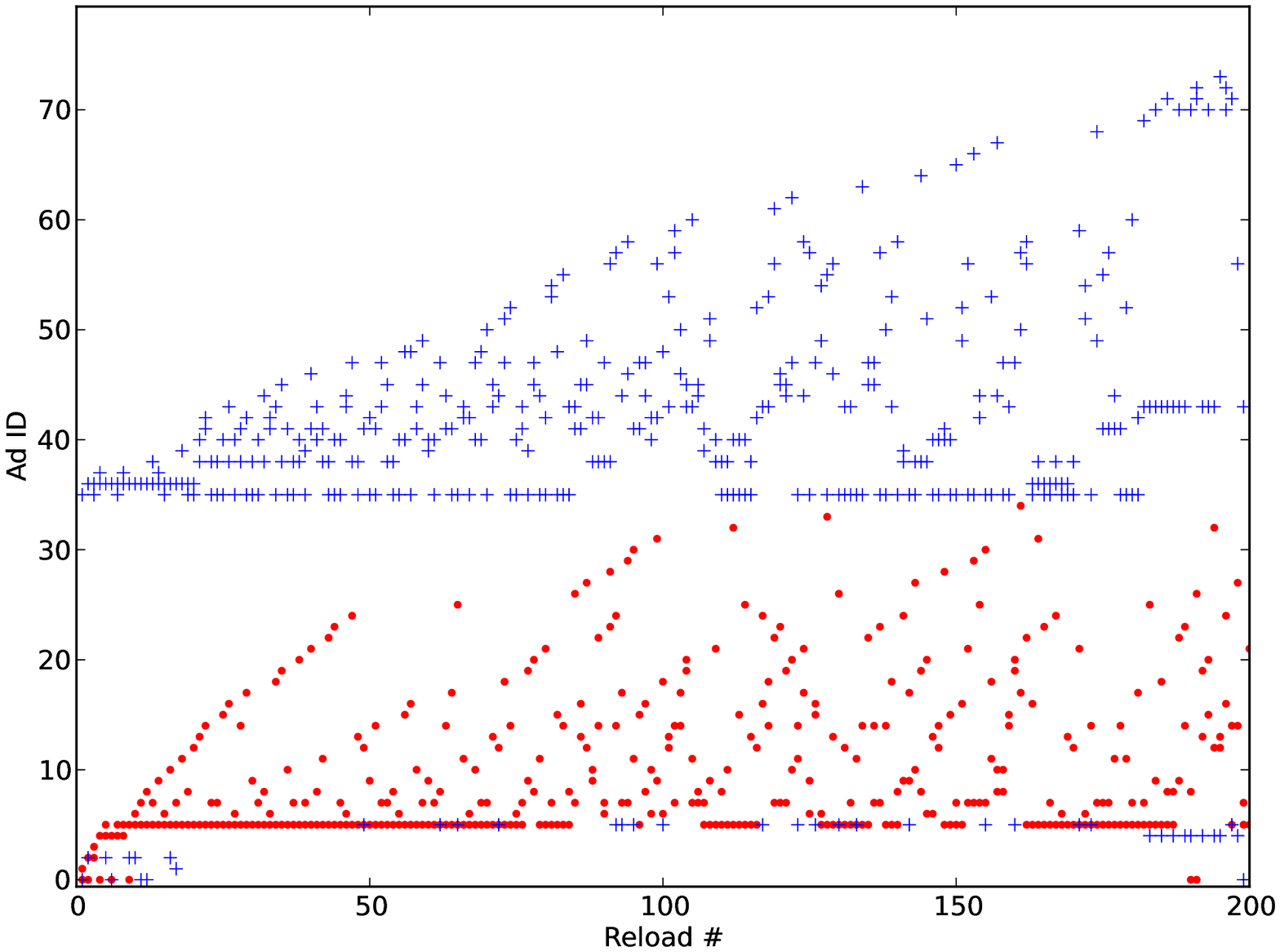} & \includegraphics[width = 3in]{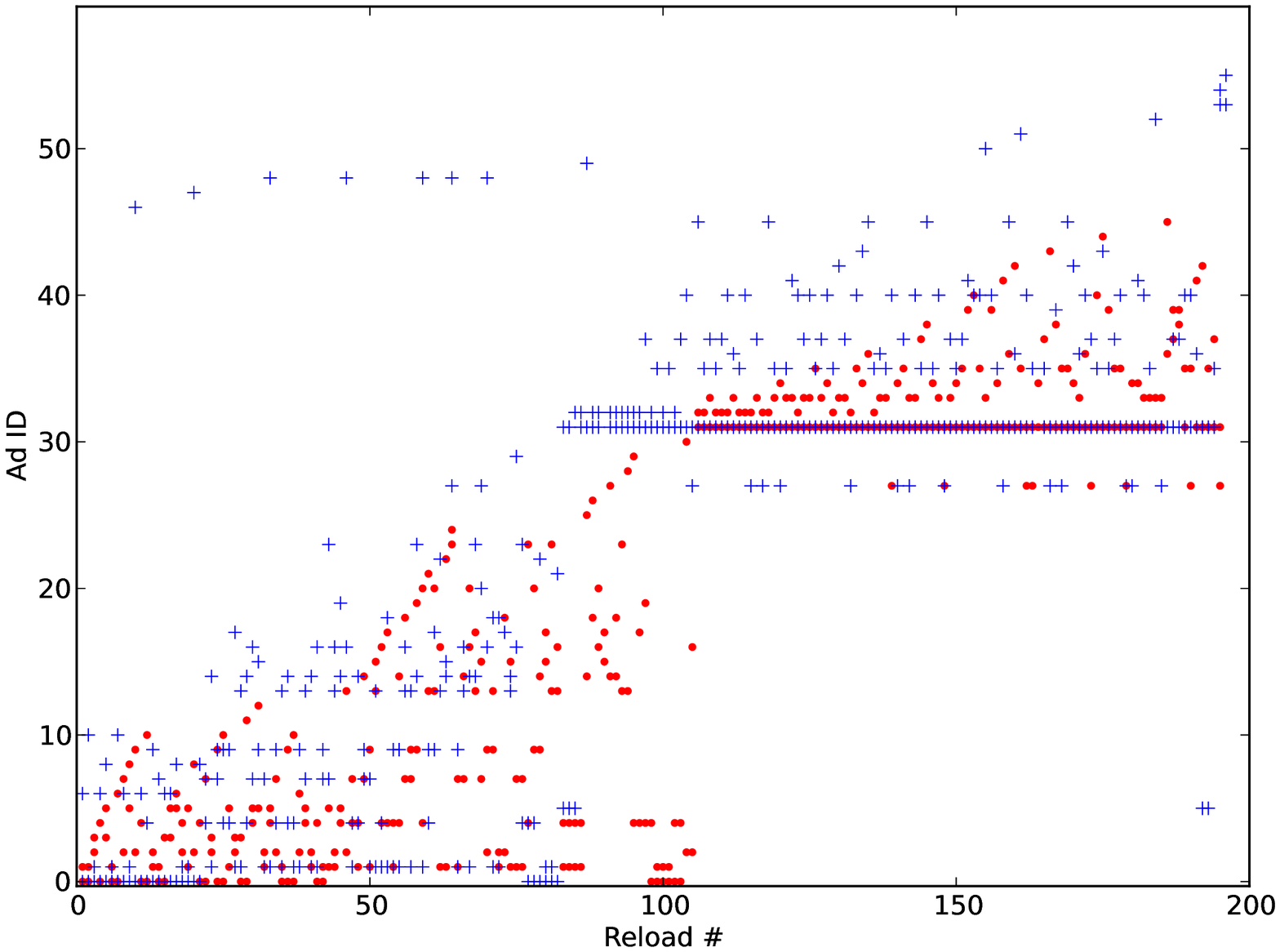} \\
(c) interval = $15s$ & (d) interval = $30s$ \\[3ex]
\includegraphics[width = 3in]{bin-fig-woline+-flipped.eps} & \includegraphics[width = 3in]{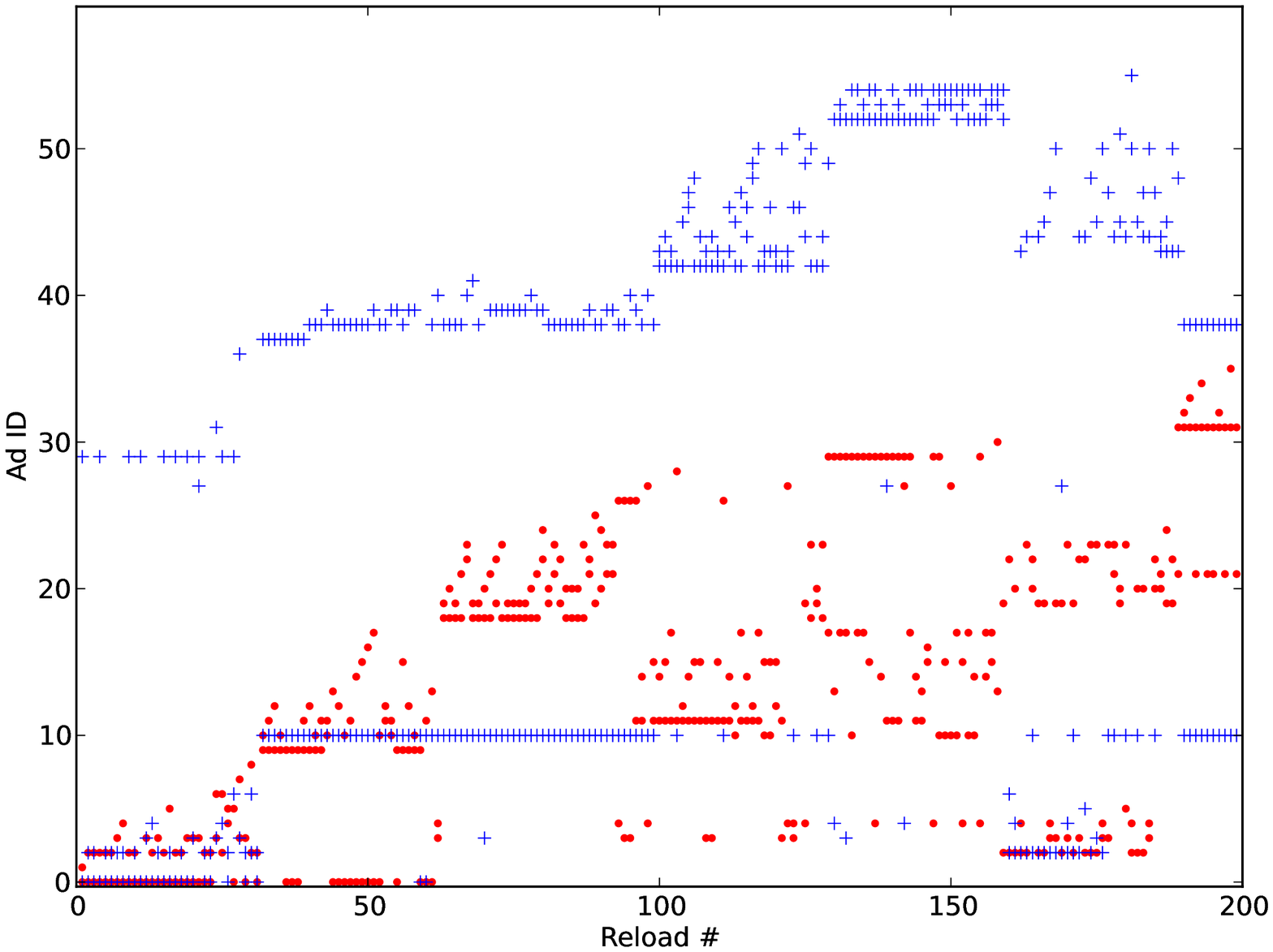} \\
(e) interval = $60s$ & (f) interval = $120s$ \\
\end{tabular}
\end{center}
\caption{For Experiment~\ref{exp:binning}, plots of ads from Instances 1 and 2 in of the six experiments with varying time intervals between reloads. Observe that the pooling behavior appears for the first time in \ref{fig:all-bin}(d), where the pool seems to switch somewhere around the $80$th reload. After that the number of these switches keep increasing in successive plots with the reload interval.}
\label{fig:all-bin}
\end{figure}
\clearpage
\subsection{Experiment~\ref{exp:gender}}
We got the top $100$ websites for females from\\
\centerline{\url{http://www.alexa.com/topsites/category/Top/Society/People/Women}} 
and the top $100$ sites for males from\\
\centerline{\url{http://www.alexa.com/topsites/category/Top/Society/People/Men}}
They are listed in Tables~\ref{tab:women} and~\ref{tab:men}, respectively.

Google's Ad Settings page (previously known as the Ad Preferences Manager) is located at\\
\centerline{\url{http://www.google.com/settings/ads}}

\vfill
\begin{table}[h]
\caption{For Experiment~\ref{exp:gender}, the list of websites for creating female personas}
\label{tab:women}
\url{shine.yahoo.com}, \url{sheknows.com}, \url{realsimple.com}, \url{cosmopolitan.com}, \url{shape.com}, \url{yourtango.com}, \url{glamour.com}, \url{allwomenstalk.com}, \url{self.com}, \url{womansday.com}, \url{indusladies.com}, \url{sofeminine.co.uk}, \url{allure.com}, \url{cosmopolitan.co.uk}, \url{redbookmag.com}, \url{bellaonline.com}, \url{chatelaine.com}, \url{womenshealth.gov}, \url{womensforum.com}, \url{more.com}, \url{blisstree.com}, \url{memsaab.com}, \url{handbag.com}, \url{bitchmagazine.org}, \url{feministing.com}, \url{divine.ca}, \url{inthepowderroom.com}, \url{penmai.com}, \url{bust.com}, \url{shoppinglifestyle.com}, \url{msmagazine.com}, \url{anewmode.com}, \url{bettyconfidential.com}, \url{ywbb.org/index.shtml}, \url{worldoffemale.com}, \url{herdaily.com}, \url{lady.co.uk}, \url{worldpulse.com}, \url{sophisticatededge.com}, \url{baggagereclaim.co.uk}, \url{pmsclan.com}, \url{girlfriendology.com}, \url{lemondrop.com}, \url{bcliving.ca}, \url{journeywoman.com}, \url{australianwomenonline.com}, \url{women-on-the-road.com}, \url{magforwomen.com}, \url{nawbo.org}, \url{dressforsuccess.org}, \url{womenshistory.about.com}, \url{wavejourney.com}, \url{secondwivescafe.com}, \url{unwomen.org}, \url{aauw.org}, \url{catalyst.org}, \url{truthaboutdeception.com}, \url{womensissues.about.com}, \url{ncwit.org}, \url{dawnali.com/lovinmysistas}, \url{mookychick.co.uk}, \url{savvy-chick.net}, \url{rawa.org}, \url{emilyslist.org}, \url{constantchatter.com}, \url{girlfriendsocial.com}, \url{womenzmag.com}, \url{ladieswholaunch.com}, \url{maitinepal.org}, \url{geniusbeauty.com}, \url{thefword.org.uk}, \url{womensenews.org}, \url{rockytravel.net}, \url{femail.com.au}, \url{onewomanmarketing.com}, \url{un.org/womenwatch}, \url{webgrrls.com}, \url{feminist.com}, \url{iwda.org.au}, \url{feminist.org}, \url{mrssurvival.com}, \url{gogirlfriend.com}, \url{nzgirl.co.nz}, \url{digital.library.upenn.edu/women/}, \url{daisygreenmagazine.co.uk}, \url{now.org}, \url{womensnetwork.com.au}, \url{jwa.org}, \url{library.duke.edu/rubenstein/}, \url{heartlessbitches.com}, \url{gogalavanting.com}, \url{redhatsociety.com}, \url{witi.com}, \url{womenslaw.org}, \url{wnba-books.org}, \url{vday.org}, \url{everywoman.com}, \url{vivmag.com}, \url{womenonlyforums.com}, \url{teachertech.rice.edu}
\end{table}

\vfill

\begin{table}[h]
\caption{For Experiment~\ref{exp:gender}, the list of websites for creating male personas}
\label{tab:men}
\url{askmen.com}, \url{complex.com}, \url{menshealth.com}, \url{esquire.com}, \url{gq.com}, \url{artofmanliness.com}, \url{thrillist.com}, \url{maxim.com}, \url{mademan.com}, \url{uncrate.com}, \url{guyism.com}, \url{everyjoe.com}, \url{coolmaterial.com}, \url{spike.com}, \url{gearpatrol.com}, \url{goodmenproject.com}, \url{fhm.com}, \url{bullz-eye.com}, \url{mensjournal.com}, \url{blessthisstuff.com}, \url{avoiceformen.com}, \url{primermagazine.com}, \url{thesmokingjacket.com}, \url{acquiremag.com}, \url{tmrzoo.com}, \url{unfinishedman.com}, \url{thecoolist.com}, \url{werd.com}, \url{gunaxin.com}, \url{ywbb.org/index.shtml}, \url{instash.com}, \url{giantlife.com}, \url{plunderguide.com}, \url{gearculture.com}, \url{hispotion.com}, \url{mensgear.net}, \url{modernman.com}, \url{manofmany.com}, \url{brash.com}, \url{fearlessmen.com}, \url{dadsdivorce.com}, \url{savethemales.ca}, \url{tempe12.com}, \url{justaguything.com}, \url{mkp.org}, \url{sharpformen.com}, \url{pinstripemag.com}, \url{thecampussocialite.com}, \url{fatherhood.org}, \url{guylife.com}, \url{mankindunplugged.com}, \url{grind365.com}, \url{manukau.ac.nz}, \url{thegearpost.com}, \url{nextluxury.com}, \url{bonjourlife.com}, \url{nomoremrniceguy.com}, \url{shavemagazine.com}, \url{nextcrave.com}, \url{toromagazine.com}, \url{ziprage.com}, \url{menstuff.org}, \url{ncfm.org}, \url{angryharry.com}, \url{fact.on.ca}, \url{aspiringgentleman.com}, \url{fataldose.com}, \url{debonairmag.com}, \url{dailyxy.com}, \url{citynetmagazine.com}, \url{male-initiation.net}, \url{faculty.washington.edu/eloftus/}, \url{losangeles.mkp.org}, \url{fancymaterial.com}, \url{owenmarcus.com}, \url{manlyadventure.com}, \url{mensactivism.org}, \url{beast.com}, \url{thepopularman.com}, \url{menstoppingviolence.org}, \url{doubleagent.com}, \url{guymanningham.com}, \url{contemporarymasculine.com}, \url{thecmg.org}, \url{rtinternational.org}, \url{justdetention.org}, \url{maninstitute.com}, \url{uk.mkp.org}, \url{man-over-board.com}, \url{jaysongaddis.com}, \url{sospapa.net}, \url{dullmensclub.com}, \url{askmamu.com}, \url{taoofbachelorhood.com}, \url{anger.org}, \url{dandyism.net}, \url{acfc.org}, \url{fathersforlife.org}, \url{singlesexschools.org}, \url{frachelli.com}
\end{table}
\vfill
\mbox{}

\clearpage
\subsection{Experiment~\ref{exp:pos}}
As in Experiment~\ref{exp:cross-browsers}, an instance manifests its interest by visiting the top $10$ websites returned by Google when queried with certain automobile-related terms: ``BMW buy'', ``Audi purchase'', ``new cars'', ``local car dealers'', ``autos and vehicles'', ``cadillac prices'', and ``best limousines''.
Thus, they visited the same websites as in Experiment~\ref{exp:cross-browsers} (see Table~\ref{tab:car}).

Across all runs of the experiment, we collected $9832$ ads with $281$ being unique.  
Table~\ref{tab:instances} shows the number of ads collected by each instance.   Notice that both outliers were in the $19$th run and in the experimental group.

\begin{table}[b!]
\caption{For Experiment~\ref{exp:pos}, how the ads were distributed over the $10$ different instances. $T$ denotes the set of all ads collected from the trained instances, while $U$ denotes the same collected from the untrained instances. The number of ads collected by each instance in $\{i_1 \dots i_{10}\}$ is shown in the left half of the table. The right half of the table shows the total number of ads and the number of unique ads in $T$ and $U$. }
\label{tab:instances}
\begin{tab}{@{}rrrrrrrrrrrrrrr@{}}%
Data set & $i_1$ & $i_2$ & $i_3$ & $i_4$ & $i_5$ & $i_6$ & $i_7$ & $i_8$ & $i_9$ & $i_{10}$ & 
Total($T$) & Unique($T$) & Total($U$) & Unique($U$) \\
\midrule
1 & 45 &50 &45 &50 &45 &50 &45 &50 &45 &50 &235 & 28 & 240 & 44\\
2 & 50 &50 &50 &49 &50 &50 &50 &50 &50 &50 &250 & 28 & 249 & 38\\
3 & 50 &50 &50 &50 &50 &50 &50 &50 &50 &50 &250 & 38 & 250 & 30\\
4 & 50 &50 &50 &50 &50 &50 &50 &50 &50 &50 &250 & 28 & 250 & 34\\
5 & 50 &50 &50 &50 &50 &50 &50 &50 &50 &50 &250 & 36 & 250 & 31\\
6 & 50 &50 &50 &50 &50 &50 &46 &50 &50 &50 &250 & 31 & 246 & 37\\
7 & 42 &50 &50 &50 &50 &50 &50 &50 &50 &50 &242 & 25 & 250 & 39\\
8 & 50 &50 &50 &50 &50 &50 &50 &50 &50 &50 &250 & 27 & 250 & 22\\
9 & 50 &50 &45 &50 &50 &50 &50 &48 &50 &50 &250 & 29 & 243 & 52\\
10 & 50 &50 &50 &50 &50 &50 &50 &50 &49 &50 &249 & 27 & 250 & 30\\
11 & 50 &50 &50 &50 &50 &50 &50 &50 &50 &50 &250 & 29 & 250 & 38\\
12 & 50 &50 &50 &48 &50 &49 &50 &50 &50 &50 &250 & 35 & 247 & 38\\
13 & 50 &50 &50 &50 &50 &50 &48 &50 &50 &50 &250 & 37 & 248 & 30\\
14 & 50 &50 &50 &50 &50 &50 &50 &50 &50 &50 &250 & 52 & 250 & 28\\
15 & 50 &50 &50 &50 &50 &50 &50 &50 &50 &50 &250 & 40 & 250 & 35\\
16 & 50 &50 &50 &50 &50 &50 &50 &50 &50 &50 &250 & 24 & 250 & 40\\
17 & 50 &50 &41 &50 &50 &48 &49 &50 &50 &50 &250 & 39 & 238 & 38\\
18 & 50 &50 &45 &50 &50 &50 &50 &50 &50 &50 &250 & 26 & 245 & 44\\
19 & 50 &50 &50 &50 &0 &50 &0 &50 &50 &50 &150 & 24 & 250 & 53\\
20 & 50 &50 &50 &50 &50 &50 &50 &50 &50 &50 &250 & 46 & 250 & 34\\
\end{tab}
\end{table}

Across all runs of the control-control experiment, we collected 9304 ads with 295 being unique.  
Table~\ref{tab:control} shows the number of ads collected by each instance.
\begin{table}
\caption{For Experiment~\ref{exp:pos}, how the ads were distributed over the $10$ different instances in the {\em control-control} experiment. $5$ out these $10$ were randomly assigned to $T$, while the remaining to $U$. Observe that data-set $8$ is an outlier because the instances in that round returned much fewer ads. }
\label{tab:control}
\begin{tab}{@{}rrrrrrrrrrrrrrr@{}}%
Data set & $i_1$ & $i_2$ & $i_3$ & $i_4$ & $i_5$ & $i_6$ & $i_7$ & $i_8$ & $i_9$ & $i_{10}$ & 
Total($T$) & Unique($T$) & Total($U$) & Unique($U$) \\
\midrule
1 & 45 &30 &30 &45 &25 &40 &45 &45 &35 &29 &190 & 31 & 179 & 44\\
2 & 50 &50 &50 &50 &50 &50 &50 &23 &50 &50 &223 & 33 & 250 & 50\\
3 & 45 &40 &45 &50 &45 &45 &45 &45 &45 &45 &225 & 37 & 225 & 39\\
4 & 50 &50 &50 &50 &50 &45 &50 &50 &50 &50 &245 & 39 & 250 & 39\\
5 & 50 &50 &50 &50 &46 &50 &50 &50 &50 &45 &245 & 33 & 246 & 57\\
6 & 50 &50 &45 &50 &50 &50 &50 &45 &45 &50 &250 & 45 & 235 & 38\\
7 & 50 &47 &50 &50 &50 &50 &50 &50 &50 &50 &250 & 42 & 247 & 34\\
8 & 25 &15 &25 &0 &9 &30 &19 &0 &0 &15 &114 & 15 & 24 & 16\\
9 & 50 &50 &50 &50 &50 &45 &50 &45 &50 &50 &245 & 37 & 245 & 33\\
10 & 45 &45 &50 &50 &50 &45 &45 &45 &50 &50 &245 & 36 & 230 & 47\\
11 & 50 &50 &45 &50 &45 &50 &50 &50 &44 &50 &239 & 35 & 245 & 37\\
12 & 50 &49 &50 &50 &50 &40 &45 &50 &50 &50 &235 & 33 & 249 & 36\\
13 & 50 &50 &50 &50 &45 &50 &50 &50 &50 &50 &245 & 36 & 250 & 24\\
14 & 50 &50 &50 &50 &50 &50 &50 &50 &50 &50 &250 & 31 & 250 & 28\\
15 & 50 &50 &50 &50 &46 &50 &50 &50 &50 &47 &246 & 45 & 247 & 43\\
16 & 50 &50 &50 &50 &50 &50 &50 &50 &50 &50 &250 & 36 & 250 & 35\\
17 & 50 &50 &50 &49 &50 &50 &50 &50 &50 &50 &249 & 37 & 250 & 36\\
18 & 50 &50 &50 &50 &50 &50 &50 &50 &50 &50 &250 & 26 & 250 & 27\\
19 & 50 &50 &50 &50 &50 &50 &50 &37 &50 &50 &237 & 36 & 250 & 33\\
20 & 50 &50 &50 &50 &50 &50 &50 &50 &49 &50 &249 & 37 & 250 & 34\\
\end{tab}
\end{table}
The p-values that the permutation tests yielded for the control-control experiment are shown in Table~\ref{tbl:control}. 
\begin{table}
\caption{For Experiment~\ref{exp:pos}, p-values for the for control-control experiment. Note that the significant p-values are from data-set 8, which we showed in Table~\ref{tab:control} to an outlier.}
\label{tbl:control}. 
\begin{tab}{@{}ccccc@{}}
 Data set & $s_{\simm}$ & $s_{\msf{kw}}$  & $s_{\percent}$ & $\chi^2$ \\
\midrule
$1$ & $0.373016$ 	 & $0.857143$ 	 & $0.777778$ 	 & $0.0690831$ \\
$2$ & $0.063492$ 	 & $0.293651$ 	 & $0.261905$ 	 & $0.0388589$ \\
$3$ & $0.603175$ 	 & $0.920635$ 	 & $0.777778$ 	 & $2.66915e-05$ \\
$4$ & $0.436508$ 	 & $0.440476$ 	 & $0.500000$ 	 & $0.565445$ \\
$5$ & $0.071429$ 	 & $0.869048$ 	 & $1.000000$ 	 & $9.85584e-05$ \\
$6$ & $0.309524$ 	 & $0.158730$ 	 & $0.500000$ 	 & $0.0139651$ \\
$7$ & $0.103175$ 	 & $0.527778$ 	 & $1.000000$ 	 & $0.947502$ \\
$8$ & $0.007937^*$ 	 & $0.003968^*$ 	 & $0.003968^*$ 	 & $0.0701231$ \\
$9$ & $0.547619$ 	 & $0.134921$ 	 & $0.222222$ 	 & $0.0216323$ \\
$10$ & $0.119048$ 	 & $1.000000$ 	 & $1.000000$ 	 & $0.000856692$ \\
$11$ & $0.936508$ 	 & $0.234127$ 	 & $0.222222$ 	 & $0.0341701$ \\
$12$ & $0.285714$ 	 & $0.769841$ 	 & $0.222222$ 	 & $0.228014$ \\
$13$ & $0.761905$ 	 & $0.440476$ 	 & $0.896825$ 	 & $0.00237996$ \\
$14$ & $0.642857$ 	 & $0.408730$ 	 & $1.000000$ 	 & $0.415073$ \\
$15$ & $0.468254$ 	 & $0.738095$ 	 & $1.000000$ 	 & $0.164419$ \\
$16$ & $0.476190$ 	 & $0.095238$ 	 & $0.500000$ 	 & $0.000130842$ \\
$17$ & $0.984127$ 	 & $0.186508$ 	 & $0.500000$ 	 & $0.0254968$ \\
$18$ & $0.746032$ 	 & $0.440476$ 	 & $0.896825$ 	 & $0.464851$ \\
$19$ & $0.611111$ 	 & $0.420635$ 	 & $0.500000$ 	 & $0.0122963$ \\
$20$ & $0.071429$ 	 & $0.936508$ 	 & $0.777778$ 	 & $2.27046e-05$ \\
\midrule
$\text{Number} < 5\%$ & 1  &              1       &      1       &                        12\\
\end{tab}
\end{table}
We can see that each of the statistics produced one statistically significant result except for the $\chi^2$, which produced $12$. This seems to indicate that the $\chi^2$-test is more prone to showing false-positives than the permutation tests.

Across all runs of the treatment-treatment experiment, we collected 9741 ads with 243 being unique.  
Table~\ref{tab:treatment} shows the number of ads collected by each instance.
\begin{table}
\caption{For Experiment~\ref{exp:pos}, how the ads were distributed over the $10$ different instances in the {\em treatment-treatment} experiment. $5$ out these 10 were randomly assigned to $T$, while the remaining to $U$.}
\label{tab:treatment}
\begin{tab}{@{}rrrrrrrrrrrrrrr@{}}
Data set & $i_1$ & $i_2$ & $i_3$ & $i_4$ & $i_5$ & $i_6$ & $i_7$ & $i_8$ & $i_9$ & $i_{10}$ & 
Total($T$) & Unique($T$) & Total($U$) & Unique($U$) \\
\midrule
1 & 50 &50 &50 &50 &50 &50 &45 &50 &50 &50 &245 & 31 & 250 & 33\\
2 & 50 &50 &50 &50 &50 &50 &45 &50 &50 &43 &250 & 33 & 238 & 42\\
3 & 50 &50 &50 &50 &50 &50 &50 &45 &50 &50 &245 & 37 & 250 & 37\\
4 & 50 &50 &50 &50 &50 &50 &50 &50 &50 &50 &250 & 31 & 250 & 45\\
5 & 49 &49 &50 &50 &50 &50 &50 &50 &50 &49 &248 & 34 & 249 & 46\\
6 & 50 &50 &50 &50 &50 &50 &50 &50 &45 &45 &240 & 40 & 250 & 32\\
7 & 45 &50 &50 &50 &50 &50 &50 &36 &50 &50 &250 & 40 & 231 & 40\\
8 & 50 &40 &50 &50 &50 &50 &50 &50 &50 &50 &240 & 36 & 250 & 35\\
9 & 50 &50 &45 &50 &45 &40 &50 &50 &40 &45 &230 & 26 & 235 & 33\\
10 & 50 &50 &50 &50 &50 &50 &50 &50 &50 &50 &250 & 33 & 250 & 32\\
11 & 50 &49 &50 &50 &50 &50 &50 &50 &50 &50 &249 & 35 & 250 & 41\\
12 & 45 &45 &50 &50 &50 &0 &50 &50 &45 &50 &195 & 25 & 240 & 43\\
13 & 45 &50 &50 &50 &50 &50 &50 &50 &45 &0 &195 & 21 & 245 & 37\\
14 & 50 &50 &50 &49 &46 &50 &50 &50 &50 &50 &250 & 37 & 245 & 28\\
15 & 50 &50 &50 &50 &50 &50 &50 &49 &50 &50 &250 & 39 & 249 & 28\\
16 & 50 &50 &50 &50 &50 &50 &50 &50 &50 &50 &250 & 28 & 250 & 33\\
17 & 50 &50 &50 &50 &50 &50 &50 &50 &50 &50 &250 & 23 & 250 & 45\\
18 & 50 &50 &50 &50 &50 &50 &50 &50 &50 &50 &250 & 38 & 250 & 37\\
19 & 50 &45 &47 &45 &50 &50 &50 &50 &50 &50 &247 & 34 & 240 & 26\\
20 & 45 &45 &50 &50 &50 &50 &50 &45 &50 &50 &240 & 44 & 245 & 34\\
\end{tab}
\end{table}
The p-values for the treatment-treatment experiments are shown in Table~\ref{tbl:treatment}.
\begin{table}
\caption{For Experiment~\ref{exp:pos}, p-values for the for treatment-treatment experiment}\label{tbl:treatment}
\begin{tab}{@{}ccccc@{}}
 Data set & $s_{\simm}$ & $s_{\msf{kw}}$  & $s_{\percent}$ & $\chi^2$ \\
\midrule
$1$ & $0.634921$ 	 	& $0.821429$ 	 	& $1.000000$ 	 	& $0.158018$ \\
$2$ & $0.722222$ 	 	& $0.357143$ 	 	& $1.000000$ 	 	& $0.554021$ \\
$3$ & $0.134921$ 	 	& $0.202381$ 	 	& $1.000000$ 	 	& $0.105753$ \\
$4$ & $0.492063$ 	 	& $0.468254$ 	 	& $1.000000$ 	 	& $0.767482$ \\
$5$ & $0.103175$ 	 	& $0.281746$ 	 	& $1.000000$ 	 	& $0.235403$ \\
$6$ & $0.952381$ 	 	& $0.650794$ 	 	& $1.000000$ 	 	& $0.478123$ \\
$7$ & $0.515873$ 	 	& $0.384921$ 	 	& $1.000000$ 	 	& $0.768996$ \\
$8$ & $0.547619$ 	 	& $0.571429$ 	 	& $1.000000$ 	 	& $0.654094$ \\
$9$ & $0.492063$ 	 	& $0.829365$ 	 	& $1.000000$ 	 	& $0.097828$ \\
$10$ & $0.523810$ 	 & $0.162698$ 	 & $1.000000$ 	 & $0.24844$ \\
$11$ & $0.198413$ 	 & $0.007937^*$ 	 & $1.000000$ 	 & $1.2326e-05$ \\
$12$ & $0.515873$ 	 & $0.781746$ 	 & $1.000000$ 	 & $0.471851$ \\
$13$ & $0.222222$ 	 & $0.734127$ 	 & $1.000000$ 	 & $0.51711$ \\
$14$ & $0.563492$ 	 & $0.825397$ 	 & $1.000000$ 	 & $0.00390297$ \\
$15$ & $0.103175$ 	 & $0.396825$ 	 & $1.000000$ 	 & $0.513125$ \\
$16$ & $0.674603$ 	 & $0.146825$ 	 & $0.500000$ 	 & $0.00960702$ \\
$17$ & $0.063492$ 	 & $0.880952$ 	 & $0.500000$ 	 & $0.00017787$ \\
$18$ & $0.325397$ 	 & $0.357143$ 	 & $1.000000$ 	 & $0.239513$ \\
$19$ & $0.119048$ 	 & $0.992063$ 	 & $1.000000$ 	 & $3.01973e-10$ \\
$20$ & $0.476190$ 	 & $0.690476$ 	 & $1.000000$ 	 & $0.189638$ \\

\midrule
$\text{Number} < 5\%$ & 0  &              1       &      0       &                        5\\
\end{tab}
\end{table}
 Here too, we would expect not to find statistically significant results. The $\chi^2$-test once again shows more false-positives than the permutation tests. These numbers indicate that the $\pt(s_\simm)$ and $\pt(s_\msf{kw})$ are good indicators of statistical significance in our setting.

\clearpage

\bibliographystyle{IEEEtran}
\bibliography{blackbox}

\end{document}